\documentclass[11pt,a4paper]{article}
\usepackage{fullpage}
\usepackage{microtype}
\usepackage{amssymb,amsmath,amsthm} \usepackage{graphicx}
\usepackage{color} \usepackage{xspace}
\usepackage[numbers]{natbib} \usepackage{hyperref,color}
\usepackage{xspace} 
\usepackage[T1]{fontenc}
\graphicspath{{figures/}}
\newtheorem{theorem}{Theorem} 
\newtheorem{lemma}[theorem]{Lemma}

\newtheorem{corollary}[theorem]{Corollary}
\newtheorem{observation}[theorem]{Observation}
\newtheorem{definition}[theorem]{Definition}

\title{The Geodesic Farthest-point Voronoi Diagram in a Simple Polygon\footnote{This work was supported by the NRF grant 2011-0030044 (SRC-GAIA) funded by the government of Korea and the MSIT(Ministry of Science and ICT), Korea, under the SW Starlab support program(IITP-2017-0-00905) supervised by the IITP(Institute for Information \& communications Technology Promotion).}}

\author{Eunjin Oh\thanks{Pohang University of Science and Technology,
		Korea. Email: {\tt{\{jin9082, heekap\}@postech.ac.kr}}} \and
	Luis Barba\thanks{Department of Computer Science, ETH Z\"{u}rich, Z\"{u}rich, Switzerland. Email: 
		{\tt{luis.barba@inf.ethz.ch}}}
	\and Hee-Kap Ahn\footnotemark[2]~\thanks{Corresponding author.} }

\newcommand{\floor}[1]{\ensuremath{\lfloor#1\rfloor}}

\newcommand{\ltopregion}[1]{\ensuremath{G_\mathrm{Ltop}(#1)}}
\newcommand{\rtopregion}[1]{\ensuremath{G_\mathrm{Rtop}(#1)}}
\newcommand{\lsideregion}[1]{\ensuremath{G_\mathrm{Lside}(#1)}}
\newcommand{\rsideregion}[1]{\ensuremath{G_\mathrm{Rside}(#1)}}
\newcommand{\inregion}[1]{\ensuremath{G_\mathrm{in}(#1)}}
\newcommand{\apexedset}{\ensuremath{A}}
\newcommand{\auxenvelope}{\ensuremath{U'}}
\newcommand{\auxtriangle}{\ensuremath{\tau_U'}}
\newcommand{\upperenvelope}[1]{\ensuremath{U(#1)}}
\newcommand{\uppertriangle}[1]{\ensuremath{\tau_U(#1)}}
\newcommand{\bd}{\partial}
\newcommand{\subchain}[2]{\ensuremath{\bd P[#1,#2]}}
\newcommand{\ch}{\mathsf{CH}} \newcommand{\fvd}{\mathsf{FVD}}
\newcommand{\rfvd}{\mathsf{rFVD}} 
\newcommand{\acell}[1]{\mathsf{aCell}(#1)}
\newcommand{\aacell}{\mathsf{aCell}}
\newcommand{\avd}{\mathsf{aVD}}

\newcommand{\complexity}[1]{|#1|}
\newcommand{\cell}[1]{\mathsf{Cell}(#1)}
\newcommand{\rcell}[1]{\mathsf{rCell}(#1)}
\newcommand{\interior}[1]{\mathrm{int}(#1)}
\newcommand{\apex}[1]{\textsc {a}(#1)}
\newcommand{\definer}[1]{\textsc{d}(#1)}
\newcommand{\pseudohalf}[2]{\ensuremath{D(#1,#2)}}
\newcommand{\tri}{\triangle}

\newcommand{\ff}[2][S]{\ensuremath{\textsc n(#2)}}

\begin{document}
\date{}
\maketitle
\begin{abstract}
  Given a set of point sites in a simple polygon, the geodesic farthest-point
  Voronoi diagram partitions the polygon into cells, at most
  one cell per site, such that every point in a cell has the same
  farthest site with respect to the geodesic metric. We present an $O(n\log\log n+ m\log m)$-time algorithm
  to compute the geodesic farthest-point Voronoi diagram of $m$ point sites in a simple $n$-gon.
  This improves the previously best known algorithm by Aronov et al.~[Discrete Comput. Geom. 9(3):217-255, 1993].
  In the case that all point sites are on the boundary of the simple polygon, we can compute the geodesic farthest-point Voronoi diagram in $O((n+m) \log\log n)$ time. 
\end{abstract}

\section{Introduction}
Let $P$ be a simple polygon with $n$ vertices.  For any two points $x$
and $y$ in $P$, the \emph{geodesic path} $\pi(x,y)$ is the shortest
path contained in $P$ connecting $x$ with $y$. Note that if the
line segment connecting $x$ with $y$ is contained in $P$,
then $\pi(x,y)$ is a line segment. Otherwise, $\pi(x,y)$ is a
polygonal chain whose vertices (other than its endpoints) are reflex
vertices of $P$. 
The \emph{geodesic distance} between $x$ and $y$, denoted by $d(x,y)$,
is the sum of the Euclidean lengths of the line segments in
$\pi(x,y)$. Throughout this paper, when referring to the distance
between two points in $P$, we mean the geodesic distance between them unless otherwise stated.
We refer the reader to the survey by Mitchell~\cite{m-gspno-00} in the handbook of computational geometry for more
information on geodesic paths and distances.

Let $S$ be a set of $m$ point sites contained in $P$.  For a
point $x\in P$, a (geodesic) \emph{$S$-farthest neighbor} of $x$, is a
site $\textsc{n}(P,S,x)$ (or simply $\ff{x}$) of $S$ that maximizes
the geodesic distance to~$x$.  
To ease the description, we assume that every vertex of $P$ has a unique $S$-farthest neighbor. 
This \emph{general position} condition was
also assumed by~\citet{aronov1993furthest} and~\citet{1-center}.

The \emph{geodesic farthest-point Voronoi diagram} of $S$ in $P$ 
is a subdivision of $P$ into \textit{Voronoi cells}.
Imagine that we decompose $P$ into Voronoi cells $\mathsf{Cell}(S,s)$ (or simply $\cell{s}$) for each site $s\in S$,
where $\mathsf{Cell}(S,s)$ is the set of points in $P$
that are closer to $s$ than to any other site of $S$.
Note that some cells might be empty.
The set $P \setminus \cup_{s \in S} \cell{s}$ defines the (farthest)
\emph{Voronoi tree} of $S$ with leaves on the boundary of $P$.  Each edge of this diagram is
either a line segment or a hyperbolic arc~\cite{aronov1993furthest}.
The Voronoi tree together with the set of Voronoi cells defines the
\emph{geodesic farthest-point Voronoi diagram} of $S$ (in $P$),
denoted by $\fvd[S]$ (or simply $\fvd$ if $S$ is clear from context).
We indistinctively refer to $\fvd$ as a tree or as a set of Voronoi cells.

There are similarities between the Euclidean farthest-point
Voronoi diagram and the geodesic farthest-point Voronoi diagram
(see~\cite{aronov1993furthest} for further references).  In the
Euclidean case, a site has a nonempty Voronoi cell if and only if it
is extreme, i.e., it lies on the boundary of the convex hull of the
set of sites.  Moreover, the clockwise sequence of Voronoi cells (at
infinity) is the same as the clockwise sequence of sites along the
boundary of the convex hull.  With these properties, the Euclidean farthest-point
Voronoi diagram can be computed in linear time if the convex hull of
the sites is known~\cite{aggarwal1989linear}.
In the geodesic case, a site with nonempty Voronoi cell lies on the
boundary of the geodesic convex hull of the sites.  The order of sites along the 
boundary of the geodesic convex hull is the same as the order of their Voronoi cells
along the boundary of $P$.
However, the cell of an extreme site may be empty, roughly because the
polygon is not large enough for the cell to
appear. In addition, the complexity of the bisector between two sites
can be linear in the complexity of the polygon.  

\paragraph{Previous Work.} 
Since the early 1980s many classical geometric problems have been
studied in the geodesic setting.  The problem of computing the
\emph{geodesic diameter} of a simple $n$-gon $P$ (and its
counterpart, the \emph{geodesic center}) received a lot of attention from the
computational geometry community.
The geodesic diameter of $P$ is the largest possible geodesic distance between
  any two points in $P$, and the geodesic center of $P$ is the point
  of $P$ that minimizes the maximum geodesic distance to the points in $P$.


\citet{c-tpca-82} gave the first
algorithm for computing the geodesic diameter of $P$, which runs in
$O(n^2)$ time using linear space.  \citet{suri1989computing} reduced
the time complexity to $O(n\log n)$ without increasing the space
complexity.  Later, \citet{hershberger1993matrix} presented a fast
matrix search technique, one application of which is a linear-time
algorithm for computing the diameter of $P$.  

The first algorithm for computing the geodesic center of $P$ was given by
\citet{at-cgcsp-85}, and runs in $O(n^4\log n)$ time.  This algorithm
computes a super set of the vertices of $\fvd[V]$, where $V$ is the
set of vertices of $P$.  In 1989, \citet{pollackComputingCenter}
improved the running time to $O(n\log n)$.  In a recent
paper, \citet{1-center} settled the complexity of this problem
by presenting a $\Theta(n)$-time algorithm to compute the geodesic
center of $P$.

The problem of computing the geodesic farthest-point Voronoi diagram
is a generalization of the problems of computing the geodesic center
and the geodesic diameter of a simple polygon.  For a set $S$ of $m$
points in $P$, \citet{aronov1993furthest} presented an algorithm to
compute $\fvd[S]$ in $O(n\log n+ m\log m)$ time.  While the best known
lower bound is $\Omega(n + m \log m)$, which is a lower bound known
for computing the geodesic convex hulls of $S$, it is not known
whether or not the dependence on $n$, the complexity of $P$, is linear
in the running time.  In fact, this problem was explicitly posed by
\citet[Chapter 27]{m-gspno-00} in the Handbook of Computational
Geometry.

\paragraph{Our Result.}
In this paper, we present an $O(n \log\log n+m\log m)$-time algorithm
for computing $\fvd[S]$ for a set $S$ of $m$ points in a simple
$n$-gon.  To do this, we present an $O((n+m)\log\log n)$-time
algorithm for the simpler case that all sites are on the boundary of
the simple polygon and use it as a subprocedure for the general
algorithm.

Our result is the first improvement on the computation of geodesic
farthest-point Voronoi diagrams since 1993~\cite{aronov1993furthest}.
It partially answers the question posed by Mitchell. Moreover, our
result suggests that the computation time of Voronoi diagrams has only
almost linear dependence in the complexity of the polygon.  We believe
our results could be used as a stepping stone to solve the question
posed by \citet[Chapter 27]{m-gspno-00}.  Indeed, after the
preliminary version~\cite{oba-fpgvdpbsp-16} of this paper had been
presented, Oh and Ahn~\cite{Oh-2017} presented an
$O(n+m\log m+ m\log^2 n)$-time algorithm for this problem.  They
observed that the adjacency graph of the Voronoi cells has complexity
smaller than the complexity of the Voronoi diagram and presented an
algorithm for the geodesic farthest-point Voronoi diagram based on a
polygon-sweep paradigm, which is optimal for a moderate-sized
point-set.

\paragraph{Outline.} 
We first assume that the site set is the vertex set of the input
simple polygon.  Then we present an algorithm for computing $\fvd[S]$,
which will be extended to handle the general cases in
Section~\ref{section:General set o sites} and
Section~\ref{sec:general}.  This algorithm consists of three
steps. Each section from Section~\ref{sec:first-step} to
Section~\ref{sec:third-step} describes a step of the algorithm. In the
first step, we compute the geodesic farthest-point Voronoi diagram
restricted to the boundary of the polygon.  In the second step, we
decompose recursively the interior of the polygon into smaller cells,
not necessarily Voronoi cells, until the complexity of each cell
becomes constant.  In the third step, we explicitly compute the
geodesic farthest-point Voronoi diagram in each cell and merge them to
complete the description of the Voronoi diagram.

In the first step, we compute the restriction of $\fvd[S]$ to $\bd P$
in linear time, where $\bd P$ denotes the boundary of $P$.  A similar
approach was used by \citet{aronov1993furthest}. However, their
algorithm spends $\Theta(n\log n)$ time and uses completely different
techniques.  The main tool used to speed up the algorithm is the
matrix search technique introduced by \citet{hershberger1993matrix}
which provides a ``partial'' description of $\fvd[S]\cap \bd P$ (i.e.,
the restriction of $\fvd[S]$ to the vertices of $P$.)  To extend it to
the entire boundary of $P$, we borrow some tools used by
\citet{1-center}.  This reduces the problem to the computation of
upper envelopes of distance functions which can be completed in linear
time.

In the second step, we recursively subdivide the polygon into cells.
To subdivide a cell whose boundary consists of $t$ geodesic paths, we
construct a closed polygonal path that visits roughly $\sqrt{t}$
endpoints of the $t$ geodesic paths at a regular
interval. Intuitively, to choose these endpoints, we start at the
endpoint of a geodesic path on the boundary of the cell. Then, we walk
along the boundary, choose another endpoint after skipping $\sqrt{t}$
of them, and repeat this.  We consider the geodesic paths, each
connecting two consecutive chosen endpoints.  The union of all these
geodesic paths can be computed in time linear in the complexity of the
cell~\cite{kpairpath} and subdivides the cell into smaller simple
polygons.  By recursively applying this procedure on each resulting
cell, we guarantee that after $O(\log \log n)$ rounds the boundary of
each cell consists of a constant number of geodesic paths.
While decomposing the polygon, we also compute $\fvd[S]$ restricted to
the boundary of each cell.  However, the total complexity of $\fvd[S]$
restricted to the boundary of each cell might be $\omega(n)$ in the
worst case. To resolve this problem, we subdivide each cell further so
that the total complexity of $\fvd[S]$ restricted to the boundary of
each cell is $O(n)$ for every iteration.  Each round can be completed
in linear time, which leads to an overall running time of
$O(n \log \log n)$. After the second step, we have $O(n\log\log n)$
cells in the simple polygon and we call them the \emph{base cells}.
%

In the third step, we explicitly compute the geodesic farthest-point
Voronoi diagram in each of the base cells by applying the linear-time
algorithm of computing the abstract Voronoi diagram given by Klein and
Lingas~\cite{klein1994}. To apply the algorithm, we define a new
distance function for each site whose Voronoi cell intersects the
boundary of each cell $T$ such that the distance function is
continuous on $T$ and the total complexity of the distance functions
for all sites is $O(n)$.  We show that the abstract Voronoi diagram
restricted to a base cell $T$ is exactly the geodesic farthest-point
Voronoi diagram restricted to $T$.  After computing the geodesic
farthest-point Voronoi diagrams for every base cell, they merge them
to complete the description of the Voronoi diagram.

For the case that the sites lie on the boundary of the simple polygon,
we cannot apply the matrix searching technique directly although the
other procedures still work.  To handle this, we apply the matrix
search technique with a new distance function to compute $\fvd[S]$
restricted to the vertices of $P$.  Then we consider the general case
that the sites are allowed to lie in the interior of the simple
polygon.  We subdivide the input simple polygon in a constant number
of subpolygons, and apply the previous algorithm for sites on the
boundary to these subpolygons.  The overall strategy is similar to the
one for sites on the boundary, but there are a few nontrivial
technical issues to be addressed.

\section{Preliminaries}\label{sec:preliminaries}
For any subset $A$ of $P$, let $\bd A$ and $\interior{A}$ denote the
boundary and the interior of $A$, respectively. For any two points $x,y\in \mathbb{R}^2$, we use $xy$ to denote
the line segment connecting $x$ and $y$.  Let $P$ be a simple
$n$-gon and $S$ be a set of $m$ point sites contained in $P$.
Let $V$ be the set of the vertices of $P$.
A vertex $v$ of a simple polygon is \textit{convex} (or
\textit{reflex}) if the internal angle at $v$ with respect to the simple polygon
is less than (or at least) $\pi$.

For any two points $x$ and $y$ on $\bd P$, let $\subchain{x}{y}$ denote
the portion of $\bd P$ from $x$ to $y$ in clockwise order.  We say
that three (nonempty) disjoint sets $A_1,A_2$ and $A_3$ contained in
$\bd P$ are in \emph{clockwise order} if $A_2\subset \subchain{a}{c}$
for any $a\in A_1$ and any $c\in A_3$. To ease notation, we say that
three points $x,y,z\in \bd P$ are in clockwise order if $\{x\}, \{y\}$
and $\{z\}$ are in clockwise order.

\subsection{Ordering Lemma}

Aronov et al.~\cite{aronov1993furthest} gave the following lemma which they call  
\emph{Ordering Lemma}. We make use of this lemma to compute $\fvd$ restricted to $\bd P$.
Before introducing the lemma, we need to define the \emph{geodesic convex hull} of a set $S$ of $m$ points in $P$.
We say a subset $A$ of $P$ is \textit{geodesically convex} if $\pi(x,y)\subseteq A$ for any two points $x,y \in A$.
The geodesic convex hull of $S$ is defined to be the intersection of all geodesic convex sets containing $S$.
It can be computed in $O(n+m\log m)$ time~\cite{guibasShortestPathQueries}.~\footnote{The paper~\cite{guibasShortestPathQueries} shows that
their running time is $O(n+m\log (n+m))$. But it is $O(n+m\log m)$. To see this,
observe that it is $O(n)$ for $m=O(n/\log n)$. Also, it is $O(n+m\log m)$ for $m=\Omega(n/\log n)$.}.

\begin{lemma}[{\cite[Corollary 2.7.4]{aronov1993furthest}}]
	\label{lem:ordering}
	The order of sites along $\bd \ch$ is the
	same as the order of their Voronoi cells along $\bd P$, where $\ch$ is the geodesic convex hull of $S$
	with respect to $P$.
\end{lemma}

\subsection{Apexed Triangles}
An \emph{apexed triangle} $\triangle = (a,b,c)$ with \emph{apex}
$\apex{\tri} = a$ is an Euclidean triangle contained in $P$ with an associated
distance function $g_\triangle(x)$
such that (1) $\apex{\tri}$ is a vertex of $P$, (2) there is an edge
of $\bd P$ containing both $b$ and $c$, and (3) there is a site 
$\definer{\tri}$ of $S$, called the \emph{definer} of $\triangle$,
such that
$$g_\triangle(x) = 
\begin{cases}

  \|x-\apex{\tri}\| + d(\apex{\tri},\definer{\tri}) = d(x, \definer{\tri}) & \text{if $x\in \triangle$}\\
  -\infty&\text{if $x\notin \triangle$},
\end{cases}
$$
where $\|x-y\|$ denote the Euclidean distance between $x$ and $y$.

Intuitively, $\triangle$ bounds a constant complexity region where the
geodesic distance function from $\definer{\tri}$ can be obtained by
looking only at the distance from $\apex{\tri}$.  We call the side of
an apexed triangle $\triangle$ opposite to the apex the \textit{bottom
  side} of $\triangle$.  Note that the bottom side of $\triangle$ is
contained in an edge of $P$.

The concept of the apexed triangle was introduced by \citet{1-center}
and was a key to their linear-time algorithm to compute the geodesic
center.  After computing the $V$-farthest neighbor of each vertex in linear time~\cite{hershberger1993matrix},
they show how to compute $O(n)$ apexed triangles in $O(n)$ 
time with the following property: for each point $p\in P$, there
exists an apexed triangle $\tri$ such that $p \in \triangle$ and
$\definer{\tri} = \textsc{n}(P,V,p)$.  By the definition of the apexed triangle,
we have $d(p,\textsc{n}(P,V,p))=g_\tri(p)$.  In other words, the distance from
each point of $P$ to its $V$-farthest neighbor is encoded in one of the
distance functions associated with these apexed triangles.  

More generally, we define a set of apexed triangles whose distance functions
encode the distances from the points of $P$ to their $S$-farthest neighbors.
We say a weakly simple polygon $\gamma$ is a \emph{funnel} of a point $p\in P$
if its boundary consists of three polygonal curves $\subchain{u}{v}$, $\pi(u,p)$ and $\pi(v,p)$ for some two points $u,v\in \bd P$.

\begin{definition}\label{def:cover}
  A set of apexed triangles \emph{covers} $\fvd[S]$ if for
  any site $s\in S$, the union of all apexed triangles with definer
  $s$ is a funnel $\gamma_s$ of $s$ such that $\mathsf{Cell}(S,s)\subset \gamma_s$.	
\end{definition}

\citet{1-center} gave the following lemma. In Section~\ref{section:General set o sites} and Section~\ref{sec:general},
we show that we can extend this lemma to compute a set of apexed triangles covering $\fvd[S]$
for any set $S$ of points in a simple polygon.
\begin{lemma}[\cite{1-center}]\label{lemma:Apexed triangles}
  Given a simple $n$-gon $P$ with vertex set $V$, we can compute a set 
  of $O(n)$ apexed triangles covering $\fvd[V]$ in $O(n)$ time.
\end{lemma}

While Lemma~\ref{lemma:Apexed triangles} is not explicitly stated by
\citet{1-center}, a closer look at the proofs of Lemmas 5.2 and 5.3,
and Corollaries 6.1 and 6.2 reveals that this lemma holds.
Lemma~\ref{lemma:Apexed triangles}
states that for each vertex $v$ of $P$, the set of apexed triangles with
definer $v$ forms a connected component. In particular, the union of
their bottom sides is a connected chain along $\bd P$.  Moreover,
these apexed triangles are interior disjoint by the definition of apexed triangles.

\subsection{The Refined Geodesic Farthest-point Voronoi Diagram}

Assume that we are given a set of $O(n+m)$ apexed triangles covering $\fvd[S]$.
We consider a refined version of $\fvd[S]$ which we call the  
\textit{refined geodesic farthest-point Voronoi diagram} defined as
follows: for each site $s\in S$, the Voronoi cell $\cell{s}$ of $\fvd[S]$
is subdivided by the apexed triangles with definer $s$. That is, for
each apexed triangle $\triangle$ with definer $s$, we define a
\emph{refined cell} $\rcell{\triangle} = \triangle \cap
\cell{s}$, where $\tri$ is the union of $\interior{\tri}$ and its bottom side (excluding the corners of $\tri$).  
Since any two apexed triangles $\triangle_1$ and
$\triangle_2$ with the same definer are interior disjoint,
$\rcell{\triangle_1}$ and $\rcell{\triangle_2}$ are also interior
disjoint.  We denote the set $P \setminus \cup_{\triangle}
\rcell{\triangle}$ by $\rfvd$.  Then, $\rfvd$ forms a tree
  consisting of arcs and vertices.  Notice that each arc of $\rfvd$
is a part of either the bisector of two sites or a side of
an apexed triangle.  Since we assume that the number of the apexed triangles is
$O(n+m)$, the complexity of $\rfvd$ is still $O(n+m)$.
(Lemma~2.8.3 in~\cite{aronov1993furthest} shows that the complexity of $\fvd$ is $O(n+m)$.)

\begin{lemma}
  \label{lem:ray_in_cell}
  For an apexed triangle $\tri$ and a point $x$ in $\rcell{\triangle}$,
  the line segment $xy$ is
  contained in $\rcell{\triangle}$, where $y$ is the point on the
  bottom side of $\triangle$ hit by the ray from $\apex{\triangle}$
  towards $x$.
\end{lemma}
\begin{proof}
  Let $p$ be a point on $xy$.  We
  have $d(\definer{\triangle},p) = d(\definer{\triangle},x)+d(x,p)$.
  On the other hand, $d(s,p) \leq d(s,x)+d(x,p)$ by the triangle
  inequality for any site $s$.  Since $d(s,x) <
  d(\definer{\triangle},x)$ for any site $s$ other than
  $\definer{\triangle}$, we have $d(s,p) < d(\definer{\triangle},p)$,
  which implies that $p$ lies in $\rcell{\triangle}$.
\end{proof}
\begin{corollary}\label{cor:ray-in-non-refined-cell}
  For any site $s\in S$ and any point $x\in \cell{s}$, the line
  segment $xy$ is contained in $\cell{s}$, where $y$ is the point on
  $\bd P$ hit by the ray from the neighbor of $x$ along $\pi(s,x)$
  towards $x$.
\end{corollary}

Throughout this paper, we use $\complexity{C}$ to denote the
number of edges of $C$ for a simple polygon $C\subseteq P$.  
For a curve $\gamma$, we use $\complexity{\rfvd\cap\gamma}$ 
to denote the number of the refined cells intersecting
$\gamma$.
For ease of description, we abuse the term ray slightly such that
  the ray from $x\in P$ in a direction denotes the line segment $xy$ of the
  halfline from $x$ in the direction, where $y$ is the first point of $\bd P$
  encountered along the halfline from $x$.

From Section~\ref{sec:first-step} to Section~\ref{sec:third-step}, we will make the assumption that $S$ is the set
of the vertices of $P$. This assumption is general enough as we show
how to extend the result to the case when $S$ is an arbitrary set of
sites contained in $\bd P$ (Section~\ref{section:General set o sites}) and in $P$ (Section~\ref{sec:general}).
The algorithm for computing $\fvd[V]$ consists of three steps. 
Each section from Section~\ref{sec:first-step} to Section~\ref{sec:third-step} describes each step.

\section{Computing \texorpdfstring{$\fvd$}{FVD} Restricted to \texorpdfstring{$\bd P$}{bd P}}
\label{sec:first-step}
Using the algorithm in \cite{1-center}, we compute a set $\mathcal{A}$ of $O(n)$ apexed
triangles covering $\fvd[S]=\fvd[V]$ in $O(n)$ time.
Recall that the apexed triangles with the same definer are
interior disjoint and have their bottom sides on $\bd P$ whose union
forms a connected chain along $\bd P$. Thus, such apexed triangles can be sorted along $\bd P$
with respect to their bottom sides.

\begin{lemma}
  Given a set $\tau_s$ of all apexed triangles of $\mathcal{A}$ with definer~$s$
  for a site $s$ of $S$, we can sort the apexed
  triangles in $\tau_s$ along $\bd P$ with respect to their bottom
  sides in $O(|\tau_s|)$ time.
\end{lemma}
\begin{proof}
  The bottom side of an apexed triangle is contained in an edge of 
  $\bd P$, and the other two sides are chords of $P$ (possibly flush
  with $\bd P$).  Assume that these chords are oriented from its apex
  to its bottom side.  Using a hash-table storing the chords of the
  apexed triangles in $\tau_s$, we can link each of these chords to
  its neighboring triangles (and distinguish between left and right
  neighbors).  In this way, we can retrieve a linked list with all the
  triangles in $\tau_s$ in sorted order along $\bd P$ in $O(|\tau_s|)$ time.
\end{proof}

\subsection{Computing the \texorpdfstring{$S$-Farthest}{S-Farthest} Neighbors of the Sites}
The following lemma was used by \citet{1-center} and is based on the
matrix search technique proposed by~\citet{hershberger1993matrix}.

\begin{lemma}[\cite{hershberger1993matrix}]\label{lemma:Matrix lemma}
	We can compute the $S$-farthest neighbor of each vertex of $P$ in
	$O(n)$ time.
\end{lemma}

Using Lemma~\ref{lemma:Matrix lemma}, we mark the vertices of $P$ that
are $S$-farthest neighbors of at least one vertex of $P$.  Let $M$
denote the set of marked vertices of $P$.  Note that $M$ consists of the vertices of $P$ each of whose
Voronoi region contains at least one vertex of~$P$.

We call an edge $uv$ a \emph{transition edge} if $\ff{u} \neq \ff{v}$.
Let $uv$ be a transition edge of $P$ such that $u$ is the clockwise
neighbor of $v$ along $\bd P$.  Recall that we already have $\ff{u}$
and $\ff{v}$ and note that $v,u, \ff{v}, \ff{u}$
are in clockwise order by Lemma~\ref{lem:ordering}.  Let $w$ be a vertex of $P$ such that $\ff{v},
w, \ff{u}$ are in clockwise order.  By Lemma~\ref{lem:ordering}, if
there is a point $x$ on $\bd P$ whose $S$-farthest neighbor is $w$, then
$x$ must lie on $uv$.  In other words, the Voronoi cell $\cell{w}$
restricted to $\bd P$ is contained in~$uv$ and hence, there is no
vertex $v'$ of $P$ such that $\ff{v'} = w$.
For a nontransition edge $uv$ such that
$\ff{u}=\ff{v}$, we know that  $\ff{u}=\ff{x}=\ff{v}$ for any
point $x\in uv$. 
Therefore, to complete the description of $\fvd$
restricted to $\bd P$, it suffices to compute $\fvd[S]$ restricted to the 
transition edges.  

\subsection{Computing \texorpdfstring{$\rfvd$}{rFVD} Restricted to a Transition Edge}
\label{sec:transition_edge_fvd}
Let $uv$ be a transition edge of $P$ such that $u$ is the clockwise neighbor
of $v$.  Without loss of generality, we
assume that $uv$ is horizontal and $u$ lies to the left of $v$.
Recall that if there is a site $s$ with $\cell{s}\cap uv \neq \emptyset$, then
$s$ lies in $\subchain{\ff{v}}{\ff{u}}$.
 Thus, to compute $\rfvd\cap uv$, it
is sufficient to consider the apexed triangles of $\mathcal{A}$ with definers in
$\subchain{\ff{v}}{\ff{u}}$.  Let $\apexedset$ be the set of apexed
triangles of $\mathcal{A}$ with definers in $\subchain{\ff{v}}{\ff{u}}$.

We give a procedure to compute $\rfvd \cap uv$ in
$O(|\apexedset|)$ time using the sorted lists of the apexed triangles
with definers in $\subchain{\ff{v}}{\ff{u}}$.  Once it is done for all
transition edges, we obtain the refined geodesic farthest-point Voronoi
diagram restricted to $\bd P$ in $O(n)$ time.  Let $s_1=
\ff{u},s_2,\ldots,s_{\ell}=\ff{v}$ be the sites lying on
$\subchain{\ff{v}}{\ff{u}}$ in counterclockwise order along $\bd P$.
See Figure~\ref{fig:boundary}.

\begin{figure}
	\begin{center}
		\includegraphics[width=0.4\textwidth]{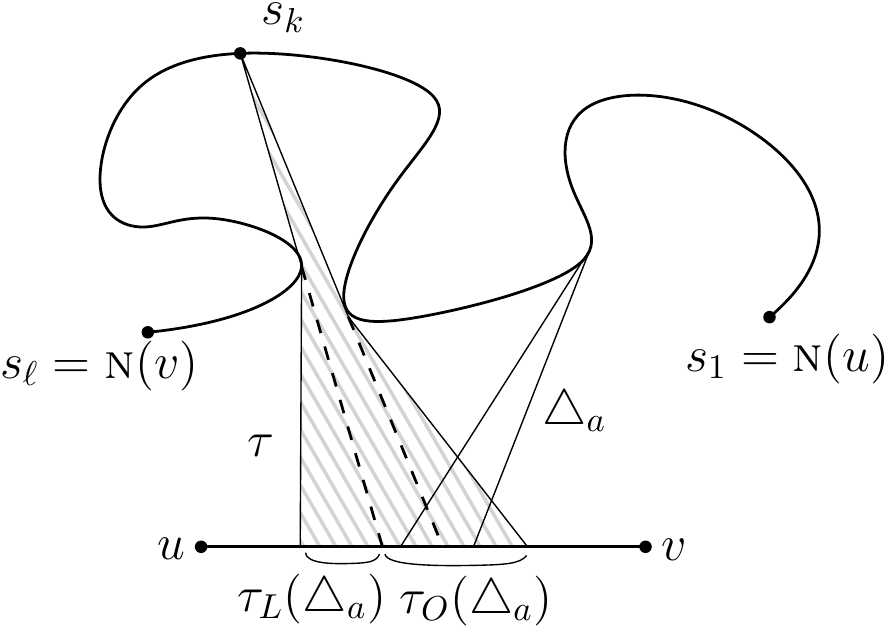}
		\caption {\small $\tau$ is the list of the three apexed triangles with definer $s_k$ sorted along $uv$.
			$\tri_a$ overlaps with the two right apexed triangles of $\tau$ while it does not overlap with the leftmost one.
			$\tau_R(\tri_a)$ is empty.}
		\label{fig:boundary}
	\end{center}
\end{figure}

\subsubsection{Upper Envelopes and \texorpdfstring{$\rfvd$}{rFVD}}
Consider any $t$ functions $f_1,\ldots,f_t$ with $f_j : D \rightarrow
\mathbb{R}\cup\{-\infty\}$ for $1 \leq j \leq t$, where $D$ is a subset of 
$\mathbb{R}^2$.  We define the \textit{upper envelope} of $f_i$'s as the
piecewise maximum of $f_i$'s. Moreover, we say that a function
$f_j$ \textit{appears} on the upper envelope if $f_j(x)\geq f_i(x)$ and $f_j(x)\in\mathbb{R}$ at some point $x\in D$
for any other functions $f_i$.

Each apexed triangle $\triangle\in A$ has a distance function $g_\triangle$ such that
$g_\triangle(x) = -\infty$ for a point $x\notin\triangle$ and
$g_\triangle(x)=d(\definer{\triangle},x)$ for a point $x\in\triangle$.
In this subsection, we restrict the domain of the distance functions
to $uv$.   By definition, the upper envelope of $g_\triangle$ for all apexed triangles
$\triangle \in \apexedset$ on $uv$ coincides with $\rfvd\cap uv$ in
its projection on $uv$.  We consider the sites one by
one from $s_1$ to $s_\ell$ in order and compute the upper envelope of $g_\triangle$ for all
apexed triangles $\triangle \in \apexedset$ on $uv$.

While the upper envelope of $g_\triangle$ for all apexed triangles
$\triangle \in \apexedset$ is continuous, the upper envelope
of $g_{\triangle'}$ of all apexed triangles $\triangle'$ with definers from
$s_1$ up to $s_k$ (we simply say the upper envelope for sites
from $s_1$ to $s_k$) might be discontinuous at some point on $uv$ for $1\leq k<\ell$. 
We let $\upperenvelope{s_k}$ be the leftmost connected component of the upper envelope for sites from $s_1$ to $s_k$ along $uv$.
By definition, $\upperenvelope{s_\ell}=\upperenvelope{\ff{v}}$ is the upper
envelope of the distance functions of all apexed triangles in
$\apexedset$.  Note that $\rcell{\triangle} \cap uv=\emptyset$ for some
apexed triangle $\triangle\in\apexedset$.  Thus the distance function
of some apexed triangle might not appear on $\upperenvelope{s_k}$.
Let $\uppertriangle{s_k}$ be the list of the apexed
triangles sorted in the order of their distance functions appearing on $U(s_k)$.
If $\definer{\triangle_i}\neq\definer{\triangle_{i+1}}$  for any two consecutive apexed triangles $\triangle_i$
and $\triangle_{i+1}$ of $\uppertriangle{s_k}$, the
bisector of $\definer{\triangle_i}$ and $\definer{\triangle_{i+1}}$
crosses the intersection of the bottom sides of $\triangle_i$ and
$\triangle_{i+1}$.

\subsubsection{Computing the Upper Envelope \texorpdfstring{$\upperenvelope{s_\ell}$}{U(sl)}}
Suppose that we have $\upperenvelope{s_{k-1}}$ and
$\uppertriangle{s_{k-1}}$ for some index $2\leq k\leq \ell$.  We
 compute $\upperenvelope{s_k}$ and $\uppertriangle{s_k}$ from
$\upperenvelope{s_{k-1}}$ and $\uppertriangle{s_{k-1}}$ as follows.
We use two auxiliary lists $\auxenvelope$ and
$\auxtriangle$ which are initially set to $\upperenvelope{s_{k-1}}$
and $\uppertriangle{s_{k-1}}$.  We update $\auxenvelope$ and
$\auxtriangle$ until they finally become $\upperenvelope{s_k}$ and
$\uppertriangle{s_k}$, respectively.
For simplicity, we use $U=\upperenvelope{s_k}$, $\tau_U=\uppertriangle{s_k}$ and $s=s_k$.

Let $\tau$ be the list of the apexed triangles of $A$ with definer $s$
sorted along $\bd P$ with respect to their bottom sides.  For any
apexed triangle $\tri$, we denote the list of the apexed triangles in
$\tau$ overlapping with $\triangle$ in their bottom sides by
$\tau_O(\triangle)$.  Also, we denote the lists of the apexed
triangles in $\tau\setminus \tau_O(\triangle)$ lying left to
$\triangle$ and lying right to $\triangle$ along $uv$ with respect to their
bottom sides by $\tau_L(\triangle)$ and $\tau_R(\triangle)$,
respectively.
See Figure~\ref{fig:boundary}.

Let $\triangle_a$ denote the the rightmost apexed
triangle of $\auxtriangle$ along $uv$.  With respect to $\triangle_a$, we
partition $\tau$ into three disjoint sublists $\tau_L(\triangle_a)$,
$\tau_O(\triangle_a)$ and $\tau_R(\triangle_a)$.
We can compute these sublists in $O(|\tau|)$ time.

\paragraph{Case 1 : Some apexed triangles in \texorpdfstring{$\tau$}{tau} overlap with
\texorpdfstring{$\tri_a$}{triangle a} (i.e. \texorpdfstring{$\tau_O(\tri_a)\neq\emptyset$}{}).}  Let $\triangle$
be the leftmost apexed triangle in $\tau_O(\triangle_a)$ along $uv$.  We compare
the distance functions $g_\triangle$ and $g_{\triangle_a}$ on
$\triangle_a\cap\triangle\cap uv$.  That is, we compare $d(x,s)$ and
$d(x,\definer{\triangle_a})$ for $x\in \triangle_a\cap\triangle\cap
uv$.

(1) If there is a point on $\triangle_a \cap
\triangle\cap uv$ that is equidistant from $s$ and $\definer{\triangle_a}$,
then $g_\triangle$ appears on $U$.
Moreover, the distance functions of the apexed triangles in
$\tau_O(\tri_a)\cup\tau_R(\triangle_a)$ also appear on $U$, but  
no distance function of the apexed triangles in $\tau_L(\triangle_a)$ appears on
$U$ by Lemma~\ref{lem:ordering}.  Thus we append 
the triangles in $\tau_O(\triangle_a)\cup\tau_R(\triangle_a)$.  We
also update $\auxenvelope$ accordingly.  Then, $\auxtriangle$ and
$\auxenvelope$ are $\tau_U$ and $U$,
respectively.

(2) If $d(x,\definer{\triangle_a}) > d(x,s )$ for all points $x\in
\triangle_a\cap\triangle\cap uv$, then $\triangle$ and its distance
function do not appear on $\tau_U$ and
$U$, respectively, by Lemma~\ref{lem:ordering}.
Thus we do nothing and scan the apexed triangles in
$\tau_O(\triangle_a)\cup\tau_R(\triangle_a)$, except $\triangle$, from
left to right along $uv$ until we find an apexed triangle $\triangle'$ such that
there is a point on $\triangle_a\cap\triangle'\cap uv$ which is
equidistant from $\definer{\triangle_a}$ and $s $.  Then we apply the
procedure in (1) with $\triangle'$ instead of $\triangle$.  If there
is no such apexed triangle, then 
$\auxtriangle$ and
$\auxenvelope$ are $\tau_U$ and $U$,
respectively.

(3) Otherwise, we have
$d(x,s ) > d(x,\definer{\triangle_a})$ for all points $x\in
\triangle_a\cap\triangle\cap uv$.
Then the distance function of
$\triangle_a$ does not appear on $U$.  Thus, we
remove $\triangle_a$ and its distance function from $\auxtriangle$ and
$\auxenvelope$, respectively.  We consider the apexed triangles in
$\tau_L(\triangle_a)$ from right to left along $uv$.  For an apexed triangle
$\triangle' \in \tau_L(\triangle_a)$, we do the following.  Since
$\auxtriangle$ is updated, we update $\triangle_a$ to the rightmost element
of $\auxtriangle$ along $uv$.  We check whether $d(x,s )\geq
d(x,\definer{\triangle_a})$ for all points $x\in
\triangle_a\cap\triangle'\cap uv$ if $\tri'$ overlaps with $\tri_a$.  If
so, we remove $\tri_a$ from $\auxtriangle$ and update $\triangle_a$ again.
We do this until we find an apexed triangle $\triangle' \in
\tau_L(\triangle_a)$ such that this test fails.
Then, there is a point on $\triangle'\cap\triangle_a\cap uv$ which is
equidistant from $\definer{\triangle_a}$ and $s $.  After we reach
such an apexed triangle $\triangle'$, we apply the procedure in (1)
with $\triangle'$ instead of $\triangle$.

\paragraph{Case 2 : No apexed triangle in \texorpdfstring{$\tau$}{tau} overlaps with
\texorpdfstring{$\tri_a$}{triangle a} (i.e.  \texorpdfstring{$\tau_O(\tri_a)=\emptyset$}{}).}  We cannot compare
the distance function of any apexed triangle in $\tau$ with the
distance function of $\triangle_a$ directly, so we need a different
method to handle this.
There are two possible subcases: either $\tau_L(\triangle_a)=\emptyset$ or
$\tau_R(\triangle_a)=\emptyset$.  Note that these are the only possible
subcases since the union of the apexed triangles with the same definer
is connected.  For the former subcase, the upper envelope of sites
from $s_1$ to $s $ is discontinuous at the right endpoint of the
bottom side of $\triangle_a$ along $uv$.  Thus $g_\triangle$ does not appear on
$U$ for any apexed triangle $\triangle \in \tau$.
Thus $\auxtriangle$ and
$\auxenvelope$ are $\tau_U$ and $U$,
respectively.
For the latter subcase, at most one of $s $ and
$\definer{\triangle_a}$ has a Voronoi cell in $\fvd[S]$ by Lemma~\ref{lem:ordering}.  
We can find a site ($s $ or $\definer{\triangle_a}$) which does
not have a Voronoi cell in $\fvd[S]$ in constant time once
we maintain some geodesic paths.
We describe this procedure at the end of this subsection.

If $s$ does not have a Voronoi cell in $\fvd[S]$, 
then $\auxtriangle$ and $\auxenvelope$ are $\tau_U$ and $U$,
respectively.  If $\definer{\triangle_a}$ does
not have a Voronoi cell in $\fvd[S]$, we remove all apexed
triangles with definer $\definer{\triangle_a}$ from $\auxtriangle$ and
their distance functions from $\auxenvelope$.  Since such apexed
triangles lie at the end of $\auxtriangle$ consecutively, this removal process takes the
time linear in the number of the apexed triangles.
We repeat this until the rightmost element of $\tau$ and the rightmost element of
$\auxtriangle$ overlap in their bottom sides along $uv$.
When the two elements overlap, we apply the procedure of Case 1.

In total, the running time for
computing $U(s_\ell)$ is $O(|A|)$ since each apexed triangle in $\apexedset$ is
removed from $\tau_U'$ at most once. Thus, we can compute $\rfvd\cap \bd P$ is $O(n)$ time in total.

\paragraph{Maintaining Geodesic Paths for Subcase of Case 2 :  \texorpdfstring{$\tau_O(\triangle_a)=\emptyset$ and $\tau_R(\triangle_a)=\emptyset$}{}.}
We maintain $\pi(s,x)$ and its
geodesic distance during the whole procedure (for all cases), where
$s$ is the site we consider and $x$ is the projection of
the rightmost \emph{breakpoint} of $U'$ onto $uv$.
That is, $x$ is the projection of the common endpoint of the two rightmost pieces of $U'$ onto $uv$.
Recall that $s$ changes from $s_1$ to $s_\ell$.
By definition, $x$ lies in the bottom side of the rightmost apexed triangle $\tri_a$ of $\tau_U'$.
Thus we can evaluate $d(\definer{\tri_a},x)$ in constant time. 
Note that the two points $s, x$ and the apexed triangle
$\tri_a$ change during the procedure.  Whenever they
change, we update $\pi(s,x)$ and its geodesic distance 
using the previous geodesic path. One of $s$ and $\definer{\tri_a}$ does not have a Voronoi cell in $\fvd[S]$ in this subcase.
But it is possible that neither $s$ nor $\definer{\tri_a}$ has a Voronoi cell in $\fvd[S]$. 
We can decide which site does not have a Voronoi cell in $\fvd[S]$ in constant time:
if $d(\definer{\triangle_a},x)>d(s ,x)$, then $s$ does not have a Voronoi cell.
Otherwise, $\definer{\tri_a}$ does not have a Voronoi cell.

We will show that the update of the geodesic path takes $O(n)$ time in total for all transition edges. 
Let $H_{uv}$ denote the region bounded by $uv$, $\pi(v,\ff{u})$,
$\subchain{\ff{v}}{\ff{u}}$ and $\pi(\ff{v},u)$. The sum of the complexities $|H_e|$ 
of $H_e$ for all transition edges $e$ is $O(n)$ and they can be
computed in $O(n)$ time~(Corollary 3.8~\cite{1-center}).  
Moreover, $|A|$ is $O(|H_{uv}|)$~(Lemma 5.2~\cite{1-center}).  
The total complexity of the shortest path trees rooted at $u$ and $v$ in $H_{uv}$
is $O(|H_{uv}|)$, and therefore we can compute them in $O(|H_{uv}|)$ time~\cite{shortest-path-tree}.  We compute them only one for each transition edge during
the whole procedure.

The edges in $\pi(s,x)$, except the edge adjacent to $x$, are also edges of the shortest path
trees, and thus we can update them by traversing the shortest path trees in time linear in the amount
of the changes on $\pi(s,x)$.
Therefore, the following lemma implies that maintaining $\pi(s,x)$ and its
length takes $O(|H_{uv}|)$ time for each transition edge $uv$.

\begin{lemma}
	The amount of the changes on $\pi(s,x)$ is $O(|H_{uv}|)$ during the whole procedure for $uv$.
\end{lemma}
\begin{proof}
We claim that each edge of the shortest path
trees is removed from $\pi(s,x)$ at most $O(1)$ times during the whole procedure for $uv$. 
Assume that we already have $\pi(s ,x)$ and we are to compute $\pi(s ',x')$.
There are three different cases: 
(1) $x'$ lies to the left of $x$ ($\tri_a$ is removed) along $uv$, 
(2) $x'$ lies to the right of $x$ (a new apexed triangle is inserted to $\tau_U'$) along $uv$,
and (3) we consider a new site (that is, $s'=s_{k+1}$ and $x'=x$.)

\begin{figure}
  \begin{center}
    \includegraphics[width=0.8\textwidth]{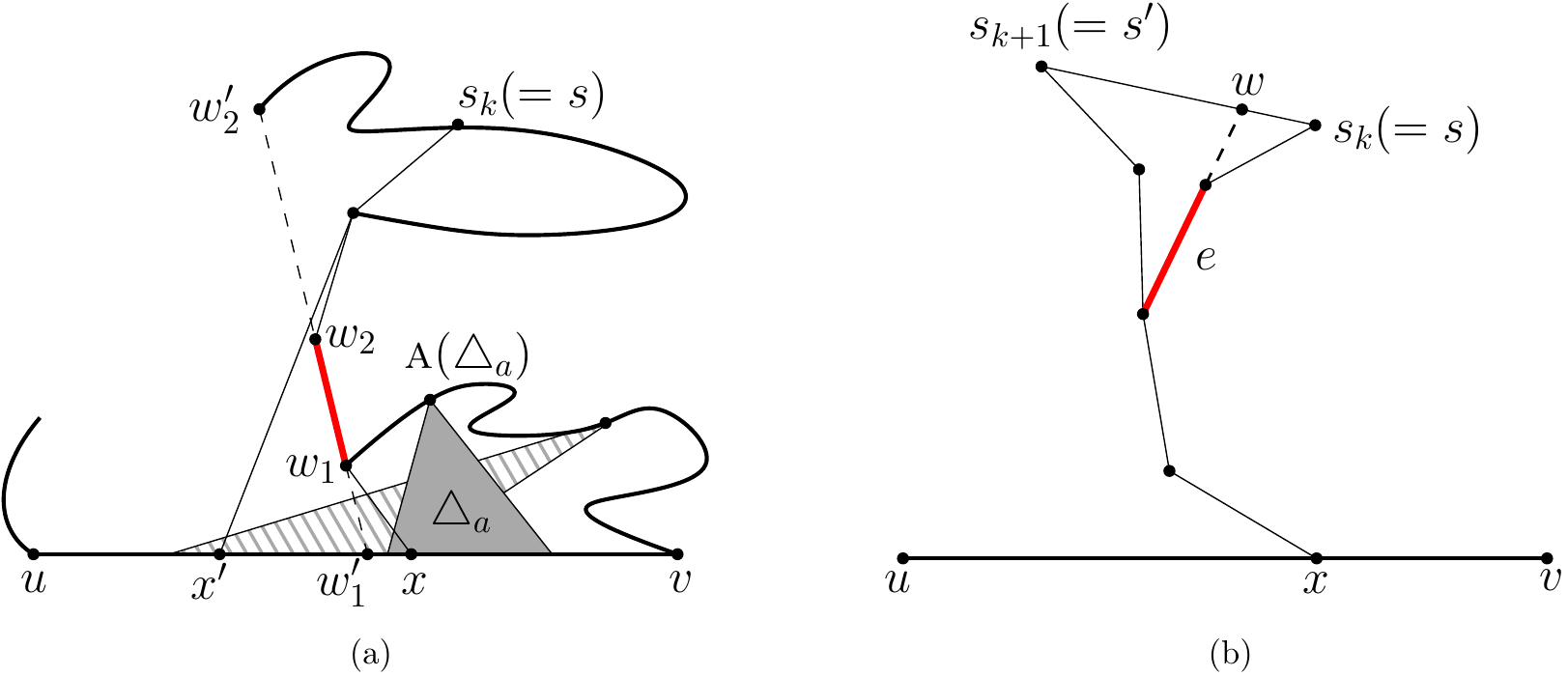}
    \caption {\small (a) When $\triangle_a$ is removed from
      $\auxtriangle$, we remove three edges from
      $\pi(s ,x)$ to obtain $\pi(s ,x')$. (b) The edge $e$
      appears on $\pi(s,x)$ for some $x \in uv$ only if $s \in
      \subchain{w}{v}$.}
    \label{fig:handle_disjoint_triangle}
  \end{center}
\end{figure}

For the first and the second cases, it is possible that we remove
more than one edge from $\pi(s ,x)$. We  prove the claim for the first case only.
The claim for the second case can be proved analogously.
See Figure~\ref{fig:handle_disjoint_triangle}(a).  Let $w_1w_2$ be an edge in $\pi(s ,x)$ which is not adjacent to
$x$ and is not in $\pi(s ,x')$ with $d(w_1,x)<d(w_2,x)$.
Let $w_1'$ and $w_2'$ be the points on $\bd P$ hit by the rays
from $w_1w_2$ towards $w_1$ and towards $w_2$, respectively. 
The right endpoint of the bottom side of $\tri_a$ lies to the right of $w_1'$ since $\pi(s ,x)$ contains $w_1w_2$.
Moreover, $s $ lies in $\subchain{w_2'}{w_2}$. Thus, $\apex{\triangle_a}$ lies in $\subchain{w_1}{v}$.

There are two possible subcases: $\definer{\triangle_a}$ is in $\subchain{w_2'}{w_2}$, or in $\subchain{w_1}{v}$.
There is at most one site $s'$ in $\subchain{w_2'}{w_2}$ such that an apexed triangle with definer $s'$ has its
apex in $\subchain{w_1}{v}$ by the construction of the set of apexed triangles in~\cite{1-center}.
(In this case, the apex lies in $\pi(w_1,v)$.) When $w_1w_2$ is deleted, all such apexed triangles
are also deleted from $\tau_U'$.
After $w_1w_2$ is deleted, no apexed triangle with definer in $\subchain{w_2'}{w_2}$ and with 
apex in $\subchain{w_1}{v}$ is inserted to $\tau_u'$ again. 
Therefore, the number of deletions of $w_1w_2$ due to the first subcase is only one.
For the second subcase, notice that once $\triangle_a$ is removed from $\auxtriangle$, no apexed
triangle with definer in $\subchain{w_1}{v}$ is added to $\auxtriangle$
again. Thus, the number of deletions of $w_1w_2$ due to the second subcase is also one. 

For the third case, $s '=s_{k+1}$ lies after $s $ from $s_1$ in counterclockwise order along $\bd P$.  It
occurs when we finish the procedure for handling $s $.  After we
consider the site $s '$, we do not consider any site from $s_1$ to
$s $ again.  Consider an edge $e$ removed from $\pi(s ,x)$ due to
this case.  Let $w$ be the point on $ss'$ hit by the 
extension of $e$.  See Figure~\ref{fig:handle_disjoint_triangle}(b).
If $\pi(s,x)$ contains $e$ for some $s \in
\subchain{\ff{v}}{\ff{u}}$ and some $x \in uv$, we have $s \in
\subchain{w}{v}$.  This means that once $e$ is removed due to the last
case, $e$ does not appear on the geodesic path $\pi(s,x)$ again in
the remaining procedure.  Thus, the number of deletions of each edge 
due to the last case is also one.
\end{proof}

Therefore, we can complete the first step in $O(n)$ time and we have the following theorem.
\begin{theorem}\label{thm:restrict-boundary}
  The geodesic farthest-point Voronoi diagram of the vertices of a
  simple $n$-gon $P$ restricted to the boundary of $P$ can be computed
  in $O(n)$ time.
\end{theorem}

\section{Decomposing the Polygon into Smaller Cells}
\label{sec:second-step}
Now we have $\rfvd \cap \bd P$ of size $O(n)$.  We add
the points in $\rfvd \cap \bd P$ (degree-$1$ vertices of $\rfvd$) to the vertex set of $P$, and apply
the algorithm to compute the apexed triangles 
with respect to the vertex set of $P$ again~\cite{1-center}.
There is no transition edge because no additional vertex has a Voronoi cell and every degree-1 Voronoi vertex
is a vertex of $P$.  Thus
the bottom sides of all apexed triangles are interior disjoint.  Moreover, we have the set $\mathcal{A}$ of the apexed triangles sorted along
$\bd P$ with respect to their bottom sides.

\begin{definition}
A simple polygon $P'\subseteq P$ is called a \textit{$t$-path-cell} for some $t\in\mathbb{N}$
if it is geodesically convex and all its vertices are on $\bd P$ among which at most $t$ are convex.
\end{definition}

In this section, we subdivide $P$ into
$t$-path-cells recursively for some $t\in\mathbb{N}$ until each cell
becomes a \emph{base cell}.  There are three types of base cells.  The first
type is a quadrilateral crossed by exactly one arc of $\rfvd$ through
two opposite sides, which we call an \textit{arc-quadrilateral}.  The
second type is a $3$-path-cell. Note that a $3$-path-cell is a
pseudo-triangle.  The third type is a region of $P$ whose boundary
consists of one convex chain and one geodesic path (concave curve), which we call a
\textit{lune-cell}.  Note that a convex polygon is a lune-cell whose
concave chain is just a vertex of the polygon.

Let $\{t_k\}$ be the sequence such that $t_1 = n$ and
$t_k=\floor{\sqrt{t_{k-1}}}+1$.  Initially, $P$ itself is a $t_1$-path-cell.
Assume that the $k$th iteration is completed. We show how to subdivide
each $t_k$-path-cell with $t_k>3$ into $t_{k+1}$-path-cells and base
cells in the $(k+1)$th iteration in Section~\ref{sec:subdivision}.
A base cell is not subdivided further.
While subdividing a cell into a number of smaller cells recursively, we compute the refined
geodesic farthest-point Voronoi diagram restricted to the boundary of
each smaller cell $C$ (of any kind) in $O(|C|+|\rfvd\cap\bd C|)$ time.
In Section~\ref{sec:third-step},
we show how to compute the refined geodesic farthest-point Voronoi
diagram restricted to a base cell $T$ in $O(|T|+\complexity{
\rfvd\cap\bd T})$ time once we have $\rfvd\cap\bd T$.
Once we compute $\rfvd$ restricted to every base cell, we obtain $\rfvd$.
%
%

\subsection{Subdividing a \texorpdfstring{$t$}{t}-path-cell into Smaller Cells}
\label{sec:subdivision}
In this subsection, we are to subdivide each $t_k$-path-cell into
$t_{k+1}$-path-cells and base cells.  If a $t_k$-path-cell is a lune-cell
or $t_k$ is at most three, the cell is a already base cell and we do
not subdivide it further.  Otherwise, we subdivide it using the algorithm described 
in this subsection.

The subdivision consists of three phases.  In Phase 1, we
subdivide each $t_k$-path-cell into $t_{k+1}$-path-cells by a curve
connecting at most $t_{k+1}$ convex vertices of the $t_k$-path-cell.  In Phase 2,
we subdivide each $t_{k+1}$-path-cell further along the
arcs of $\rfvd$ crossing the cell if there are such arcs.  In Phase 3,
we subdivide the cells that are created in Phase 2 and have vertices in $\interior{P}$
into $t_{k+1}$-path-cells and lune-cells.  

\subsubsection{Phase 1. Subdivision by a Curve Connecting at Most \texorpdfstring{$t_{k+1}$}{t k+1} Vertices}

\begin{figure}
  \begin{center}
    \includegraphics[width=0.9\textwidth]{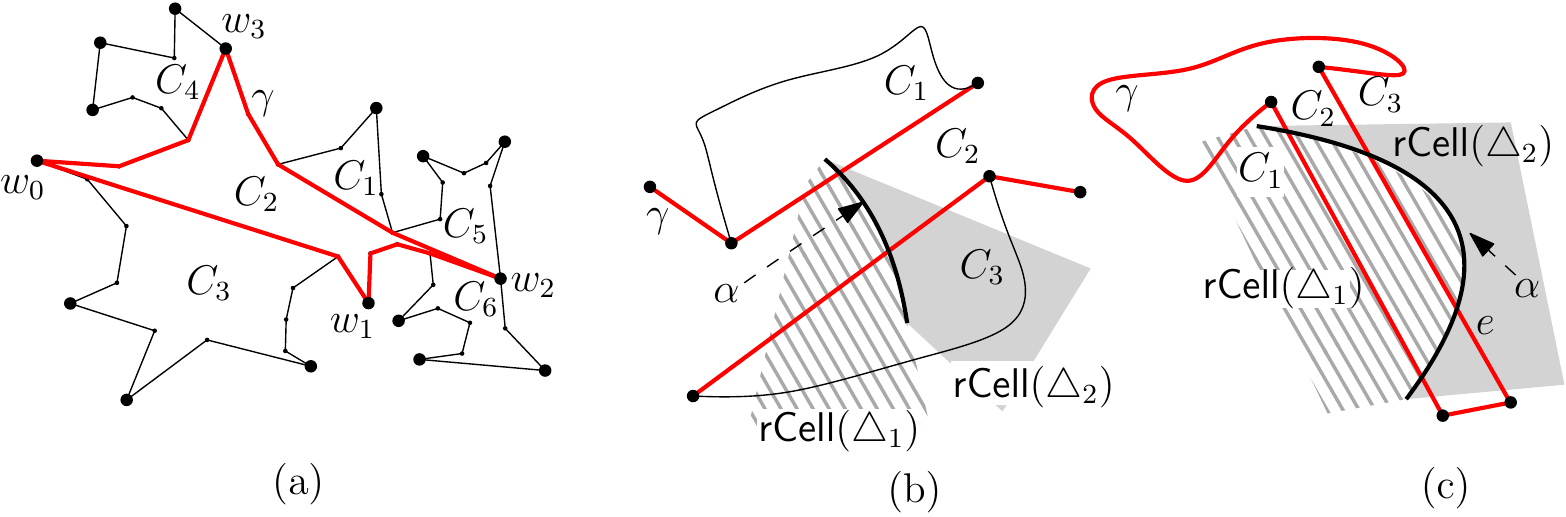}
    \caption {\small (a) A 16-path-cell. All convex vertices of the cell are marked with black
      disks. The region is subdivided into six 5-path-cells by the curve consisting of
      $\pi(w_0,w_1),\pi(w_1,w_2),\pi(w_2,w_3)$ and $\pi(w_3,w_0)$.
      (b) The arc $\alpha$ of $\rfvd$ intersects $C_1, C_2, C_3$ and
      crosses $C_2$.  (c) The arc $\alpha$ of $\rfvd$ intersects $C_1,
      C_2, C_3$ and crosses $C_2$.  Note that $\alpha$ does not cross
      $C_3$.}
    \label{fig:tpathcell}
  \end{center}
\end{figure}

Let $C$ be a $t_k$-path-cell computed in the $k$th iteration.  Recall
that $C$ is a simple polygon which has at most $t_k$ convex vertices.
Let $\beta$ be the largest integer satisfying that
$\beta\floor{\sqrt{t_k}}$ is less than the number of the convex
vertices of $C$.  Then we have $\beta \leq
\floor{\sqrt{t_k}}+1=t_{k+1}$.

We choose $\beta+1$ vertices $w_0,w_1,\ldots,w_\beta$ from the convex
vertices of $C$ at a regular interval as follows.
We choose an arbitrary convex vertex of $C$ and denote it by $w_0$.
Then we choose the $j\floor{\sqrt{t_k}}$th convex vertex of $C$ from
$w_0$ in clockwise order and denote it by $w_j$ for all
$j=1,\ldots,\beta$. We set $w_{\beta+1}=w_0$.
Then we construct the closed curve $\gamma_C$ (or simply $\gamma$ when
$C$ is clear from context) consisting of the geodesic paths
$\pi(w_0,w_1),\pi(w_1,w_2),\ldots,(w_\beta,w_0)$.
See Figure~\ref{fig:tpathcell}(a).  In other words, the closed curve
$\gamma_C$ is the boundary of the geodesic convex hull of
$w_0,\ldots,w_\beta$.  Note that $\gamma$ 
does not cross itself.  Moreover, $\gamma$ is contained in $C$
since $C$ is geodesically convex.

We compute $\gamma$ in time linear in the number of edges of $C$ using the 
algorithm in~\cite{kpairpath}.  This algorithm 
takes $k$ source-destination pairs as input, where both sources and
destinations are on the boundary of the polygon.  It returns the
geodesic path between the source and the destination for every input
pair assuming that the $k$ shortest paths do not cross (but possibly overlap)
one another.  Computing the $k$ geodesic paths takes $O(N+k)$ time
in total, where $N$ is the complexity of the polygon.  In our case, the
pairs $(w_j,w_{j+1})$ for $j=0,\ldots,\beta$ are $\beta+1$ input
source-destination pairs.  Since the geodesic paths for all input
pairs do not cross one another, $\gamma$ can be computed in
$O(\beta+\complexity{C}) =O(\complexity{C})$ time.
Then we compute $\rfvd\cap\gamma$ in
$O(|C|+\complexity{\rfvd\cap\bd C})$ time using $
\rfvd\cap\bd C$ obtained from the $k$th iteration.  We
will describe this procedure in Section~\ref{sec:cellboundary_fvd}.

The curve $\gamma$ subdivides $C$ into $t_{k+1}$-path-cells. To be specific, 
$C \setminus \gamma$ consists of at least $\beta+2$ connected
components.  Note that the closure of each connected component is a
$t_{k+1}$-path-cell.  Moreover, the union of the closures of all
connected components is exactly $C$ since $C$ is simple.  
These components define the \textit{subdivision of} $C$ \textit{induced by} $\gamma$.

\subsubsection{Phase 2. Subdivision along an Arc of \texorpdfstring{$\rfvd$}{rFVD}}
After subdividing $C$ into $t_{k+1}$-path-cells $C_1,\ldots,C_\delta$ 
$(\delta \geq \beta+2)$
by the curve $\gamma_C$, an arc $\alpha$ of $\rfvd$ may \emph{cross} $C_j$
for some $1 \leq j \leq \delta$. 
We say an arc $\alpha$ of $\rfvd$ \emph{crosses} a
cell $C'$ if $\alpha$ intersects at least two distinct edges of $C'$.
For example, in Figure~\ref{fig:tpathcell}(c), $\alpha$ crosses $C_2$
while $\alpha$ does not cross $C_3$ 
because $\alpha$ crosses only one edge of $C_3$.
In Phase 2, for each arc
$\alpha$ crossing $C_j$, we isolate the subarc $\alpha \cap C_j$.
That is, we subdivide $C_j$ further into three subcells so that
only one of them intersects $\alpha$. We call such a subcell an arc-quadrilateral. Moreover, for
an arc-quadrilateral $\square$ created by an arc $\alpha$ crossing $C_j$, we have
$\rfvd\cap\square  = \alpha \cap C_j$.

\begin{lemma}
  For a geodesic convex polygon $C$ with $t$ convex vertices ($t
  \in \mathbb{N}$), let ${\gamma}$ be a simple closed curve
  connecting at most $t$ convex vertices of $C$ lying on $\bd P$ 
  such that every two consecutive vertices are connected
  by a geodesic path.
   Then, each arc $\alpha$ of $\rfvd$ intersecting $C$ 
   intersects at most three cells in the subdivision of
  $C$ induced by ${\gamma}$ and at most two edges of $\gamma$.
  \label{lem:at_most_three}
\end{lemma}
\begin{proof}
  Consider an arc $\alpha$ of $\rfvd$ intersecting $C$.
  The arc $\alpha$ is a part of either a side of some apexed triangle or the bisector of two sites.
  For the first case, the arc $\alpha$ is a line segment. Thus $\alpha$
  intersects at most three cells in the subdivision of $C$ by
  ${\gamma}$ and at most two edges of $\gamma$. 
  For the second case, $\alpha$ is part of a hyperbola.  Let $s_1$ and
  $s_2$ be the two sites defining $\alpha$ in $\rfvd$.  The
  combinatorial structure of the geodesic path from $s_1$ (or $s_2$)
  to any point in $\alpha$ is the same. This means that $\alpha$ is
  contained in the intersection of two apexed triangles $\tri_1$ and $\tri_2$, one with definer $s_1$ and the other with definer $s_2$.
  Observe that $\tri_1 \cap \tri_2$ intersects $ {\gamma}$ at most twice
  and contains no vertex of $ {\gamma}$ in its interior.
  By construction, $\tri_1 \cap \tri_2$ intersects at most two edges $e_1$ and $e_2$
  of $ {\gamma}$, and thus so does $\alpha$.  For a cell $C'$ in the
  subdivision of $C$ by $ {\gamma}$, the arc $\alpha$ intersects $C'$ if
  and only if $C'$ contains $e_1$ or $e_2$ on its boundary.
  Thus there exist at most three such cells in the subdivision of $C$
  by $ {\gamma}$ and the lemma holds for the second case.
  See Figure~\ref{fig:tpathcell}(b-c).  
\end{proof}

First, we find $\alpha\cap C_j$ for every arc $\alpha$ of $\rfvd$ crossing $C_j$.
If $\rcell{\tri}\cap \bd C_j$ consists of at most two connected components (line segments) for 
every apexed triangle $\tri\in\mathcal{A}$,
we can do this by scanning all points in $\rfvd\cap \bd C_j$ along $\bd C_j$.
However, $\rcell{\tri}\cap \bd C_j$ might consist of more than
two connected components (line segments) for some apexed triangle $\tri\in\mathcal{A}$.  See Figure~\ref{fig:connected_component}.
Despite of this fact, we can compute all such arcs in $O(\complexity{\rfvd\cap \bd C_j})$ time by the following lemma.

\begin{figure}
  \begin{center}
    \includegraphics[width=0.4\textwidth]{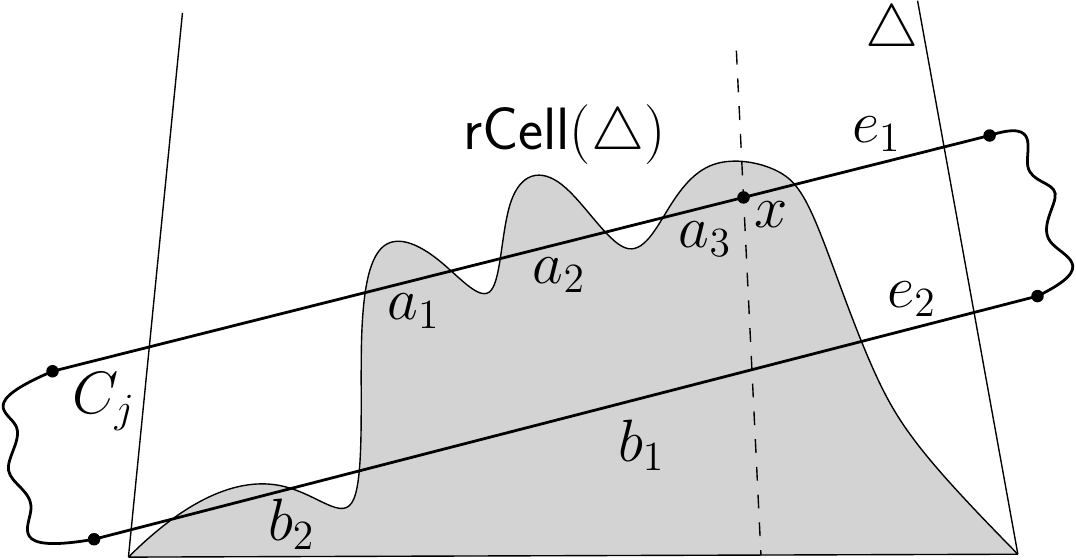}
    \caption {\small $\rcell{\tri} \cap \bd C_j$ consists of five
      connected components $a_i$ contained in $e_1$ and 
      $b_j$ contained in $e_2$ for $i=1,2,3$ and $j=1,2$.}
    \label{fig:connected_component}
  \end{center}
\end{figure}

\begin{lemma}  \label{lem:find_arcs}
  For every arc $\alpha$ of $\rfvd$ crossing $C_j$, we can find the part of $\alpha$ contained in $C_j$ in
  $O(\complexity{\rfvd \cap \bd C_j})$ time in total.  Moreover, for each such arc $\alpha$,
  the pair $(\tri_1,\tri_2)$ of apexed triangles 
  such that $\alpha \cap C_j = \{x\in C_j :
  g_{\triangle_1}(x)=g_{\triangle_2}(x)>0\}$ can be found in the same time.
\end{lemma}
\begin{proof}	
  For each apexed triangle $\tri\in\mathcal{A}$ intersecting $C_j$, we find all connected
  components of $\rcell{\tri}\cap\bd C_j$. Since we already have $\rfvd\cap \bd C_j$, 
  this takes $O(\complexity{\rfvd \cap \bd C_j})$ time for all apexed triangles in $\mathcal{A}$ intersecting $C_j$.
  There are at most two edges of $\bd
  C_j$ that are intersected by $\rcell{\tri}$ due to Lemma~\ref{lem:at_most_three}. Let 
  $e_1$ and $e_2$ be such edges, and we assume that $e_1$ contains the point in $\rcell{\tri}\cap\bd C_j$ closest to $\apex{\tri}$
  without loss of generality.  We insert all
  connected components of $\rcell{\tri} \cap e_1$ in the clockwise
  order along $\bd C_j$ into a queue.  Then, we consider the
  connected components of $\rcell{\tri} \cap e_2$ in the clockwise
  order along $\bd C_j$ one by one.

  To handle a connected component $r$ of $\rcell{\tri} \cap e_2$, we
  do the following.  Let $x$ be a point in the first element $r'$ of
  the queue.  If the line passing through $x$ and $\apex{\tri}$
  intersects $r$, then we remove $r'$ from the queue and check whether
  $r$ and $r'$ are incident to the same refined cell $\rcell{\tri'}$.
  If so, we compute the part of the arc defined by $\tri$ and $\tri'$ inside $C_j$, and
  return the part of the arc and the pair $(\tri,\tri')$.  
  If $r$ and $r'$ are not incident to the same refined cell, we remove $r'$ from the queue since.
  We repeat this until the line passing through a point of the first element of the queue
  does not intersect $r$. Then we handle the connected component of
  $\rcell{\tri} \cap e_2$ next to $r$.

  Every arc computed from this procedure is an arc of $\rfvd$ crossing $C_j$.
  The remaining work is to show that we can find all arcs of $\rfvd$ crossing $C_j$ using this procedure.
  Consider an arc $\alpha$ of $\rfvd$ crossing $C_j$.
  There are two connected
  components $r\subseteq e_1$ and $r'\subseteq e_2$ of $\rfvd(\tri) \cap \bd C_j$ incident to a point in $\alpha\cap \bd C_j$.
  The line passing through any point $x\in \alpha\cap C_j$ and $\apex{\tri}$  intersects $e_2$ once.
  Moreover, this intersection is in $\rcell{\tri}$ by Lemma~\ref{lem:ray_in_cell}.  Since $r'$ contains 
  the set of all such intersections, the line passing through a point in $r$ and
  $\apex{\tri}$ intersects $r'$. Thus the procedure finds $\alpha$.
\end{proof}

By Lemma~\ref{lem:at_most_three}, $\alpha \cap C_j$ consists of at most two connected components.
For the case that it consists of
exactly two connected components, we consider each connected component
separately.  In the following, we consider only the case that $\alpha \cap C_j$ is connected.

\begin{figure}
	\begin{center}
		\includegraphics[width=0.7\textwidth]{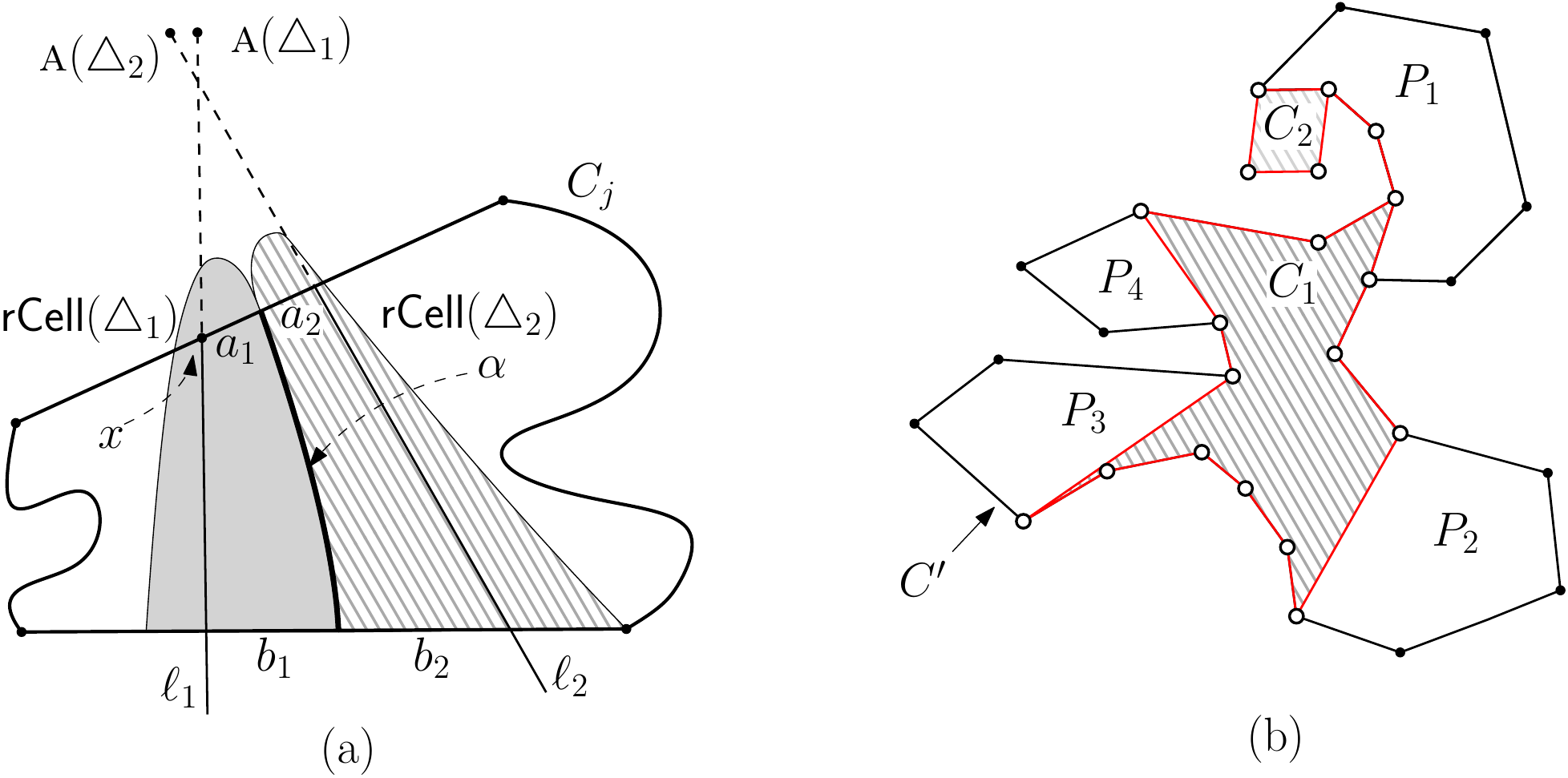}
		\caption {\small (a) The arc $\alpha$ of $\rfvd$ crosses $C_j$.
			We isolate $\alpha$ by subdividing $C_j$ into three subcells
			with respect to $\ell_1$ and $\ell_2$. (b) The vertices marked with
			empty disks are the vertices of $P$ while the others are vertices of arc-quadrilaterals lying
			in 	$\interior{P}$. We subdivide the cell into two $t$-path-cells $C_1, C_2$
			and four lune-cells $P_1,\ldots,P_4$.}
		\label{fig:compute_quadrilateral}
	\end{center}
\end{figure}

For an arc $\alpha$ crossing $C_j$, we subdivide $C_j$ further into
two cells with $t'$ convex vertices for $t' \leq t_{k+1}$ and one
arc-quadrilateral by adding two line segments bounding $\alpha$ so 
that no arc other than $\alpha$ intersects the arc-quadrilateral.  
Let $(\triangle_1, \triangle_2)$ be the pair of apexed
triangles defining $\alpha$.  Let $a_1, b_1$ (and $a_2,b_2$)
be the two connected components of $\rcell{\triangle_1}\cap\bd C_j$
(and $\rcell{\triangle_2}\cap\bd C_j$) incident to $\alpha$ such that
$a_1, a_2$ are adjacent to each other and $b_1,b_2$ are adjacent to
each other.  See Figure~\ref{fig:compute_quadrilateral}(a).
Without loss of generality, we assume that $a_1$ is closer than $b_1$
to $\apex{\triangle_1}$.  Let $x$ be any point on $a_1$.  Then the
$V$-farthest neighbor of $x$ is the definer of $\triangle_1$.  We consider
the line $\ell_1$ passing through $x$ and the apex of $\triangle_1$.
Then the intersection between $C_j$ and $\ell_1$ is contained in the
closure of $\rcell{\triangle_1}$ by Lemma~\ref{lem:ray_in_cell}.
Similarly, we find the line $\ell_2$ passing through the apex of
$\triangle_2$ and a point on $a_2$.

We subdivide $C_j$ into two cells with at most $t_{k+1}$ convex
vertices and one arc-quadrilateral by $\ell_1$ and $\ell_2$.
The quadrilateral bounded by the two lines and $\bd C_j$
is an arc-quadrilateral since $\alpha$ crosses the quadrilateral but
no other arcs of $\rfvd$ intersect the quadrilateral.
We do this for all arcs crossing $C_j$.  Note that no arc crosses
the resulting cells other than arc-quadrilaterals by the construction.
The resulting cells with at most $t_{k+1}$ convex vertices and
arc-quadrilaterals are the cells in the subdivision of $C$ obtained from Phase 2.
Therefore, we have the following lemma.

\begin{lemma}
  \label{lem:nonbase_cell_cross}
  No arc of $\rfvd$ crosses cells other than arc-quadrilaterals created in
  Phase 2.
\end{lemma}

\subsubsection{Phase 3. Subdivision by a Geodesic Convex Hull}
Note that some cell $C'$ with $t'$ convex vertices for $3< t ' \leq
t_{k+1}$ created in Phase~2 might be neither a $t'$-path-cell
nor a base cell. This is because a cell created in Phase~2 might have some vertices in $\interior{P}$.  
In Phase~3, we subdivide such cells further
into $t'$-path-cells and lune-cells.

To subdivide $C'$ into $t_{k+1}$-path-cells and lune-cells, we first
compute the geodesic convex hull $\ch$ of the vertices of $C'$ which come from the vertex set of $P$
in time linear in the number of edges in $C'$ using the
algorithm for computing $k$ shortest paths in \cite{kpairpath}.
Consider the connected components of $C' \setminus \bd \ch$.
They belong to one of two types defined as follows. 
A connected component of the first
type is enclosed by a closed simple curve which is part of $\bd \ch$.
For example, $C_1$ and $C_2$ in
Figure~\ref{fig:compute_quadrilateral}(b) are the connected
components belonging to this type.  
A connected component of the second type is enclosed by a
subchain of $\bd \ch$ from $u$ to $w$ in clockwise order and a
subchain of $\bd C'$ from $w$ to $u$ in counterclockwise order for
some $u, w \in \bd P$.  For example, $P_i$ in
Figure~\ref{fig:compute_quadrilateral}(b) is the connected
component belonging to the second type for $i=1,\ldots,4$.

By the construction, a connected component belonging to the first type
has all its vertices from the vertex set of $P$.  Moreover, it has at most $t'$ convex
vertices since $C'$ has $t'$ convex vertices.  Therefore, the closure
of a connected component of $C'\setminus \bd \ch$ belonging to the
first type is a $t'$-path-cell with $t'\leq t_{k+1}$.

Every vertex of $C'$ lying in $\interior{P}$ is convex with
respect to $C'$ by the construction of $C'$.  Thus, for a connected
component $P'$ belonging to the second type, 
the part of $\bd P'$
from $\bd C'$ is a convex chain with respect to $P'$.  Moreover,
the part of $\bd P'$ from $\bd \ch$ is the geodesic path between
two points, and thus it is a concave chain with respect to $P'$.  
Therefore, the closure of a connected component belonging to the
second type is a lune-cell.

Since $C'$ is a simple polygon, the union of the closures of all
connected components of $C' \setminus \ch$ is exactly the closure of $C'$.  The
closures of all connected components belonging to the first and the
second types are $t_{k+1}$-path-cells and lune-cells created at the end
of the $(k+1)$th iteration, respectively.  We compute the
$t_{k+1}$-path-cells and the lune-cells induced by $\bd \ch$.
Then, we compute $\rfvd\cap \bd \ch$
using the procedure in Section~\ref{sec:cellboundary_fvd}.

The resulting $t_{k+1}$-path-cells and base cells form the final
decomposition of $C$ of the $(k+1)$th iteration.

\subsubsection{Analysis of the Complexity}

We first give the combinatorial complexity of the refined geodesic farthest-point
Voronoi diagram restricted to the boundary of the cells from each iteration.
Note that an arc of $\rfvd$ might cross some $t_{k}$-path-cells in the decomposition 
at the end of the $k$th iteration for any $k$ while no arc of $\rfvd$ crosses cells
other than arc-quadrilaterals created in Phase~2.
The following lemma is used to prove the complexity.

\begin{lemma}
  \label{lem:arc_intersect_cell}
  An arc $\alpha$ of $\rfvd$ intersects at most nine $t_k$-path-cells and
  $O(k)$ base cells at the end of the $k$th iteration for any $k\in\mathbb{N}$.  Moreover,
  there are at most three $t_k$-path-cells that $\alpha$ intersects
  but does not cross at the end of the $k$th iteration.
\end{lemma}
\begin{proof}
  Let $\alpha$ be an arc of $\rfvd$. 
  If $\alpha$ is a line segment, the lemma holds directly. Thus, we consider the case that
  $\alpha$ is a part of hyperbola defined by a pair $(\tri_1,\tri_2)$ of apexed triangles.
  That is, $\alpha \subseteq \{x\in\tri_1\cap \tri_2 : g_{\tri_1}(x)=g_{\tri_2}(x) \geq 0\}$.
	
  We first show that there is at most one $t_k$-path-cell at the end of the $k$th
  iteration that $\alpha$ intersects but does not cross, and no endpoint of $\alpha$ is contained in.
  Assume to the contrary that there are two such $t_k$-path-cells. 
  Consider two edges $e_1$ and $e_2$ from the two cells which intersect $\alpha$.
  Notice that they are distinct.

  We claim that $e_1$ and $e_2$ intersect at some point other than their endpoints, which makes a contradiction.
  To prove the claim, we assume that the line containing
  $\apex{\tri_1}$ and $\apex{\tri_2}$ is the $x$-axis.  Then
  $\alpha$ is part of a hyperbola whose foci lie on the $x$-axis.
  The arc $\alpha$ does not intersect the $x$-axis.
  Let $h_1$ and $h_2$ be the lines tangent to the hyperbola containing $\alpha$ at the
  endpoints of $\alpha$.  We denote the region bounded by $h_1$, $h_2$
  and $\alpha$ by $R$.  Then to prove the claim, it suffices to
  show that $R$ is contained in $\tri_1\cap \tri_2$ because no vertex lies in the interior of $\tri_1\cap\tri_2$.
  Assume that $R$ is not contained in $\tri_1$.
  Then one of the sides of
  $\tri_1$ incident to $\apex{\tri_1}$ intersects $R$.  Thus, there is
  a line passing through $\apex{\tri_1}$ which intersects $\alpha$
  twice by the property of the hyperbola, which is a contradiction by Lemma~\ref{lem:ray_in_cell}.  
  The case that $R$ is not contained in $\tri_2$ is analogous.
  Therefore, the claim holds. Including the two $t_k$-path-cells containing an endpoint of $\alpha$,
  there are at most three $t_k$-path-cells that $\alpha$ intersects but does not cross.
  
  Now we show that $\alpha$ intersects at most nine $t_k$-path-cells
  and $O(k)$ base cells at the end of the $k$th iteration.  For $k=1$,
  $P$ itself is the decomposition of $P$, thus there exists only one
  cell.
  For $k\geq 2$, assume that the lemma holds for the $k'$th iterations
  for all $k' < k$.  

  We claim that the $k$th iteration creates
  a constant number of the arc-quadrilaterals that $\alpha$ intersect.
  Due to the assumption, $\alpha$ intersects at most nine
  $t_{k-1}$-path-cells at the end of the $(k-1)$th iteration. Thus, at the end of Phase 1 of the $k$th iteration,
  $\alpha$ crosses at most 27 $t_{k}$-path-cells by
  Lemma~\ref{lem:at_most_three}.  Note that $\alpha \cap C'$ might
  consist of two connected components for a $t_{k}$-path-cell $C'$
  created in Phase 1. See Figure~\ref{fig:tpathcell}(c).  In
  this case, we create two arc-quadrilaterals. If $\alpha \cap C'$ is
  connected, we create one arc-quadrilateral.  Thus, we create at most
  54 arc-quadrilaterals crossed by $\alpha$ in the $k$th iteration.
  Therefore, in the $k$th iteration, there are $O(k)$ arc-quadrilaterals intersecting $\alpha$.
  
  We claim that the number of the $t_k$-path-cells that $\alpha$ intersects at the end of the $k$th iteration is 
  at most nine.
  There are three cells from Phase~2 other than arc-quadrilaterals intersecting $\alpha$.
  Note that $\alpha$ does not cross more than two cells. Thus, it is sufficient to consider only these three cells.
  Each cell $C$ from Phase~2 intersecting $\alpha$ is subdivided into smaller cells in Phase~3.
  Due to Lemma~\ref{lem:at_most_three}, at most three smaller cells intersect $\alpha$.
  Thus, in total, the $k$th iteration creates at most nine cells of the $t_k$-path-cells that $\alpha$ intersects.
  Similarly, we can prove that 
  the $k$th iteration creates a constant number of lune-cells intersecting $\alpha$. 
  Therefore, the lemma holds.
\end{proof}

Now we are ready to prove the complexities of 
the cells and $\rfvd$ restricted to the cells in each iteration.
Then we finally prove that the running time of the algorithm  
in this section is $O(n\log\log n)$.

\begin{lemma}
  \label{lem:complexity_tpathcell}
  At the end of the $k$th iteration for any $k\in\mathbb{N}$, the following holds.
  \begin{itemize}
  \item[] \hspace{0\textwidth}\rlap{$\sum _{C : \textnormal{a }
        t_k\textnormal{-path-cell}} {\complexity{\rfvd\cap\bd C}} =
      O(n)$.}
  \item[] \hspace{0\textwidth}\rlap{$\sum _{C : \textnormal{a }
        t_k\textnormal{-path-cell}} {\complexity{C}}= O(n)$.}
  \item[] \hspace{0\textwidth}\rlap{$\sum _{T : \textnormal{a base
          cell}} {\complexity{\rfvd\cap\bd T}}= O(kn)$.}
  \item[] \hspace{0\textwidth}\rlap{$\sum _{T : \textnormal{a base
          cell}} {\complexity{T}}= O(kn)$.}
  \end{itemize}
\end{lemma}
\begin{proof}
  Let $\alpha$ be an arc of $\rfvd$.  The first and the third
  complexity bounds hold by Lemma~\ref{lem:arc_intersect_cell} and the
  fact that the number of the arcs of $\rfvd$ is $O(n)$.

  The second complexity bound holds since the set of all edges of the
  $t_k$-path-cells is a subset of the chords in some triangulation of
  $P$.  Any triangulation of $P$ has $O(n)$ chords.
  Moreover, each chord is incident to at most two $t_k$-path-cells.

  For the last complexity bound, the number of the edges of the base cells
  whose endpoints are vertices of $P$ is $O(n)$ since they are chords
  in some triangulation of $P$.  Thus we count the number of edges of the base cells
  which are not incident to vertices of $P$.  In Phase 1, we do not create any such
  edge.  In Phase 2, we create at most
  $O(1)$ such edges whenever we create one arc-quadrilateral.  All
  edges created in Phase 3 have their endpoints from the vertex set of $P$.
  Therefore, the total number of the edges of all base cells is
  asymptotically bounded by the number of arc-quadrilaterals, which is
  $O(kn)$.
\end{proof}

\begin{corollary}
  \label{lem:last_iter_base_cell}
  In $O(\log\log n)$ iterations, the polygon is subdivided into
  $O(n\log\log n)$ base cells.
\end{corollary}

\begin{lemma}
  The subdivision in each iteration can be done in $O(n)$ time.
\end{lemma}
\begin{proof}
  In Phase 1, we compute $\gamma$ and $\rfvd \cap \gamma$
  for each $t$-path-cell $C$ from the previous iteration.  The running
  time for this is linear in the total complexity of all
  $t$-path-cells in the previous iteration and $\rfvd$ restricted on
  the boundary of all $t$-path-cells by Lemma~\ref{lem:gamma_cap_fvd}, which is $O(n)$ by
  Lemma~\ref{lem:complexity_tpathcell}.
	
  In Phase 2, we first scan $\rfvd \cap \bd C'$ for all cells
  $C'$ from Phase 1 to find an arc of $\rfvd$ crossing some cell. This
  can also be done in linear time by Lemma~\ref{lem:find_arcs} and
  Lemma~\ref{lem:complexity_tpathcell}.  For each arc crossing some
  $t$-path-cell, we compute two line segments bounding the arc and
  subdivide the cell into two smaller regions and one
  arc-quadrilateral in $O(1)$ time.  Each arc of $\rfvd$ crosses at most $O(1)$
  cells from Phase 1, and the time for this step is $O(n)$ in
  total.
	
  In Phase 3, we further subdivide each cell which is not a
  base cell from Phase 2.  In the subdivision of a cell which
  is not a base cell $C$ in Phase 2, we first compute the
  geodesic convex hull $\ch$ of the vertices of $C'$ which are vertices of $P$.
  The geodesic convex hull can be computed in time linear in the
  complexity of $C'$.  By Lemma~\ref{lem:gamma_cap_fvd}, $
  \rfvd\cap\bd \ch$ can be computed in $O(\complexity{\rfvd\cap \bd C'} +
  \complexity{C'})$ time.  Note that all cells other than the base
  cells from Phase 2 are interior disjoint.  Moreover, the
  total number of the edges of such cells is $O(n)$.  Similarly, the
  total complexity of $\rfvd \cap \bd C'$ for all such cells $C'$ is
  $O(n)$.  Therefore, the $t$-path-cells and lune-cells can be
  computed in $O(n)$ time.
\end{proof}

\subsection{Computing \texorpdfstring{$\rfvd$}{rFVD} Restricted to a Curve Connecting Vertices of \texorpdfstring{$P$}{P}}
\label{sec:cellboundary_fvd}
In this section, we describe a
procedure to compute $\rfvd\cap\gamma$ in $O(\complexity{\rfvd\cap\bd C} + \complexity{C})$ time once 
we have $\rfvd\cap\bd C$, where $C\subseteq P$ is a geodesic convex polygon and $\gamma$ is a simple closed curve connecting
some convex vertices of $C$ lying on $\bd P$  in clockwise order along $\bd C$ by the geodesic paths connecting two consecutive vertices.
For an apexed triangle $\triangle$ with $\rcell{\triangle}\cap\gamma\neq\emptyset$, we have 
$\rcell{\triangle}\cap\bd C\neq\emptyset$ by Lemma~\ref{lem:ray_in_cell}.  Thus 
we consider only the apexed triangles $\triangle$ with
$\rcell{\triangle}\cap\bd C\neq\emptyset$.  Let $\mathcal{L}$ be the list
of all such apexed triangles sorted along $\bd P$ with respect to
their bottom sides. (Recall that the bottom sides of all apexed triangles are
interior-disjoint.  Moreover, the union
of them is $\bd P$ by the construction.) Note that $|\mathcal{L}|=
O(\complexity{\rfvd\cap\bd C})$.

Consider a line segment $ab$ contained in $P$.  Without loss of
generality, we assume that $ab$ is horizontal and $a$ lies to the left
of $b$.  Let $\triangle_a$ and $\triangle_b$ be the apexed triangles
that  maximize $g_{\triangle_a}(a)$ and $g_{\triangle_b}(b)$,
respectively. If there is a tie by more than one apexed triangle, we
choose an arbitrary one of them. Note that $\rcell{\tri_a}$ and $\rcell{\tri_b}$
contain $a$ and $b$ in their closures, respectively. With the two apexed triangles, we
define two sorted lists $\mathcal{L}_{ab}$ and $\mathcal{L}_{ba}$ as follows.
Let $\mathcal{L}_{ab}$ be the sorted list of the apexed triangles in
$\mathcal{L}$ which intersect $ab\setminus\{a,b\}$ and whose bottom sides on $\bd P$ lie from the
bottom side of $\triangle_a$ to the bottom side of $\triangle_b$ including $\tri_a$ and $\tri_b$ in
clockwise order along $\bd P$. Similarly, let $\mathcal{L}_{ba}$ be
the sorted list of the apexed triangles in $\mathcal{L}$ which
intersect $ab\setminus\{a,b\}$ and whose bottom sides lie from the bottom side of
$\triangle_b$ to the bottom side of $\triangle_a$ including $\tri_a$ and $\tri_b$ in clockwise order
along $\bd P$. Note that no apexed triangle other than $\tri_a$ and $\tri_b$
appears both $\mathcal{L}_{ab}$ and $\mathcal{L}_{ba}$. See Figure~\ref{fig:list_apex}(a).
	
The following lemma together with Section~\ref{sec:gamma_cap_fvd}
gives a procedure to compute $\rfvd\cap ab$.
This procedure is similar to the procedure in Section~\ref{sec:first-step} that computes
$\rfvd$ restricted to $\bd P$.

\begin{figure}
  \begin{center}
    \includegraphics[width=0.8\textwidth]{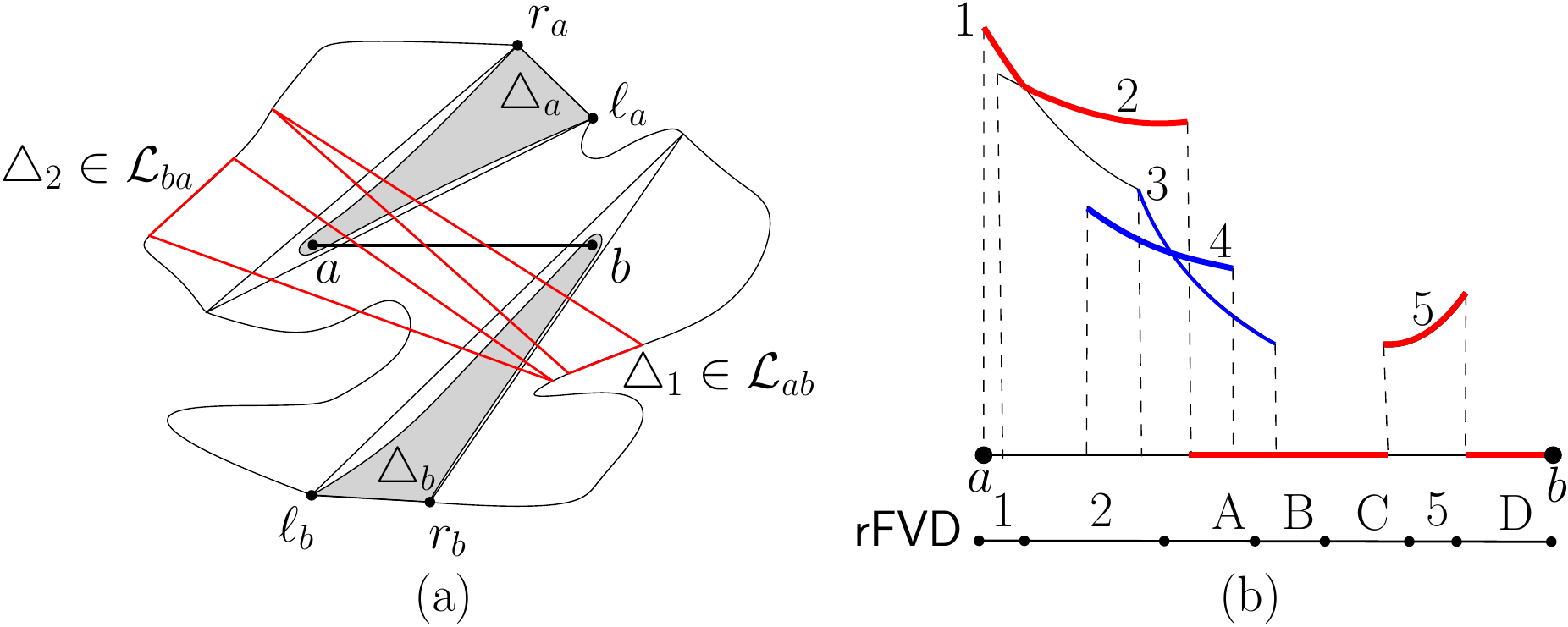}
    \caption {\small (a) The apexed triangles with their bottom sides
      in $\subchain{\ell_a}{r_b}$ and $\subchain{\ell_b}{r_a}$ are
      in $\mathcal{L}_{ab}$ and $\mathcal{L}_{ba}$,
      respectively. $\tri_1\in \mathcal{L}_{ab}$ and $\tri_2\in\mathcal{L}_{ba}$.  (b) The hyperbolic arcs are the graphs of the distance
      functions associated with the apexed triangles in
      $\mathcal{L}_{ab}$. Curves 1, 2, and 5 represent a partial upper envelope of all distance functions.  
      Note that curves 3 and 4 do not appear on the complete upper envelope
      which coincides with $\rfvd$ because curves A and B from $\mathcal{L}_{ba}$ appear on the complete upper envelope.
	}
    \label{fig:list_apex}
  \end{center}
\end{figure}

\begin{lemma}
	\label{lem:fvd-restricted-line}
  Let $C$ be a geodesic convex polygon, and let $a$ and $b$ be two points with
  $ab \subset C$.  Given the two sorted lists $\mathcal{L}_{ab}$ and
  $\mathcal{L}_{ba}$, we can compute $\rfvd\cap ab$ in 
  $O(|\mathcal{L}_{ab}|+|\mathcal{L}_{ba}|)$ time.
\end{lemma}
\begin{proof}
  Recall that the upper envelope of $g_\triangle$ on $ab$ for all apexed
  triangles $\triangle \in \mathcal{L}_{ab}\cup\mathcal{L}_{ba}$
  (simply, the upper envelope for
  $\mathcal{L}_{ab}\cup\mathcal{L}_{ba}$) coincides with $\rfvd\cap
  ab$ in its projection on $ab$ by definition.  Thus we compute the
  upper envelope for $\mathcal{L}_{ab}\cup\mathcal{L}_{ba}$.  To this
  end, we compute a ``partial'' upper envelope of $g_\triangle$ on
  $ab$ for all apexed triangles $\triangle \in \mathcal{L}_{ab}$.
  After we do this also for the apexed triangles in $\mathcal{L}_{ba}$, we
  merge the two ``partial'' upper envelopes on $ab$ to obtain the
  complete upper envelope of $g_\triangle$ on $ab$ for all apexed
  triangles $\triangle \in \mathcal{L}_{ab}\cup\mathcal{L}_{ba}$.
	
  A \textit{partial upper envelope} for $\mathcal{L}_{ab}$ is the
  upper envelope for $A \subseteq \mathcal{L}_{ab}$ satisfying that
  $\triangle \in \mathcal{L}_{ab}$ 
  belongs to $A$ if $\rcell{\triangle}\cap ab \neq \emptyset$. Here, an apexed triangle $\tri\in A$
  does not necessarily have a refined Voronoi cell on $ab$.  Thus, a partial upper
  envelope for $\mathcal{L}_{ab}$ (and $\mathcal{L}_{ba}$) is not necessarily
  unique. The upper envelope of two
  partial upper envelopes, one for $\mathcal{L}_{ab}$ and one for
  $\mathcal{L}_{ba}$, is the complete upper envelope for
  $\mathcal{L}_{ab}\cup\mathcal{L}_{ba}$ by definition. 
  See  Figure~\ref{fig:list_apex}(b).  

  In the following, we show how to compute one of the partial upper
  envelopes for $\mathcal{L}_{ab}$.  A partial upper envelope for
  $\mathcal{L}_{ba}$ can be computed analogously.
  Then the complete upper envelope can be
  constructed in $O(|\mathcal{L}_{ab}|+|\mathcal{L}_{ba}|)$ time by scanning
  the two partial upper envelopes along $ab$.

  For any two apexed triangles $\triangle_1, \triangle_2 \in
  \mathcal{L}_{ab}$ such that $\triangle_1$ comes before $\tri_2$ in the sorted
  list $\mathcal{L}_{ab}$, 
  $\rcell{\triangle_1}\cap ab$ lies to the left of $\rcell{\triangle_2}\cap
  ab$ along $uv$ if they exist.  If it is not true, there is a point 
  contained in $\rcell{\triangle_1} \cap \rcell{\triangle_2}$ by
  Lemma~\ref{lem:ray_in_cell}, which contradicts that all refined
  cells are pairwise disjoint.  
  With this property, 
  a  partial upper envelope for $\mathcal{L}_{ab}$ can be constructed in
  a way similar to the procedure for computing $\rfvd \cap \bd P$ in
  Section~\ref{sec:transition_edge_fvd}.  The difficulty here is that
  we must avoid maintaining geodesic paths as it takes $O(n)$ time,
  which is too much for our purpose.

  We consider the apexed triangles in $\mathcal{L}_{ab}$ from
  $\triangle_a$ to $\triangle_b$ one by one as follows.  Let $U$ be
  the current partial upper envelope of the distance functions of the 
  apexed triangles from $\triangle_a$ to $\triangle'$ of
  $\mathcal{L}_{ab}$ and $\tau$ be the list of the apexed triangles
  whose distance functions restricted to $ab$ appear on $U$ in the
  order in which they appear on $U$. Note that $U$ is not necessarily continuous.
  We maintain all connected components of $U$ here while 
  we maintain only one connected component of $U$ in Section~\ref{sec:first-step}.
  We show how to update $U$ to a
  partial upper envelope of the distance functions of the apexed triangles
  from $\triangle_a$ to $\triangle$, where $\triangle$ is the apexed
  triangle next to $\triangle'$ in $\mathcal{L}_{ab}$.  Let
  $\triangle_r$ be the last element in $\tau$ and $\mu$ be the line
  segment contained in $ab$ such that $g_{\triangle_r}(x) = U(x) > 0$
  for every point $x \in \mu$.
	



  There are three possibilities: (1) $\triangle\cap\mu\neq\emptyset$.  In
  this case, we compare the distance functions of $\triangle$ and
  $\triangle_r$ on $\triangle \cap\mu$. Depending on the result, we
  update $U$ and $\tau$ as we did in
  Section~\ref{sec:transition_edge_fvd}.  (2) $\triangle\cap\mu=\emptyset$
  and $\triangle\cap ab$ lies to the right of $\mu$.  We append
  $\triangle$ to $\tau$ at the end and update $U$ accordingly.  (3)
  $\triangle\cap\mu=\emptyset$ and $\triangle\cap ab$ lies to the left of
  $\mu$.  We have to use a method different from the one in
  Section~\ref{sec:transition_edge_fvd} to handle this case. Here,
  contrast to the case in Section~\ref{sec:transition_edge_fvd},
  $\triangle_r$ intersects $\triangle$.
  Thus, we can check whether $\rcell{\triangle}\cap ab = \emptyset$ or
  $\rcell{\triangle_r}\cap ab = \emptyset$ easily as follows.  Consider the
  set $R=\triangle\cap\triangle_r$. The distance functions associated
  with $\triangle$ and $\triangle_r$ have positive values on $R$, and thus
  we can compare the geodesic distances from $\definer{\triangle}$ and
  $\definer{\triangle_r}$ to any point in $R$.  Depending on the result,
  we can check in constant time whether $\rcell{\triangle}$ and
  $\rcell{\triangle_r}$ intersect the connected regions
  $\triangle\setminus R$ and $\triangle_r\setminus R$ containing
  $\apex{\triangle}$ and $\apex{\triangle_r}$, respectively.  If
  $\rcell{\triangle}$ does not intersect the connected region
  $\triangle\setminus R$ containing $\apex{\triangle}$, then
  $\rcell{\triangle_r}$ does not intersect $ab$.  This also holds for
  $\rcell{\triangle_r}$.  Depending on the result, we apply the procedure
  in Section~\ref{sec:transition_edge_fvd}.
	
  In this way, we append an apexed triangle $\tri$ to $\tau$ if $\rcell{\tri}\cap ab\neq\emptyset$.
  Similarly, we remove some apexed triangle $\triangle$ from $\tau$
  only if $\rcell{\triangle}\cap ab =\emptyset$.
  Thus, by definition, $U$ is a partial upper envelope of the distance
  functions for $\mathcal{L}_{ab}$.
	
  As mentioned above, we do this also for $\mathcal{L}_{ba}$.
  Then we compute the upper envelope of 
  the two resulting partial upper envelopes, which is the
  complete upper envelope for $\mathcal{L}_{ab}\cup\mathcal{L}_{ba}$.
  This takes $O(|\mathcal{L}_{ab}|+|\mathcal{L}_{ba}|)$ time.
\end{proof}

\begin{corollary}
  \label{lem:const_gamma}
  Let $C$ be a geodesic convex polygon and $E$ be a set of $O(1)$ line
  segments which are contained in $C$.  Then $\rfvd\cap ab$ 
  for all $ab \in E$ can be computed in $O(\complexity{\rfvd\cap\bd C})$ time.
\end{corollary}

Due to Lemma~\ref{lem:fvd-restricted-line}, we can compute $\rfvd \cap \gamma$ in $O(\sum_{ab\in\gamma}{|\mathcal{L}_{ab}|+|\mathcal{L}_{ba}|})$ time
once we compute $\mathcal{L}_{ab}$ and $\mathcal{L}_{ba}$ for all edges $ab$ of $\gamma$.
Recall that every apexed triangle in $\mathcal{L}_{ab}\cup\mathcal{L}_{ba}$ intersects $ab$ by
the definitions of $\mathcal{L}_{ab}$ and $\mathcal{L}_{ba}$.
Since every apexed triangle intersects $\gamma$ in at most two edges, each apexed triangle in $\mathcal{L}$
is contained in $\mathcal{L}_{ab}\cup\mathcal{L}_{ba}$ for at most two edges $ab$ of $\gamma$. 
Therefore, once we have $\mathcal{L}_{ab}$ and $\mathcal{L}_{ba}$ for every edge $ab$ of $\gamma$,
we can compute $\rfvd\cap\gamma$ in $O(|\mathcal{L}|)=O(|\rfvd\cap \bd C|)$ time.
The remaining procedure is computing $\mathcal{L}_{ab}$ and $\mathcal{L}_{ba}$
for all edges $ab$ of $\gamma$.

\subsubsection{Computing \texorpdfstring{$\mathcal{L}_{ab}$}{Lab} and \texorpdfstring{$\mathcal{L}_{ba}$}{Lba} for
  All Edges \texorpdfstring{$ab$}{ab} of \texorpdfstring{$\gamma$}{gamma}}
\label{sec:gamma_cap_fvd}

We show how to compute $\mathcal{L}_{ab}$ and
$\mathcal{L}_{ba}$ for all edges $ab$ of $\gamma$ in
$O(|\mathcal{L}|+\complexity{C})$ time.
By definition, every vertex of $\gamma$ is a vertex of $P$.
Let $ab$ be an edge of $\gamma$,
where $b$ is the clockwise neighbor of $a$ along $\gamma$.  The edge $ab$ 
is a chord of $P$ and divides $P$
into two subpolygons such that $\gamma \setminus ab$ is contained in
one of the subpolygons. Let $P_1(ab)$ be the subpolygon containing
$\gamma \setminus ab$ and $P_2(ab)$ be the other subpolygon. By the
construction, $P_2(ab)$ and $P_2(e')$ 
are disjoint in their interior 
for any edge $e' $ of $\gamma$ other than $ab$. 
For an apexed triangle in $\mathcal{L}_{ab}\cup\mathcal{L}_{ba}$,
either its bottom side lies in $\bd P_2(ab)$ or its apex lies in $\bd P_2(ab)$. Moreover, if its apex
lies in $\bd P_j(ab)$, so does its definer for $j=1,2$ since every apexed triangle in $\mathcal{L}_{ab}\cup\mathcal{L}_{ba}$ intersects $ab\setminus\{a,b\}$.

Using this, we compute $\mathcal{L}_{ab}$ and $\mathcal{L}_{ba}$ for all edges $ab$ in
$\gamma$ as follows.
Initially, we set $\mathcal{L}_{ab}$ and $\mathcal{L}_{ba}$ for all edges $ab$ to $\emptyset$.
We update them by scanning the apexed triangles in
$\mathcal{L}$ from the first to the last.  When we handle an apexed
triangle $\triangle \in \mathcal{L}$, we first find the edge $ab$ of
$\gamma$ such that $P_2(ab)$ contains the bottom side of $\triangle$ 
and check whether $\triangle \cap ab=\emptyset$. If it is nonempty, we
append $\triangle$ to $\mathcal{L}_{ab}$ or $\mathcal{L}_{ba}$ accordingly.
Since we have $\tri_a$ and $\tri_b$ (recall that $a$ and $b$ are also vertices of $P$),
we can decide if $\mathcal{L}_{ab}$ or $\mathcal{L}_{ba}$ contains $\tri$ in constant time.  
Otherwise, we do nothing.
We repeat this with the apexed triangles in $\mathcal{L}$ one by one in order
  until the last one of $\mathcal{L}$ is handled. Then we scan $\mathcal{L}$ again
and update $\mathcal{L}_{ab}$ and $\mathcal{L}_{ba}$
analogously, except that we find the edge
$ab$ of $\gamma$ such that $P_2(ab)$ contains the \textit{definer} of $\triangle$.
This means that we scan $\mathcal{L}$ twice in total, once with respect to the bottom sides and once
with respect to the definers.

This can be done in $O(|\mathcal{L}|)$ time in total for all edges of $\gamma$ and
all apexed triangles in $\mathcal{L}$.
To see this, observe that the order of any three apexed triangles appearing on $\mathcal{L}$ is the same as the order
of their definers (and their bottom sides) appearing on $\bd P$.  Thus
to find the edge $ab$ of $\gamma$ such that $P_2(ab)$ contains 
the definer (or the
bottom side) of $\triangle$, it is sufficient to check at most two
edges: the edge $e'$ such that $P_2(e')$ contains the bottom side of
the apexed triangle previous to $\triangle$ in $\mathcal{L}$ and the clockwise
neighbor of $e'$.  
Thus we can find the edge $ab$ such that $\mathcal{L}_{ab}$ or $\mathcal{L}_{ba}$ contains $\tri$
for each triangle $\tri$ in $\mathcal{L}$ in constant time.
In the first scan, we simply append $\tri$ to one of the two sorted lists, but in the second scan,
we find the location of $\tri$ in one of the two sorted lists. The second scan can also be 
done in $O(|\mathcal{L}|)$ time since the order of apexed triangles in $\mathcal{L}_{ab}$ 
(and $\mathcal{L}_{ba}$) is the same as their order in $\mathcal{L}$.
Therefore, this procedure takes in $O(|\mathcal{L}|)$ time in total.

The following lemmas summarize this section.

\begin{lemma}
  \label{lem:gamma_cap_fvd}
  Let $C\subseteq P$ be a geodesic convex polygon and 
  ${\gamma}$ be a simple closed curve
  connecting some convex vertices of $C$ lying on $\bd P$ 
  such that two consecutive vertices in clockwise order are connected
  by a geodesic path.
  Once $\rfvd\cap\bd C$ is computed, $\rfvd\cap{\gamma}$ can be
  computed in $O(\complexity{\rfvd\cap\bd C} + \complexity{C})$ time.
\end{lemma}

\begin{lemma}
  Each iteration takes $O(n)$ time and the algorithm in this section
  terminates in $O(\log\log n)$ iterations. Thus the algorithm in this section takes $O(n\log \log n)$ time.
\end{lemma}

\section{Computing \texorpdfstring{$\rfvd$}{rFVD} in the Interior of a Base Cell}
\label{sec:third-step}
In the second step of the algorithm described in Section~\ref{sec:second-step},
we obtained a subdivision of $P$ into $O(n\log\log n)$ base cells.
Moreover, we have $\rfvd\cap \bd T$ for every such base cell $T$.
For a concave chain of $\bd T$, we define the \emph{angle-span}
of the chain as follows. While traversing the chain from one endpoint to the other,
consider the turning angle at each vertex of the chain, other than the two endpoints,
which is the angle turned at the vertex. 
The angle-span of the chain is set to the sum of the turning angles.
For a technical reason, we define the angle-span of a point as $0$.

Our goal in this section is to compute $\rfvd\cap T$ using $\rfvd \cap \bd T$ in  $O(\complexity{\rfvd\cap\bd T})$ time.
To make the description easier, we first make four assumptions: (1) $T$ is a lune-cell, and 
(2) $\rcell{\triangle}\cap \bd T$ is
connected and contains the bottom side of $\tri$ for
any apexed triangle $\triangle$ with $\rcell{\tri}\cap \bd T\neq\emptyset$. 
(3) If $\apex{\tri}$ is on $\bd T$, the closure of $\rcell{\tri}$ does not coincide
	with $\tri$.
(4) The maximal concave chain of $\bd T$ has angle-span at most $\pi/2$.
In Sections~\ref{sec:assump-1},~\ref{sec:assum-2}, and~\ref{sec:assum-4} 
we generalize the algorithm to compute $\rfvd\cap T$ without these
assumptions.

\subsection{Linear-time Algorithms for Computing Abstract Voronoi Diagrams}
We first introduce the algorithms for computing abstract Voronoi diagrams
by Klein~\cite{Klein} and Klein and Lingas~\cite{klein1994}, 
which will be used for our algorithm.
Abstract Voronoi diagrams are based on systems of simple curves~\cite{Klein}.
Let $S=\{1,\ldots,N\}$.
Each site is represented by an index in $S$. Any pair $(i,j)$ of indices in $S$ 
has a simple unbounded curve $B(i,j)$ which is called a \emph{bisecting curve}.
The bisecting curve $B(i,j)$ partitions the plane into two unbounded open domains, $D(i,j)$ and
$D(j,i)$.
Then the abstract Voronoi diagram $\avd[S]$ under the family $\{B(i,j) \mid i\neq j\in S\}$ is defined as follows.
\begin{align*}
	\aacell(i,S) &= \bigcap_{j\in S} D(i,j),\\
	\avd[S]&=\bigcup_{i\in S}\overline{\aacell}(i,S),
\end{align*}
where $\overline{A}$ is the closure of a point set $A\subseteq P$.
The abstract Voronoi diagram can be computed in $O(N\log N)$ time if the family of bisecting curves
is \emph{admissible}~\cite{Klein}.
\begin{definition}[{\cite[Definition 2.1.2]{Klein}}]
	\label{def:admissible}
	The family $\{B(i,j) \mid i\neq j\in S\}$ is called \emph{admissible} if the followings hold.
	\begin{enumerate}
		\item Given any two indices $i,j\in S$, we can obtain their bisecting curve $B(i,j)$ in constant time. (This condition is assumed implicitly in~\cite{Klein}.)
		\item The intersection of any two bisecting curves consists of finitely many connected components.
		\item For each nonempty subset $S'$ of $S$ with $|S'|\geq 3$,
		\begin{enumerate}
			\item[A.] $\aacell(i,S')$ is path-connected and has a nonempty interior, for each $i\in S'$.
			\item[B.] $\mathbb{R}^2$ is the union of $\overline{\aacell}(i,S')$ over all indices $i\in S'$.
		\end{enumerate}
	\end{enumerate}
\end{definition}

Klein and Lingas~\cite{klein1994} presented a linear-time algorithm for computing the abstract Voronoi diagram 
for an admissible family of bisecting curves if a \emph{Hamiltonian curve} of the abstract Voronoi diagram is given.
\begin{definition}[{\cite[Lemma 3 and Definition 4]{klein1994}}]
	\label{def:hamiltonian}
	The family of bisecting curves is \emph{Hamiltonian} with respect to a simple and unbounded curve $H$ if $H$ has the following properties.
	\begin{enumerate}
		\item $H$ is homeomorphic to a line.
		\item For any $S'\subseteq S$ with $|S'|\geq 2$, $\aacell(i,S')$ is visited by $H$ exactly once for every $i\in S'$.
	\end{enumerate}
	In this case, we call $H$ a \emph{Hamiltonian curve} of $\avd[S]$.
\end{definition}

Using the algorithms in~\cite{Klein,klein1994}, we can compute the nearest-point Voronoi diagrams under a variety of metrics.
However, these algorithms do not work for computing Euclidean farthest-point Voronoi diagrams
because some site may not have their (nonempty) Voronoi cells in the diagram (thus they violate 3A in Lemma~\ref{def:admissible}).
In our case, we will show that every site has a nonempty Voronoi cell, which allows us to compute $\rfvd$
using the algorithm in~\cite{klein1994}.

\subsection{New Distance Function}
\label{sec:define_func}
Recall that our goal is to compute $\rfvd\cap T$ from $\rfvd\cap \bd T$.
We cannot apply the algorithm in~\cite{klein1994} directly because the geodesic metric does not satisfy
the first condition in Definition~\ref{def:admissible}.
Thus we propose a new distance function whose corresponding system of bisecting curves satisfies  
the conditions in Definition~\ref{def:admissible} and Definition~\ref{def:hamiltonian}.

Let $\tri$ be an apexed triangle having its refined Voronoi cell on $\bd T$. Without loss of generality, we assume that 
the bottom side of $\tri$ is horizontal. 
We partition $\mathbb{R}^2$ into five regions, as depicted in Figure~\ref{fig:pseudo_dist_func}(b),
with respect to $\triangle$.  
Consider five halflines $\ell_1$, $\ell_2$, $\ell_3$, $\ell_4$ and
$\ell_5$ starting from $\apex{\tri}$ as follows.
The halflines $\ell_1$ and $\ell_2$ go towards the left and the right corners of $\tri$,
respectively.  The halflines $\ell_3$ and $\ell_5$ are orthogonal
to $\ell_2$ and $\ell_1$, respectively.  
The halfline $\ell_4$ bisects the angle of $\triangle$ at $\apex{\triangle}$
but does not intersect $\interior{\triangle}$.

Consider the region partitioned by the five halflines.
We denote the region bounded by $\ell_1$ and $\ell_2$
that contains $\triangle$ by $\inregion{\tri}$. 
The remaining four
regions are denoted by $\lsideregion{\tri}$, $\ltopregion{\tri}$,
$\rtopregion{\tri}$, and $\rsideregion{\tri}$ in the clockwise order
from $\inregion{\tri}$ around $\apex{\tri}$.

For a point $x \in \lsideregion{\tri} \cup \ltopregion{\tri}$, let
$\hat{x}_\triangle$ denote the orthogonal projection of $x$ on the line containing $\ell_1$.
 Similarly, for a point $x \in
\rsideregion{\tri} \cup \rtopregion{\tri} \setminus \ell_4$, 
let $\hat{x}_\triangle$
denote the orthogonal projection of $x$ on the line containing $\ell_2$.
For a point $x \in \inregion{\tri}$, we set
$\hat{x}_\triangle = x$.
When $\triangle$ is clear in the context, we simply use $\hat{x}$ to denote
$\hat{x}_\triangle$.

We define a new distance function $f_\triangle$ : $\mathbb{R}^2 \rightarrow
\mathbb{R}$ for each apexed triangle $\triangle$ with
$\rcell{\triangle}\cap\bd T\neq\emptyset$ as follows.

\begin{align*}
  f_\triangle(x) &=
  \begin{cases}
    d(\apex{\triangle},\definer{\triangle})-\|\hat{x}_\tri-\apex{\triangle}\| & \quad \text{if } x \in \ltopregion{\tri}\cup\rtopregion{\tri},\\
    d(\apex{\triangle},\definer{\triangle})+\|\hat{x}_\tri-\apex{\triangle}\|
    & \quad \text{otherwise,}
  \end{cases}
\end{align*}
where $\|x-y\|$ denote the Euclidean distance between $x$ and $y$.
Note that $f_\triangle$ is continuous.  Each contour curve,
that is a set of points with the same function value, consists of two
line segments and at most one circular arc. See
Figure~\ref{fig:pseudo_dist_func}(c). 

Here, we assume that there is no pair
$(\tri_1,\tri_2)$ of apexed triangles such that two sides, one from
$\tri_1$ and the other from $\tri_2$, are parallel.  If there exists such a
pair, contour curves for two apexed triangles may overlap.  
We will show how to avoid this assumption in Section~\ref{sec:assum-3} 
by slightly perturbing the distance function defined in this section.

 By the definition of
 $f_\triangle$, the following lemma holds.
\begin{lemma}
  \label{lem:max_pseudo_dist}
  The difference of $f_\triangle(x_1)$ and $f_\triangle(x_2)$
  is less than or equal to $\|x_1-x_2\|$ for any two
  points $x_1,x_2 \in \mathbb{R}^2$, where $\|x-y\|$ is the Euclidean distance
  between $x$ and $y$.
\end{lemma}

\begin{figure}
  \begin{center}
    \includegraphics[width=0.75\textwidth]{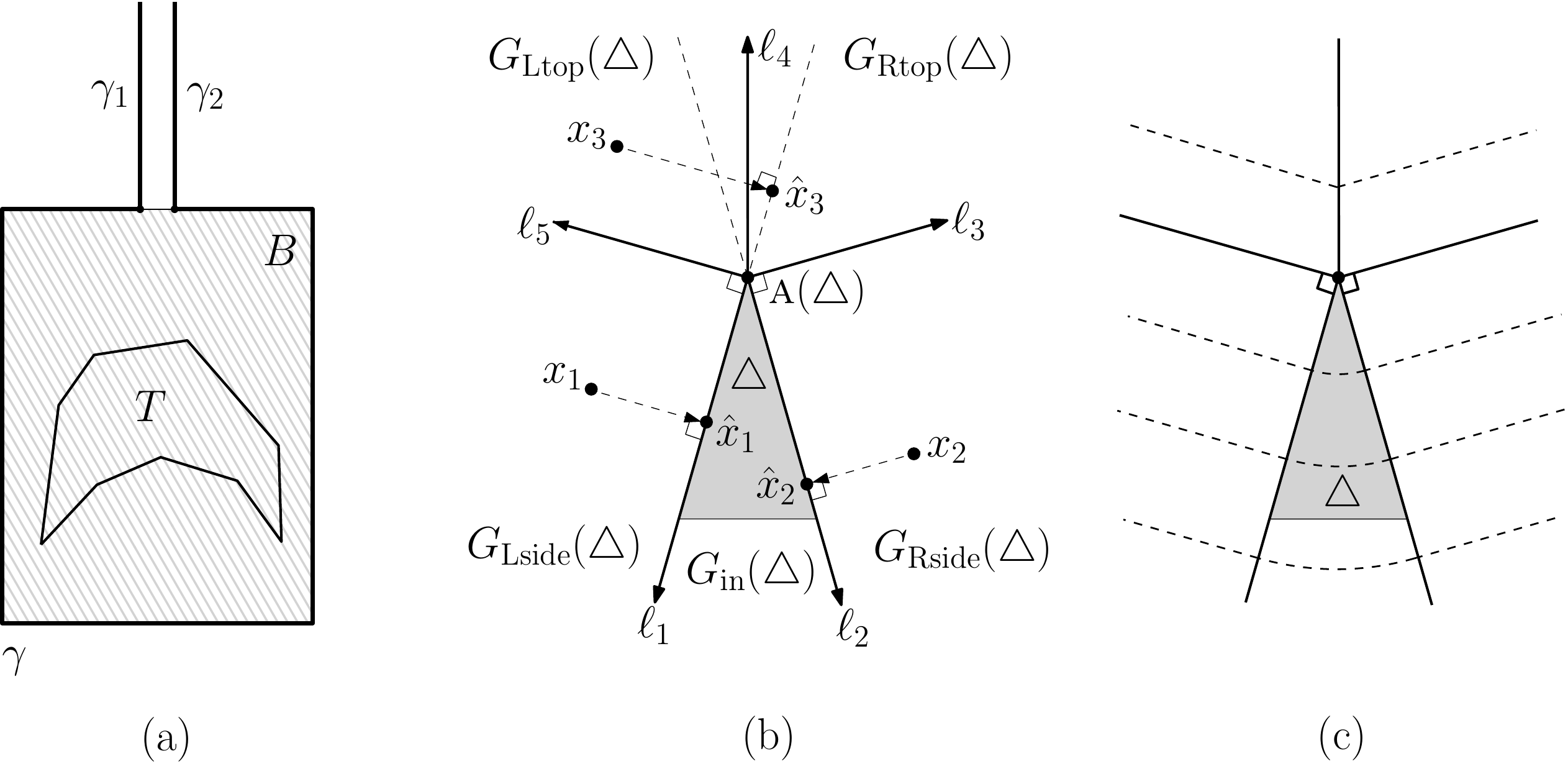}
    \caption {\small (a) The thick black curve $\gamma$ is a Hamiltonian curve of the abstract Voronoi diagram. 
    	(b) Five regions are defined by the five halflines 
      from $\ell_1$ to $\ell_5$. The gray triangle is $\tri$. (c) The dashed curves are 
      contour curves with respect to $f_\tri$.}
    \label{fig:pseudo_dist_func}
  \end{center}
\end{figure}

\subsection{Algorithm for Computing \texorpdfstring{$\rfvd\cap T$}{rFVD cap T}}
\label{sec:computing_abstract}
To compute the geodesic farthest-point Voronoi diagram restricted to
$T$, we apply the algorithm in \cite{klein1994} that computes the abstract Voronoi diagram. 
Let $A$ be the set of all apexed triangles having their refined Voronoi cells on $\bd T$.
In our problem,
we regard the apexed triangles in $A$ as the sites.
 For two apexed triangles
$\triangle_1$ and $\triangle_2$ in $A$, we define the bisecting curve 
$B(\triangle_1,\triangle_2)$ as the set $\{x \in \mathbb{R}^2 :
f_{\triangle_1}(x) = f_{\triangle_2}(x)\}$.
The bisecting curve partitions $\mathbb{R}^2$ into two regions 
$\pseudohalf{\tri_1}{\tri_2}$ and $\pseudohalf{\tri_2}{\tri_1}$
such that 
$f_{\triangle_1}(x) > f_{\triangle_2}(x)$ for $x\in \pseudohalf{\tri_1}{\tri_2}$ and
$f_{\triangle_2}(x) > f_{\triangle_1}(x)$ for $x\in \pseudohalf{\tri_2}{\tri_1}$.
We denote the abstract Voronoi diagram for the apexed triangles by
$\avd$ and the cell of $\triangle\in A$ on $\avd$ by $\acell{\triangle}$.

To apply the algorithm in~\cite{klein1994}, we show that the family of the bisecting curves is
admissible and Hamiltonian in the following subsection.
We also prove that $\avd\cap T$ is exactly $\rfvd\cap T$.
After computing $\avd$, we traverse $\avd$ and extract $\avd$ lying inside $T$.
This takes $O(|\rfvd \cap\bd T|)$ time since no refined cell $\rcell{\tri}$ contains
a vertex of $T$ in its interior.

In addition, to apply the algorithm in~\cite{klein1994}, we have to choose a Hamiltonian curve $\gamma$.
This algorithm requires $\gamma\cap\avd$ to be given.
To do this, we first choose an arbitrary box $B$ containing $T$. 
We compute one Voronoi cell of $\avd$ directly in $O(|\rfvd\cap\bd T|)$ time by considering all
apexed triangles in $A$.
We also choose two arbitrary curves $\gamma_1$ and $\gamma_2$ with endpoints on the same edge of $\bd B$ 
which are contained in the Voronoi cell.
See Figure~\ref{fig:pseudo_dist_func}(a).
Then we can compute $\gamma$ consisting of $\gamma_1, \gamma_2$ and a part of $\bd B$ such that
$\gamma$ contains the four corners of $B$.
Note that $\gamma$ is homeomorphic to a line.
We will see that the order of the refined Voronoi cells along $\bd T$ coincides with
the order of the Voronoi cells along $\bd B$ in $\avd$ in Corollary~\ref{cor:ordering-aVD}.
Therefore, we can obtain $\avd\cap\gamma$ in $O(|\rfvd \cap\bd T|)$ time
once we have $\rfvd\cap\bd T=\avd\cap\bd T$.

\subsection{Properties of Bisecting Curves and Voronoi Diagrams}

\subsubsection{\texorpdfstring{$\avd\cap T$ Coincides with $\rfvd\cap T$}{aFVD cap T coincides with rFVD cap T}}
Recall that $T$ is a lune-cell, which is bounded by a convex
chain and a concave chain. Also, recall that the bottom side of every apexed triangle of $A$
is contained in $\bd T$.
The following technical lemmas are used to prove that
$\rcell{\tri}\cap T$ coincides with $\acell{\tri}\cap T$
for any apexed triangle $\tri \in A$.
	For a halfline $\ell$, we let $\bar{\ell}$ be the directed line
containing $\ell$ with the same direction as $\ell$. 

\begin{figure}
  \begin{center}
     \includegraphics[width=0.8\textwidth]{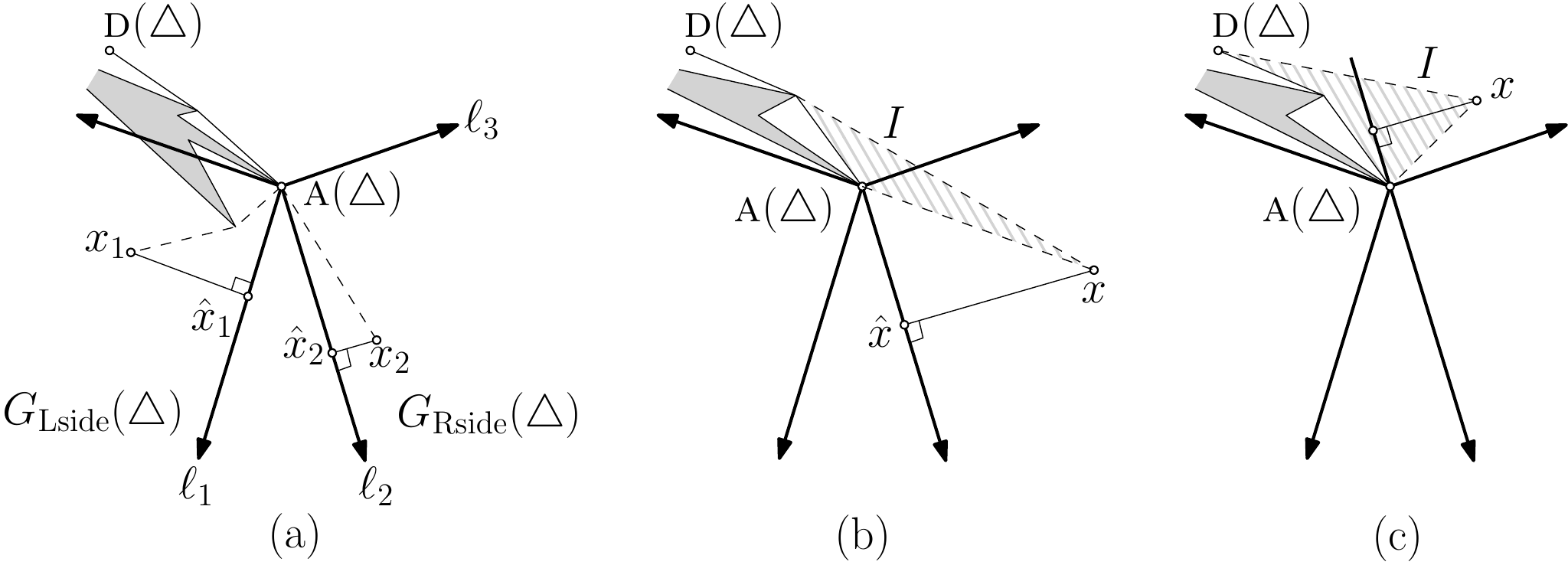}
    \caption {\small (a) If $\pi(x,\definer{\tri})$ contains $\apex{\tri}$, we have
    	$g_\tri (x)\leq d(\definer{\tri},x)$. (b) Since $x$ lies in $\rsideregion{\tri}$, the
    	angle at $\apex{\tri}$ is at least $\pi/2$, and thus $f_\tri (x)< d(\definer{\tri},x)$.
    	(c) By triangle inequality, $f_\tri(x) < d(\definer{\tri},\hat{x})$. Also, we have
    	$d(\definer{\tri},\hat{x})< d(\definer{\tri},x)$. Therefore 
    	we have $f_\tri (x)< d(\definer{\tri},x)$.}
    \label{fig:apexed_trinagle_intersect}
  \end{center}
\end{figure}

\begin{lemma}\label{lem:contain-vertex}
	For any apexed triangle $\tri$ of $A$, we have
	$f_\tri(x)< d(\definer{\tri},x)$ for any point $x$ such that
	$\pi(\definer{\tri},x)$ contains $x'$ with $f_\tri(x')< d(\definer{\tri},x')$. 
\end{lemma}
\begin{proof}
	By definition, $d(\definer{\tri},x)=d(\definer{\tri},x')+d(x,x')$. Also,
	$f_\tri(x) \leq f_\tri(x')+ \|x-x'\|$ by Lemma~\ref{lem:max_pseudo_dist}.
	Thus $f_\tri(x) < d(\definer{\tri},x') + d(x,x') = d(\definer{\tri}, x)$.
\end{proof}

\begin{lemma}\label{lem:contain-apex}
	For any apexed triangle $\tri$ of $A$, we have
	$f_\tri(x)\leq d(\definer{\tri},x)$ for any point $x$ such that
	$\pi(\definer{\tri},x)$ contains $\apex{\tri}$. The equality holds 
	if and only if $x$ lies in $\tri$.
\end{lemma}
\begin{proof}
	If $x\in\tri$, the lemma holds immediately. Thus we assume that $x$ is not in $\tri$ and
	show that $f_\tri(x)<d(\definer{\tri},x)$.
	The Euclidean distance between $\apex{\tri}$ and $\hat{x}$ is less than the
	Euclidean distance between $\apex{\tri}$ and $x$. Since $d(\apex{\tri},x)$ is at least
	their Euclidean distance, we have $f_\tri(x)< d(\definer{\tri},x)$. 
	See Figure~\ref{fig:apexed_trinagle_intersect}(a).
\end{proof}

\begin{lemma}
  For an apexed triangle $\tri$ of $A$ such that the edge of $\pi(\apex{\tri},\definer{\tri})$ incident to
  $\apex{\tri}$ lies to the left of $\bar{\ell}_4$, we have 
  $f_\triangle(x) \leq d(\definer{\tri},x)$ for any point $x \in T\setminus \ltopregion{\tri}$.
  The equality holds if and only if  $x$ lies in $\triangle$.
  \label{lem:pseudo_dist_farthest_case1}
\end{lemma}
\begin{proof}
	Let $x$ be a point in $T\setminus \ltopregion{\tri}$. 
	If $\pi(\definer{\tri},x)$ contains $\apex{\tri}$, the lemma holds by Lemma~\ref{lem:contain-apex}.
	Thus we assume that $\pi(\definer{\tri},x)$ does not contain $\apex{\tri}$. Then 
	$x$ does not lie in $\lsideregion{\tri}$ since the maximal concave curve of $\bd T$ has 
	angle-span at most $\pi/2$ by the assumption made in the beginning of this section
	and the angle at $\apex{\tri}$ in $\lsideregion{\tri}$ is $\pi/2$. See Figure~\ref{fig:apexed_trinagle_intersect}(a). 
	By construction of the apexed triangles and the assumption of the lemma, 
	the edge of $\pi(\apex{\tri},\definer{\tri})$ incident to
	$\apex{\tri}$ lies to the right of $\bar{\ell}_2$. Consider the \emph{interior} $I$ of the geodesic convex hull of $\definer{\tri}$, $\apex{\tri}$ and $x$. Note that $\apex{\tri}$ is a vertex of $I$.

	Consider the case that $x$ is a vertex of $I$. Then $x$ lies in $\rsideregion{\tri}\cup\rtopregion{\tri}$.
	If $x$ lies in $\rsideregion{\tri}$,
	the angle at $\apex{\tri}$ with respect to $I$ is at least $\pi/2$.
	Thus $d(\definer{\tri}, \hat{x})$ is at most $d(\definer{\tri},x)$, and thus the lemma holds 
	for this case. See Figure~\ref{fig:apexed_trinagle_intersect}(b).
	If $x$ lies in $\rtopregion{\tri}$,
	we know that $f_{\tri}(x)< d(\definer{\tri},\hat{x})$ by triangle
	inequality, and $d(\definer{\tri},\hat{x})<
	d(\definer{\tri},x)$. Thus the lemma holds for this case. See Figure~\ref{fig:apexed_trinagle_intersect}(c). 
	
	Now consider the case that $x$ is not a vertex of $I$. Let $x'$ be the vertex of $I$ contained in
	$\pi(\apex{\tri},x)$ which is closest to $x$. 
	We can prove that $f_\tri(x')< d(\definer{\tri},x')$ as we did for the previous cases that
	$x$ is a vertex of $I$ since $x'$ is a vertex of $I$. Then the lemma holds by Lemma~\ref{lem:contain-vertex}.
\end{proof}

\begin{lemma}
	For an apexed triangle $\tri$ of $A$ such that $\pi(\apex{\tri},\definer{\tri})$ does not overlap
	with the maximal concave curve of $\bd T$, 
	we have 
	$f_\triangle(x) \leq d(\definer{\tri},x)$ for any point $x \in T$.
	The equality holds if and only if  $x$ lies in $\triangle$.
	\label{lem:pseudo_dist_farthest_case2}
\end{lemma}
\begin{proof}
	Without loss of generality, we assume that the edge of $\pi(\apex{\tri},\definer{\tri})$ incident to
	$\apex{\tri}$ lies to the left of $\bar{\ell}_4$. By Lemma~\ref{lem:pseudo_dist_farthest_case1},
	the lemma holds, except for a point $x$ in $\ltopregion{\tri}$.
	Since  $\pi(\apex{\tri},\definer{\tri})$ does not overlap
	with the maximal concave curve of $\bd T$, 
	$\pi(x,\definer{\tri})$ contains $\apex{\tri}$ or a vertex $x'$ of the maximal concave curve for any point $x$ in $\ltopregion{\tri}$. If $\pi(x,\definer{\tri})$ contains $\apex{\tri}$,
	the lemma holds by Lemma~\ref{lem:contain-apex}. Otherwise, $x'$ lies on $\rtopregion{\tri}\cup\rsideregion{\tri}$. Thus we have $f_\tri(x')< d(\definer{\tri},x')$.
	Therefore the lemma holds by Lemma~\ref{lem:contain-vertex}.
\end{proof}

The following lemma implies that $\avd\cap T$, that is
the abstract Voronoi diagram of $A$ restricted to $T$, 
coincides with $\rfvd \cap T$.
\begin{lemma}\label{lem:acell}
  For an apexed triangle $\triangle$ and a point $x \in T \cap
  \rcell{\triangle}$, $x$ lies in $\acell{\triangle}$.
\end{lemma}
\begin{proof}
  Assume to the contrary that there are an apexed triangle $\tri\in A$ and a
  point $x \in T \cap \rcell{\triangle}$ such that $x \notin
  \acell{\triangle}$.  This means that there is another apexed triangle
  $\triangle'$ such that $f_\triangle(x) \leq
  f_{\triangle'}(x)$.  Among all such apexed triangles, we choose the
  one with the maximum $f_{\triangle'}(x)$.
  Without loss of generality, we assume that the edge of $\pi(\apex{\tri'},\definer{\tri'})$ incident to
  $\apex{\tri'}$ lies to the left of $\bar{\ell}_4$.
  We claim that $x$ lies in $\ltopregion{\tri'}$ and $\pi(\definer{\tri'},\apex{\tri'})$ overlaps 
  with the maximal concave curve of $\bd T$, . Otherwise, we have 
  $f_{\triangle'}(x) \leq d(\definer{\tri'},x)$ by Lemmas~\ref{lem:pseudo_dist_farthest_case1}
  and~\ref{lem:pseudo_dist_farthest_case2}.  
  By definition, we have $f_{\triangle'}(x) \leq
  d(\definer{\tri'},x) < d(\definer{\tri},x) = f_\tri(x)$, which is a
  contradiction.  
	\begin{figure}
		\begin{center}
			\includegraphics[width=0.3\textwidth]{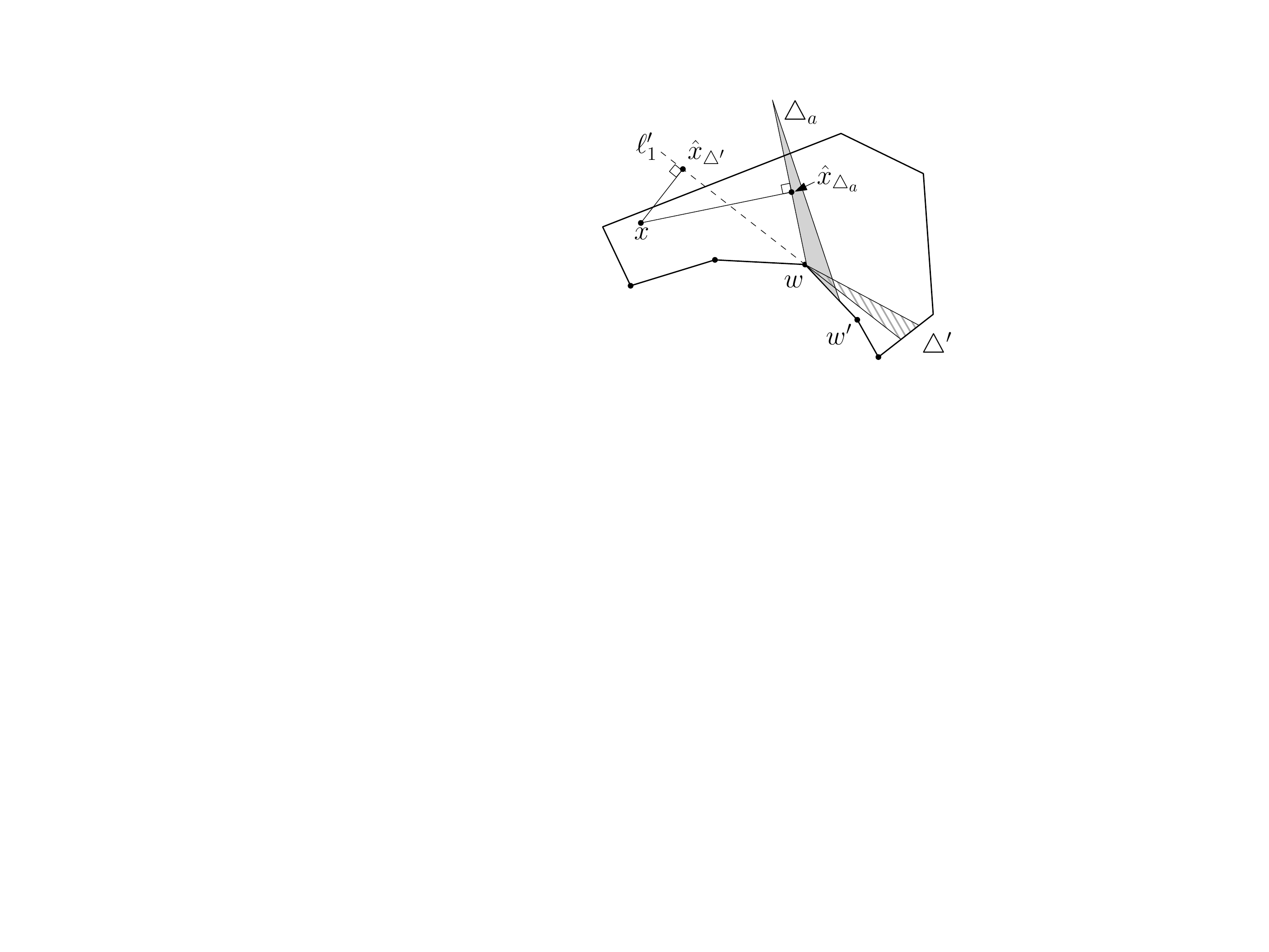}
			\caption {\small For $x$ lying in $T\cap \ltopregion{\tri'}$,
				$\|\hat{x}_{\tri_a}-w\|<\|\hat{x}_{\triangle'}-w\|$, and thus
				$f_{\tri_a}(x)>f_{\tri'}(x)$. }
			\label{fig:apexed_trinagle_intersect2}
		\end{center}
	\end{figure}

  In the following, we show that there is another apexed
  triangle $\tri_a$ such that $f_{\triangle'}(x) <
  f_{\triangle_a}(x)$.  This is a contradiction as we chose the apexed
  triangle $\tri'$ with maximum $f_{\tri'}(x)$. Recall that we assume in the beginning of this 
  section that $\tri'$ does not coincide with the closure of ${\rcell{\tri'}}$.  
  Let $w$ be $\apex{\tri'}$ and $w'$ be the clockwise neighbor of $w$ along $\bd T$.  See
  Figure~\ref{fig:apexed_trinagle_intersect2}.  In this case, there
  is another apexed triangle $\triangle_a$ such that
  $d(\definer{\tri_a},w) > d(\definer{\tri'},w)$ and the bottom side of $\tri_a$ is
  contained in $ww'$. Otherwise, $w$ is in $\rcell{\tri'}$, and thus $\tri'$ coincides with 
  the closure of $\rcell{\tri'}$, which is a contradiction.  
  Let $\ell_1'$ be the line containing the side of
  $\tri'$ which is closer to $w'$ other than its bottom side.
  The point $\hat{x}_{\tri_a}$ lies on the line passing through $w$ and $\apex{\tri_a}$. 
  If it lies on the halfline starting from $w$ in direction opposite to $\apex{\tri_a}$, the  
  claim holds immediately. Thus we assume that it lies on the halfline starting from $w$ in direction to
$\apex{\tri_a}$.
   Then $x$ lies in the side of $\ell_1'$ containing $w'$ since $x\in \ltopregion{\tri}$. 
  Moreover, $\hat{x}_{\tri_a}$ and $x$ lie in different sides of
  $\ell_1'$ since
  $\tri_a$ has its bottom side on the line containing $ww'$.
  Therefore, we have
  $\|w-\hat{x}_{\tri_a}\|\leq\|w-\hat{x}_{\triangle'}\|$.  This implies that 
  $f_{\tri_a}(x) = d(\definer{\tri_a},w)-\|w-\hat{x}_{\tri_a}\| \geq
  d(\definer{\tri_a},w)-\|w-\hat{x}_{\triangle'}\| >
  d(\definer{\tri'},w)-\|w-\hat{x}_{\triangle'}\|=
  f_{\tri'}(x)$, which is a
  contradiction.
\end{proof}

\begin{corollary}
  \label{lem:projection_fvd}
  The abstract Voronoi diagram with respect to the functions $f_\tri$
  restricted to $T$
  for all apexed triangles $\tri \in A$ coincides with
  the refined geodesic farthest-point
  Voronoi diagram restricted to $T$.
\end{corollary}

\subsubsection{The Family of Bisecting Curves is Admissible}
Conditions~1 and~2 of Definition~\ref{def:admissible} hold due to Lemma~\ref{lem:pseudo_bisector}.
Condition~3A holds due to Lemmas~\ref{lem:condition-connected} and~\ref{lem:condition-nonempty}.
Condition~3B holds by the definition of the new distance function.
	
\begin{lemma}
 For an apexed triangle $\triangle$ in any subset $A'$ of $A$, $\aacell(\tri, A')$ is connected.
 \label{lem:condition-connected}
\end{lemma}
\begin{proof}
  Here we use $\aacell(\tri)$ to denote $\aacell(\tri,A')$ for an apexed triangle $\tri\in A'$. 
  By Corollary~\ref{lem:projection_fvd}, we have $\rcell{\triangle}\subseteq\acell{\tri}$.
  Assume to the contrary that there are at least two connected components of $\acell{\tri}$.
  Note that one of them contains $\rcell{\tri}$.
  
  	\begin{figure}
  	\begin{center}
  		\includegraphics[width=0.85\textwidth]{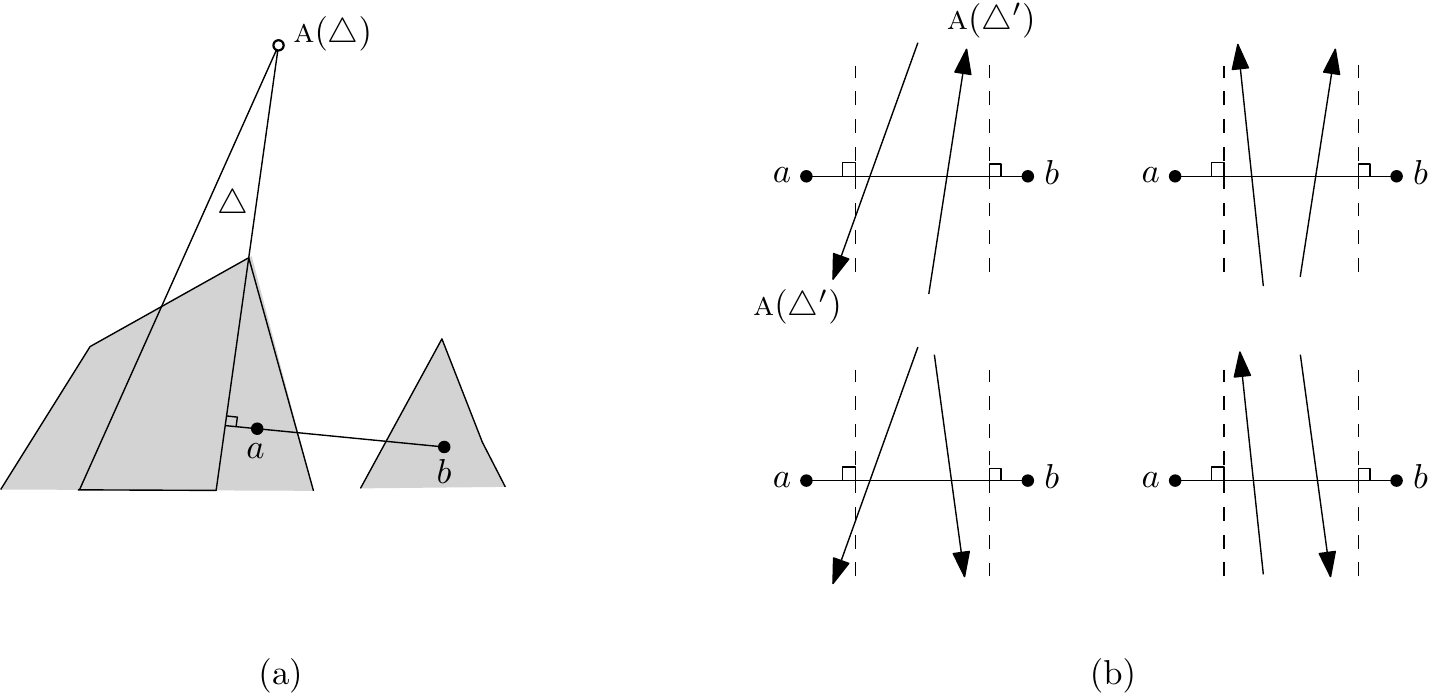}
  		\caption {\small (a) There are two points $a$ and $b$ in different connected components
  		of $\acell{\tri}$ (the gray region) such that $f_\tri(a)=f_\tri(b)$.
  		(b) The arrows denote the position of $\apex{\tri'}$ so that
  	the functions $f_{\tri'}$ with
  	domain $h_1$ is increasing, and $f_{\tri'}$ with domain $h_3$ is decreasing.}
  		\label{fig:abstract-connected}
  	\end{center}
  \end{figure}

  For any point $x\in \acell{\tri}$, there is a halfline from $x$ such that
    every point in the halfline, except $x$, has distance value $f_\tri(\cdot)$ larger than
    $f_\tri(x)$. 
  This can be shown in a way similar to
  Lemma~\ref{lem:ray_in_cell} together with  Lemma~\ref{lem:max_pseudo_dist}.
  Thus there are two points $a$ and $b$ from different connected components of $\acell{\tri}$
  such that $f_\tri(a)=f_\tri(b)$. See Figure~\ref{fig:abstract-connected}(a).
  
  Since $a$ and $b$ are in different connected components of $\acell{\tri}$, there is another
  triangle $\tri'$ in $A'$ such that $f_{\tri}(x)< f_{\tri'}(x)$ for some point $x\in ab$.
  Consider $f_{\tri'}$ restricted to the domain $ab$. Here, to make the description easier,
  we consider $ab$ as a line segment on $\mathbb{R}$, and each point on $ab$ as a real number. 
  Without loss of generality, we assume that $a$ is smaller than $b$.
  If $\inregion{\tri'}$ does not intersect $ab$, the function is linear in $ab$, and thus
  $f_{\tri}(x)\geq  f_{\tri'}(x)$ for any point $x\in ab$. This is a contradiction.
  Thus we assume that $\inregion{\tri'}$ intersects $ab$.

  We claim that $ab \setminus \inregion{\tri'}$ consists of two connected components. 
  If $\inregion{\tri'}$ contains $ab$, the function $f_{\tri'}$ restricted to $ab$ is
  convex. Thus $f_{\tri}(x)\geq f_{\tri'}(x)$ for any point $x\in ab$.
  If $\inregion{\tri'}$ contains only one endpoint of $ab$, 
  the function $f_{\tri'}$ restricted to $ab$ increases or decreases.
  Thus $f_{\tri}(x)\geq f_{\tri'}(x)$ for any point $x\in ab$. Therefore, the claim holds.
  
  Let $h_1$ and $h_3$ denote the connected components of $ab \setminus \inregion{\tri'}$ containing $a$ and $b$, respectively.
  The function $f_{\tri'}$ with
  domain $h_1$ is increasing, and $f_{\tri'}$ with domain $h_3$ is decreasing. However, it is not possible. To see this, see Figure~\ref{fig:abstract-connected}(b). 
  There are four cases on the sides of $\tri'$: the clockwise angle from each side of $\tri'$
  to $ab$ is at least $\pi/2$ or not. The arrows denote the position of $\apex{\tri'}$ so that
  the function $f_{\tri'}$ with
  domain $h_1$ is increasing, and $f_{\tri'}$ with domain $h_3$ is decreasing.
  Each of the four cases makes a contradiction. Therefore, the lemma holds.
\end{proof}

\begin{lemma}
	\label{lem:condition-nonempty}
	For an apexed triangle $\triangle$ in any subset $A'$ of $A$,
        $\aacell(\tri,A')$ is nonempty.
\end{lemma}
\begin{proof}
	Every apexed triangle $\tri$ in $A$ has a refined Voronoi cell in $T$.
	Since $\aacell(\tri,A')$ contains $\rcell{\tri}$, it is nonempty.
\end{proof}
\begin{lemma}
  \label{lem:pseudo_bisector}
  For any two apexed triangles $\triangle_1$ and $\triangle_2$ in $A$, the
  set $\{x \in \mathbb{R}^2 : f_{\triangle_1}(x) = f_{\triangle_2}(x)\}$ is a
  curve consisting of $O(1)$ algebraic curves. 
\end{lemma}
\begin{proof}
  For any two apexed triangles $\tri_1$ and $\tri_2$, the set $\{x \in \mathbb{R}^2 :
  f_{\triangle_1}(x) = f_{\triangle_2}(x)\}$ is a curve homeomorphic to a line
  by Lemma~\ref{lem:condition-connected}.
  Consider the subdivision of $\mathbb{R}^2$ by overlaying the two
  subdivisions as depicted in Figure~\ref{fig:pseudo_dist_func}(b), 
  one from $\triangle_1$ and the other from $\triangle_2$.
  There are at most nine cells in the subdivision of $\mathbb{R}^2$.  In each
  cell, $f_{\triangle_1}$ and $f_{\triangle_2}$ are algebraic
  functions.  Thus, the set $\{x \in C : f_{\triangle_1}(x) =
  f_{\triangle_2}(x)\}$ is an algebraic curve for each cell $C$.
\end{proof}

\subsubsection{The Family of Bisecting Curves is Hamiltonian}
We show that $\gamma$ is a Hamiltonian curve. Recall that $B$ is a box containing $T$ and
$\gamma$ is a curve that contains a part of $\bd B$ and is homeomorphic to a line.
See Figure~\ref{fig:pseudo_dist_func}(a).

\begin{lemma}
  \label{lem:assumption_a4}
  For every subset $A'\subseteq A$ and $\tri\in A'$, $\aacell(\tri,A')$ is intersected by $\gamma$ exactly once.
\end{lemma}
\begin{proof}
  We claim that $\aacell(\triangle)$ is incident to $\bd B$ for any apexed triangle $\triangle \in A'$.
  Since the halfline from a point $x$ in $\rcell{\tri}$ in direction opposite to $\apex{\tri}$ is
  contained in $\rcell{\tri}$  
  and the bottom side of $\tri$ is contained in $\rcell{\tri}$,
  the region $\inregion{\tri}\setminus \tri$ is contained in
  $\rcell{\triangle}$, and thus is contained in $\acell{\tri}$.  
  Since $\inregion{\tri}\setminus \tri$
  intersects $\bd B$, all $\acell{\triangle}$ are incident
  to $\bd B$.
  
  We claim that $\aacell(\tri)$ is intersected by $\bd B$ exactly once for every apexed triangle $\tri \in A'$. Otherwise, $\rcell{\tri}$ intersects the boundary of $\bd T$ more than once since
  $\aacell(\tri')$ is connected and contains $\rcell{\tri'}$ for every apexed triangle $\tri'\in A'$. This contradicts the 
  assumption made in the beginning of this section: $\rcell{\triangle}\cap \bd T$ is
  connected. Therefore, the claim holds.
  
  For the part of $\gamma$ not contained in $\bd B$, recall that   
  we chose $\gamma$ such that $\gamma\setminus \bd B$ is contained in $\aacell(\tri')$ for some $\tri' \in A'$.
  Therefore, $\aacell(\tri)$ is intersected by $\gamma$ exactly once.
\end{proof}
\begin{corollary}\label{cor:ordering-aVD}
	The order of $\rcell{\cdot}$ along $\bd T$ coincides with the order of $\aacell(\cdot,A)$ along $\bd B$.
\end{corollary}
\subsection{Dealing with Cases Violating the Assumptions} 
\label{sec:avoiding-assumption}
In the previous subsections, we made the following four assumptions.
Note that the last one is made in Subsection~\ref{sec:define_func}
for defining the distance function $f_\tri(\cdot)$.
\begin{enumerate}
	\item  $T$ is a lune-cell.
	\item  $\rcell{\triangle}\cap \bd T$ is
connected and contains the bottom side of $\tri$ for
any apexed triangle $\triangle$ with $\rcell{\tri}\cap \bd T\neq\emptyset$.
	\item  If $\apex{\tri}$ is on $\bd T$, the closure of $\rcell{\tri}$ does not coincide
	with $\tri$.
	\item The maximal concave chain of $\bd T$ has angle-span at most $\pi/2$.
	\item  There is no pair $(\tri_1,\tri_2)$ of apexed triangles of $A$ such that two sides, one from
	$\tri_1$ and the other from $\tri_2$, are parallel.
\end{enumerate}

\begin{figure}
  \begin{center}
    \includegraphics[width=0.55\textwidth]{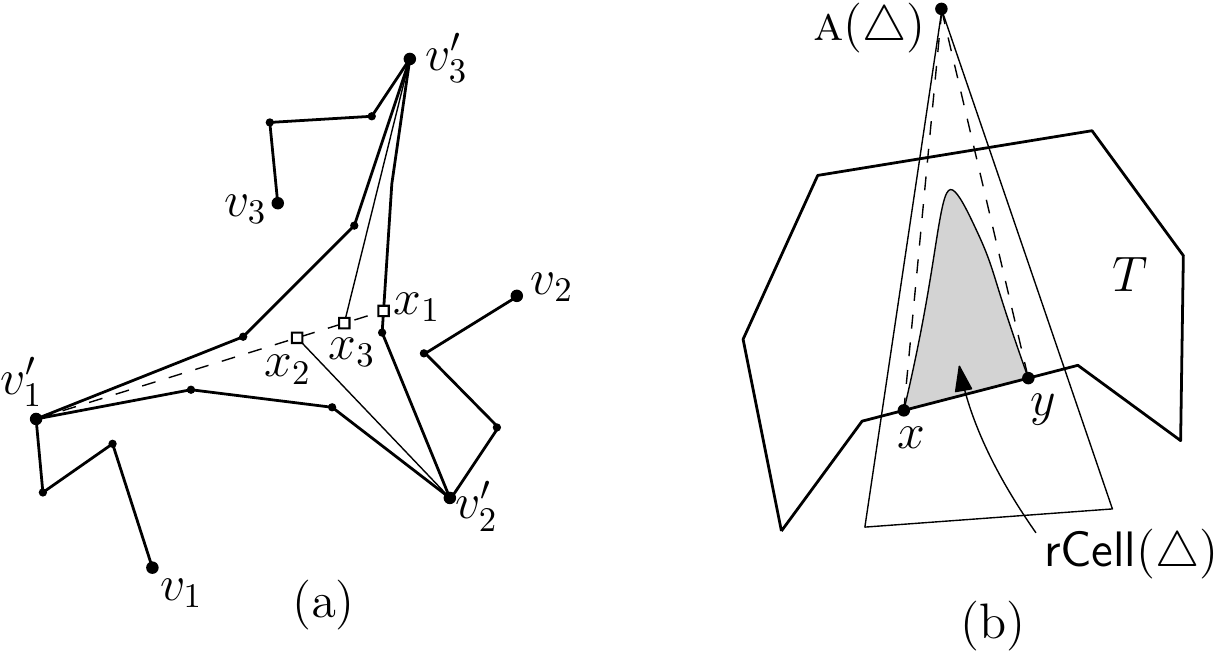}
    \caption {\small (a) The pseudo-triangle is subdivided into
      interior-disjoint four lune-cells. 
      (b) We trim an apexed triangle $\tri$ so that the bottom side of $\tri$ is contained in
      $\rcell{\tri} \cap \bd T$.
      }
    \label{fig:subdivide_pseudotriangle}
  \end{center}
\end{figure}

\subsubsection{Satisfying Assumption~1}\label{sec:assump-1}
  To satisfy Assumption~1, we subdivide each base cell further into subcells so that
  each subcell satisfies Assumption~1.
  For a pseudo-triangle $T$, we subdivide $T$ into four subcells as depicted
  in Figure~\ref{fig:subdivide_pseudotriangle}(a).  Let $v_1, v_2$ and 
  $v_3$ be three corners of $T$.  
  Consider the three vertices $v_1',v_2',v_3'$ of $T$ such that the maximal common
  path for $\pi(v_i,v_j)$ and $\pi(v_i,v_k)$ is $\pi(v_i,v_i')$ for
  $i=1,2,3$, where $j$ and $k$ are two distinct indices other than
  $i$.
  
  First, we find a line segment $v_1'x_1\subset T$ such that $v_1'x_1
  \cap \bd T = \{v_1',x_1\}$.  Then, we find two line segments
  $v_i'x_i \subset T$ such that $v_i'x_i \cap \bd T = \{v_i'\}$ and $x_i \in
  v_1'x_1$ for $i=2,3$.  This takes $O(|T|)=O(\complexity{
  	\rfvd\cap\bd T})$ time.  Then the three line segments $v_i'x_i$
  subdivide $T$ into four lune-cells $T_1, T_2, T_3$ and $T_4$ for $i=1,2,3$.  
  Note that to apply the
  algorithm in this section, $\rfvd \cap \bd T_j$ must be given. It
  can be computed in $O(\complexity{\rfvd\cap\bd T})$ time by
  Corollary~\ref{lem:const_gamma}.  
  Moreover, the total complexity
  of $\rfvd\cap\bd T_j$ for $j=1,2,3,4$
  is $O(\complexity{\rfvd\cap \bd T})$.
  Then we handle each lune-cell separately.
  Now every base cell is a lune-cell.
  
  \subsubsection{Satisfying Assumptions~2 and~3}\label{sec:assum-2}
  We first subdivide each lune-cell $T$ further to satisfy the first part of Assumption~2 and Assumption~3 using a set of line segments with both endpoints on $\bd T$ as follows.
For each endpoint $a$ of each connected component of $\rcell{\tri}\cap \bd T$ for an apexed
	triangle $\tri$ with $\rcell{\tri}\cap \bd T\neq\emptyset$, we compute
the ray from $a$ in direction opposite to $\apex{\tri}$ that intersects
	$T$, it it exists, as we did in Phase~2 of 
	the subdivision. By Lemma~\ref{lem:ray_in_cell}, each such ray is contained
	in the refined cell of its corresponding apexed triangle.
Let $\mathcal{R}$ be the set of all such rays. 
If $\apex{\tri}$ is in $\bd T$ and its $S$-farthest neighbor is $\definer{\tri}$ for some
apexed triangle $\tri$, we also add the sides of $\tri$ other than its bottom side to $\mathcal{R}$.
We can obtain $\mathcal{R}$ in  $O(\complexity{\rfvd\cap\bd T}+\complexity{T})$ time
as we did in Phase~2. 

Then we subdivide $T$ with respect to the rays in $\mathcal{R}$. Since no ray in $\mathcal{R}$
intersects an arc of $\rfvd$ in $T$, the sum of
$O(\complexity{\rfvd\cap\bd T'}+\complexity{T'})$ for all
subcells $T'$ of $T$ is $O(\complexity{\rfvd\cap\bd T})$.  
We can compute this subdivision in $O(\complexity{\rfvd\cap\bd T}+\complexity{T})$ time.
Moreover, each subcell is a lune-cell.  
The refined cell of every apexed triangle appears on each subcell at most once, and thus
the first part of Assumption~2 is satisfied.
Also, if $\apex{\tri}$ is on $T$ and the closure of $\rcell{\tri}$ is $\tri$,
we know that only one subcell $T'$ intersects the interior of $\tri$, and it coincides
with $\tri$. Thus we already know $\rfvd$ restricted to $T'$, which is simply
$\rcell{\tri}$. Thus we do not need to apply the algorithm in this section.
Therefore, every subcell $T'$ such that we do not have $\rfvd$ restricted to $T'$ yet
satisfies Assumption~3.
Also, these subcells still satisfy Assumption~1.

Then we decompose $\tri$ into two triangles by the line passing through $\apex{\tri}$ and $v$
for every $\tri$ such that $\rcell{\tri}\cap\bd T'$ contains a convex vertex $v$ of a subcell
$T'$ (there are at most two such vertices). Note that $v$ lies in the interior of $P$ in this case. 
Only one of the triangles intersects the interior of $T'$ by the definition of $T'$.
We replace $\tri$ with the triangle.
Then, $\rcell{\tri}\cap\bd T'$ is contained in an edge of $T'$ for each apexed triangle $\tri \in A$ if $\rcell{\tri}\cap \bd T'\neq\emptyset$.
Let $x$ and $y$ be the two endpoints of
$\rcell{\tri}\cap\bd T'$.  See
Figure~\ref{fig:subdivide_pseudotriangle}(b).  We trim $\triangle$
into the triangle whose corners are $x, y$ and $\apex{\tri}$.  From
now on, when we refer an apexed triangle $\triangle$, we mean
its trimmed triangle.  Then $\rcell{\triangle}\cap T'$ is still
contained in $\triangle$ by Lemma~\ref{lem:ray_in_cell}.  
We do this
for all apexed triangles.  Every apexed triangle with
$\rcell{\triangle}\cap\bd T' \neq\emptyset$ has its bottom side on $\bd
T'$, and thus Assumption~2 is satisfied.

\subsubsection{Satisfying Assumption~4}\label{sec:assum-4}
Since every base cell satisfies Assumptions~1,2 and~3, the
  maximal concave chain of each base cell has angle-span at most $\pi$ by the following lemma,
but it is possible that the angle-span is larger than $\pi/2$.

\begin{lemma}\label{lem:concave-monotone}
  For every base cell $T$ satisfying Assumptions~1,2 and~3, the maximal
  concave curve of $\bd T$ has angle-span at most $\pi$. 
\end{lemma}
\begin{proof}
  Consider the maximal concave curve $\gamma$ of $\bd T$ and every
  apexed triangle $\tri$ such that $\rcell{\tri}$ intersects $\gamma$.
  The bottom sides of such apexed triangles are pairwise interior
  disjoint and are contained in $\gamma$ by the
  assumptions. Thus the apexed triangles can be sorted along $\gamma$ with respect to
  their bottom sides.  Consider any two apexed triangles $\tri_1$ and
  $\tri_2$ appearing consecutively on the sorted list.  There is an
  arc of $\rfvd$ induced by $(\tri_1,\tri_2)$ intersecting $\bd T$.
  Thus a side of $\tri_1$ intersects a side of $\tri_2$ at a point
  lying outside of $\gamma$.  In other words, the interior of $\tri_1$
  intersects the interior of $\tri_2$, or a side of one of $\tri_1$
  and $\tri_2$ contains a side of the other triangle. Since this
  holds for every pair of consecutive apexed triangles along $\gamma$,
  the lemma holds.
\end{proof}
For each base cell $T$ whose maximal concave curve $\gamma$
  has angle-span larger than $\pi/2$, we subdivide $T$ into at most three
  subcells so that the maximal concave curve of every subcell has
  angle-span at most $\pi/2$.  While traversing $\gamma$ from one endpoint
  $v$ to the other endpoint $u$, we accumulate the turning angles at
  the vertices we traverse.
  Once the accumulated turning angle exceeds $\pi/2$ at a vertex $v'$
  of $\gamma$, we subdivide $T$ into three cells by the line through $v'u'$, where
  $u'$ is the vertex next to $v'$ along $\gamma$.
  Clearly, the part of $\gamma$ from $v$ to $v'$
  has angle-span at most $\pi/2$. The part of $\gamma$
  from $u'$ to $u$ also has angle-span at most $\pi/2$ by Lemma~\ref{lem:concave-monotone}.
  Therefore, each of the three subcells 
  has a maximal concave chain of angle-span at most $\pi/2$. 
  As
  we did before, we compute $\rfvd \cap \bd T'$ for every subcell $T'$ in
  $O(|\rfvd \cap \bd T|+|T|)$ time in total. Then now every base cell
  still satisfies Assumptions~1 and~3, and the first part of
  Assumption~2. But it is possible that the second part of
  Assumption~2 is violated. In this case, we trim each apexed triangle
  again as we did before.

  \subsubsection{Satisfying Assumption~5}\label{sec:assum-3}
  For every apexed triangle $\tri$, we modify $f_\tri(\cdot)$ as follows by defining $\hat{x}_\triangle$ differently.
  The algorithm~\cite{klein1994} computes the abstract Voronoi diagram restricted to each 
  side of a given Hamiltonian curve. In our case, it is sufficient to compute the
  abstract Voronoi diagram restricted to the side of $\gamma$ containing $T$.
  Thus, we restrict $f_\tri(\cdot)$ to be defined in the side of $\gamma$ containing $T$.
  
  First, we perturb $\ell_3$ and $\ell_5$ slightly such that $\ell_3$
  and $\ell_5$ are circular arcs with common endpoint $\apex{\tri}$ and
  the other endpoints on $\bd B$.  
  We choose a sufficiently large number $r(\tri)$ 
  which is the common radius of $\ell_3$ and $\ell_5$. The rules for choosing $r(\tri)$
  will be described later.
  We let the center of $\ell_3$
  lie on the halfline from $\apex{\tri}$ in the direction opposite
  to $\ell_1$. Note that the center is fixed since the radius $r(\tri)$ is fixed.
  Similarly, we let the center of $\ell_5$ lie on the halfline from $\apex{\tri}$
  in the direction opposite to $\ell_2$.
  See Figure~\ref{fig:assumptoin4}(a).
  Then, for a point $x
  \in \lsideregion{\tri}\cup\ltopregion{\tri}$, we map $x$ into the
  point $\hat{x}_\tri$ on the line containing $\ell_1$ such that $r(\tri)=\|\hat{x}_\tri - c
  \|=\|x - c\|$, for some point $c$ on the halfline from $\apex{\tri}$ in the
  direction opposite to $\ell_1$. Note that
  $\hat{x}_\tri$ is unique. Similarly, we define $\hat{x}_\tri$
  for a point $x \in \rsideregion{\tri}\cup\rtopregion{\tri}$.
  Now, each contour curve consists of three circular
  arcs. See Figure~\ref{fig:assumptoin4}(b).
  
  \begin{figure}
  	\begin{center}
  		\includegraphics[width=0.5\textwidth]{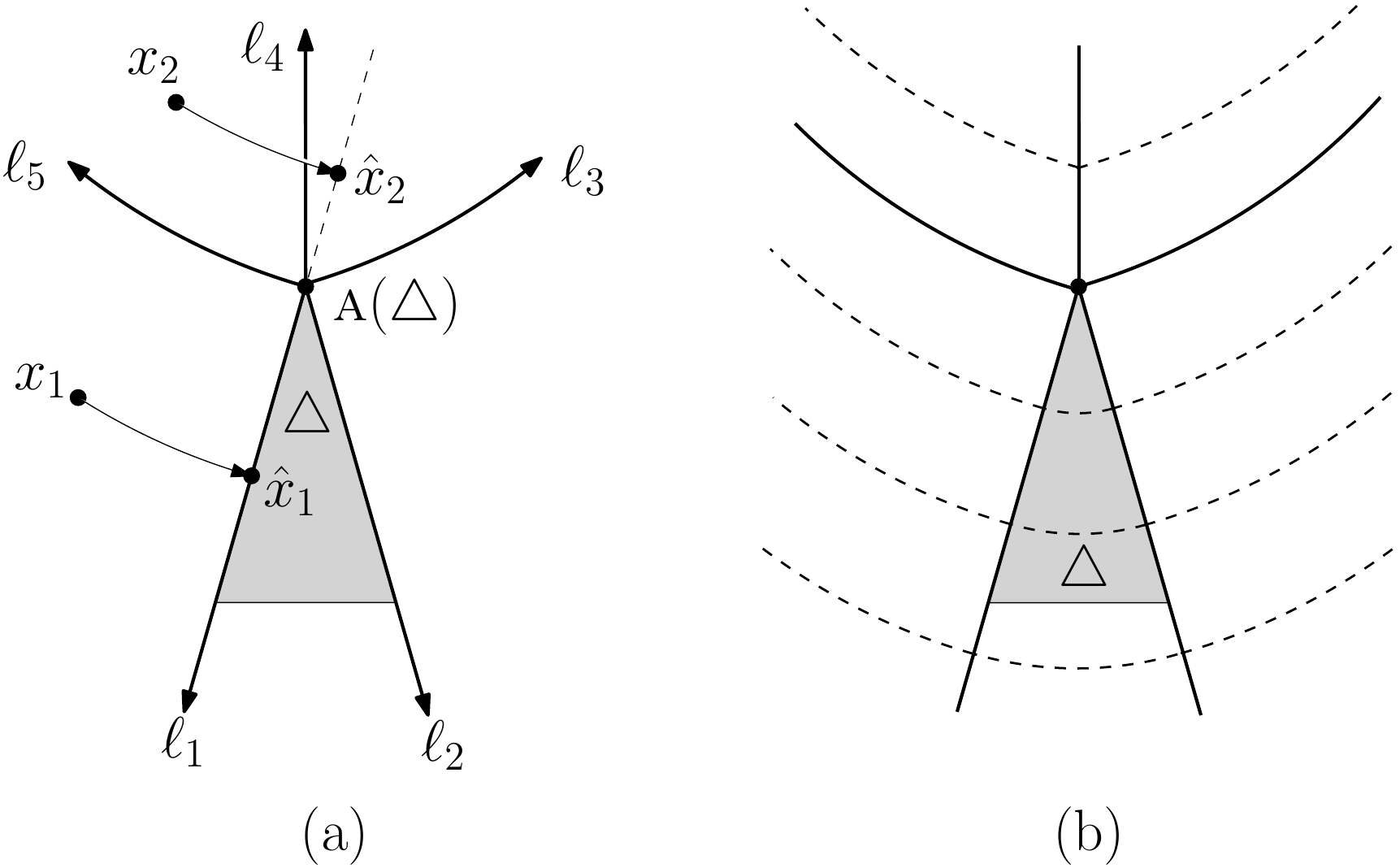}
  		\caption {\small (a) The box $B$ is not depicted. Now $\ell_3$ and $\ell_5$ are circular arcs.
  			(b) A contour curve consists of three circular arcs.
  		}
  		\label{fig:assumptoin4}
  	\end{center}
  \end{figure}

  There are three rules with regard to choosing $r(\tri)$: (1)
  $r(\tri)\neq r(\tri')$ for any two distinct apexed triangles $\tri$
  and $\tri'$, (2) $r(\tri)$ is larger than the diameters of $B$ and $P$ for every $\tri$, 
  and (3) for any two apexed triangles $\tri$ and $\tri'$ such that
  the bottom side of $\tri$ is contained in the maximal concave curve of $T$ and
  the bottom side of $\tri'$ is contained in the maximal convex curve of $T$,
   we have $r(\tri) < r(\tri')$. We can choose $r(\tri)$ for every apexed triangle $\tri$ in
   time linear in the number of the apexed triangles in $A$.
   Here, we need Rule~(1) to satisfy Lemma~\ref{lem:pseudo_bisector},
   Rule~(2) to satisfy Lemmas~\ref{lem:contain-vertex},~\ref{lem:contain-apex}
   and~\ref{lem:pseudo_dist_farthest_case1}, and Rule~(3) to satisfy Lemma~\ref{lem:acell}.

  All previous lemmas and corollaries, except Lemma~\ref{lem:pseudo_dist_farthest_case1}, 
  hold for the new distance function. For Lemma~\ref{lem:pseudo_dist_farthest_case1},
  we can prove that the clockwise angle from $\ell_1$ to the line passing through
  $\apex{\tri}$ and $\ell_5\cap \bd B$ 
  is at most a constant $\pi/\alpha$ 
  depending only on the ratio between the maximum of $r(\tri)$ and the diameter of $B$ (or $P$).
  Since we can choose the maximum of $r(\tri)$ to be a constant times the diameter of $B$ (or $P$), we can regard this ratio as a constant.
  Then we replace $\pi/2$ with $\pi/\alpha$ in Assumption~4, and subdivide a base cell using at most
  $\alpha$ lines instead of subdividing it using only one line in Section~\ref{sec:assum-4}. Then we can obtain base cells
  satisfying Assumption~4, and Lemma~\ref{lem:pseudo_dist_farthest_case1} holds.
With the new distance function, we can compute $\rfvd\cap T$
without any assumptions in $O(|\rfvd\cap \bd T|+|T|)$ time.
\begin{lemma}
	Given a base cell $T$ constructed by the subdivision algorithm in Section~\ref{sec:subdivision}, 
	$\rfvd\cap T$ can be computed in $O(\complexity{\rfvd\cap\bd T}+\complexity{T})$ time.
\end{lemma}

\paragraph{Remark on the space complexity.}
By Lemma~\ref{lem:complexity_tpathcell}, $\rfvd$ restricted to all
cells in the final iteration of Step~2 is of complexity
$O(n\log\log n)$. Thus the space complexity is $O(n\log\log n)$.  We
can improve the space complexity to $O(n)$ as follows.  When we
recursively apply the subdivision of Step~2 to each cell, we give a
specific order of the cells: when the recursion is completed for one
cell, we apply the subdivision to one of its neighboring cells.
Moreover, when we obtain a base cell, we apply Step~3 without waiting
until Step~2 is completed.  Assume that we complete the recursions for
two adjacent cells. Then we have $\rfvd$ restricted to these cells. We
merge two Voronoi diagrams and discard the information on the common
boundary of these cells.  In this way, the part of $\rfvd$ we have is
of complexity $O(n)$ and the cells we maintain are of complexity
$O(n)$ in total at any time.  Therefore, we can compute $\rfvd$ in
$O(n\log\log n)$ time using $O(n)$ space.

\begin{theorem}
  The geodesic farthest-point Voronoi diagram of the vertices of a
  simple $n$-gon can be computed in $O(n\log\log n)$ time using $O(n)$
  space.
\end{theorem}

\section{Sites Lying on the Boundary of a Simple Polygon}\label{section:General set o sites}
In this section, we show that the results presented in the previous sections are general
enough to work for an arbitrary set $S$ of sites contained
in the boundary of $P$.
In this case, we assume without loss of generality that
all sites of $S$ are vertices of $P$. This can be achieved
by splitting each edge $uv$ that contain a site $s$ of $S$
into two, $us$ and $sv$, with a new vertex $s$.

We decompose the boundary of $P$ into chains
of consecutive vertices that share the same $S$-farthest neighbor and
edges of $P$ whose endpoints have distinct $S$-farthest neighbors.
The following lemma is a counterpart of Lemma~\ref{lemma:Matrix lemma}.
Lemma~\ref{lemma:Matrix lemma}
is the only place where it was assumed that $S$ is the set of
vertices of $P$.
The algorithm for computing a set of apexed triangles in~\cite{1-center} is based on Lemma~\ref{lemma:Matrix lemma}.
By replacing Lemma~\ref{lemma:Matrix lemma} with Lemma~\ref{lemma:Matrix lemma for general S}, this algorithm works for a set of sites on the boundary.

\begin{lemma}\label{lemma:Matrix lemma for general S}
  Given a set $S$ of $m$ sites contained in $\bd P$, we can compute
  the $S$-farthest neighbor of each vertex of $P$ in $O(n+m)$ time.
\end{lemma}
\begin{proof}
  Let $w:P\to \mathbb{R}$ be a real valued function on the vertices of
  $P$ such that for each vertex $v$ of $P$,
$$w(v) = \left\{\begin{array}{ll} 
    D_P&\text{if $v\in S$}\\
    0&\text{otherwise,}
  \end{array}\right.$$
where $D_P$ is any fixed constant larger than the geodesic diameter of
$P$. Recall that the diameter of $P$ can be computed in linear
time~\cite{hershberger1993matrix}.

For each vertex $v\in P$, we want to identify
the $S$-farthest neighbor $\ff{v}$. To this end,
we define a new distance function $d^*:P\times P\to \mathbb{R}$ such that for
any two points $p$ and $q$ of $P$, $d^*(p,q) = d(p,q) + w(p) + w(q)$.
Using a result from Hershberger and Suri~\cite[Section 6.1 and
6.3]{hershberger1993matrix}, we can compute the
farthest neighbor of each vertex of $P$ with respect to $d^*$ in $O(n+m)$ time. 

By the definition of the function $w$, the maximum distance from any
vertex of $P$ is achieved at a site of $S$.  Therefore, the farthest
neighbor from a vertex $v$ of $P$ with respect to $d^*$ is indeed the
$S$-farthest neighbor, $\ff{v}$, of~$v$.
\end{proof}

\begin{theorem}\label{thm:boundary}
  The geodesic farthest-point Voronoi diagram of $m$ points on the
  boundary of a simple $n$-gon can be computed in $O((n+m)\log\log n)$
  time.
\end{theorem}

\subsection{Few Convex or Few Reflex Vertices}
We can compute $\fvd[P,S]$ in
$O((n+m)\log\log \min\{c,r\})$ time for a simple $n$-gon $P$ and a set
$S$ of $m$ points on the boundary of $P$, where $c$ is the number of
the convex vertices of $P$ and $r$ is the number of the reflex
vertices of $P$. 
To achieve this, we apply two different algorithms depending on
whether $r\geq c$ or $r<c$. 

\paragraph{Few Convex Vertices.}  In the case that $r\geq c$, we
simply apply the algorithm in Theorem~\ref{thm:boundary}. We give a
tighter analysis that the running time is $O((n+m)\log\log c)$.
A basic observation is that the $\log\log n$
factor in the running time of Theorem~\ref{thm:boundary} is the number
of iterations of the subdivision in the second phase described in
Section~\ref{sec:second-step}. In the second phase, we choose every
$\lfloor\sqrt{t}\rfloor$th \emph{convex} vertices of a $t$-path-cell
along its boundary to subdivide the cell into $(\lfloor\sqrt{t}\rfloor+1)$-path-cells 
and base cells for some $t\in\mathbb{R}$.
They recursively subdivide $t'$-path-cells for $t'>3$ until every cell becomes a $3$-path-cell
or a base cell. Initially, we are given $P$ as a $c$-path-cell.
This implies that the
number of iterations of the subdivision is indeed $\log\log c$. 
Therefore, we can obtain $\fvd[P,S]$ in $O((n+m)\log\log \min\{c,r\})$
time if $r\geq c$.

\paragraph{Few Reflex Vertices.}
We first compute $\fvd[P,S]$ restricted to the boundary of $P$ in
$O(n+m)$ time using Theorem~\ref{thm:restrict-boundary}.  Then we
subdivide $P$ into a number of lune-cells and $r$-path-cells as
follows.  Recall that a lune-cell is a subpolygon of $P$ whose
boundary consists of a convex chain and a concave chain.  We compute
the geodesic convex hull $\ch$ of the reflex vertices of $P$.  The
interior of $\ch$ consists of a number of connected regions. Note that
each connected region has complexity $O(r)$.  In other words, each
connected region is an $r$-path-cell. Moreover, the total complexity
of all connected regions is $O(r)$.

Consider the connected regions of $P\setminus \ch$. The boundary of each connected region
consists of a part of the boundary of $\ch$ and a convex chain connecting some convex vertices of $P$ and sites of $S$.
Thus, each connected region is a lune-cell. See Figure~\ref{fig:CH}.

\begin{figure}
	\begin{center}
		\includegraphics[width=0.3\textwidth]{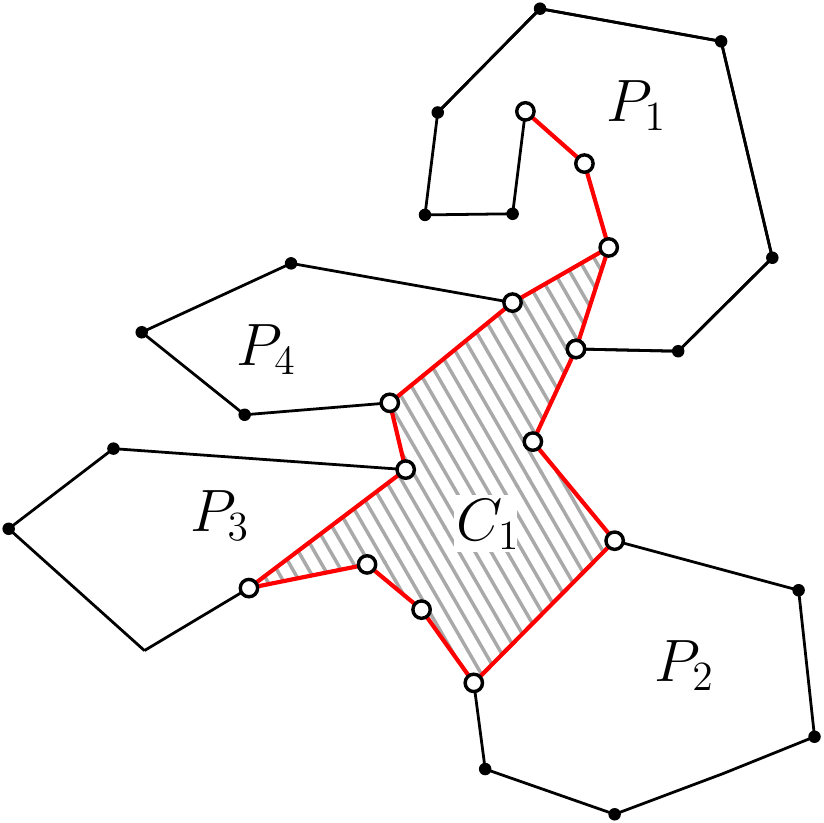}
		\caption {\small The geodesic convex hull (dashed region) of the reflex vertices of $P$ 
			subdivides $P$ into lune-cells ($P_i$'s for $i=1,2,3,4$) and one $t$-path-cell
			($C_1$).\label{fig:CH}}
	\end{center}
\end{figure}

We compute $\fvd$ restricted to the boundary of $\ch$ in $O(n+m)$ time
using Lemma~\ref{lem:gamma_cap_fvd}.  Then for each $r$-path-cell, we
apply the algorithm in Section~\ref{sec:second-step} and compute
$\fvd[P,S]$ restricted to the cell in total $O((n+m)\log\log r)$ time.
For each lune-cell, we apply the algorithm in
Section~\ref{sec:third-step} and compute $\fvd[P,S]$ restricted to
each cell. The total complexity of $\fvd$ restricted to all lune-cells
and all $r$-path-cells is $O(n+m)$. Therefore, we can obtain $\fvd$
restricted to each lune-cell in $O(n+m)$ time in total.  Since the
cells are pairwise interior disjoint, we can obtain $\fvd[P,S]$ by
simply combining all of them.  Therefore, we can obtain $\fvd[P,S]$ in
$O((n+m)\log\log \min\{c,r\})$ time if $r<c$, and we have the
following theorem.

\begin{theorem}\label{thm:few}
  For a simple $n$-gon $P$ with $c$ convex vertices and $r$ reflex
  vertices, we can compute the farthest-point geodesic Voronoi diagram
  of a set of $m$ sites on the boundary of $P$ in
  $O((n+m)\log\log \min\{c,r\})$ time.
\end{theorem}

\section{Sites Lying in a Simple Polygon} \label{sec:general}
In this section, we consider a set $S$ of point sites lying in $P$.
It is known that a site of $S$ appears on the boundary of the geodesic convex hull
$\ch$ of $S$ if it has a nonempty Voronoi cell in $P$.  Thus, we first compute $\ch$ in
$O(n+m\log m)$ time~\cite{guibasShortestPathQueries}.  Since the
sites lying in the interior of $\ch$ do not have nonempty Voronoi
cells, we remove them from $S$.  Then every site of $S$ lies on the
boundary of $\ch$.

Our algorithm consists of two steps. In the first step,
we subdivide $P$ into three interior
disjoint subpolygons,
which we call \emph{funnels}, in $O(n+m)$ time. The boundary of a funnel
consists of two line segments and a part of the boundary of $P$.
In the second step,
we compute $\fvd[P,S]$ restricted to each
funnel in $O((n+m)\log\log n)$ time. By merging them, we can obtain
$\fvd[P,S]$ in $O((n+m)\log\log n)$ time excluding the time for
computing $\ch$, or $O(n\log\log n+m\log m)$ time including
the time for computing $\ch$.

\subsection{Subdivision of \texorpdfstring{$P$}{P} into Three Funnels}\label{sec:subdivision-general}
We compute the geodesic center $g$ of $S$ with respect to $P$. 
Recall that it is the point in $P$ that minimizes
the maximum geodesic distance to all sites of $S$. Moreover, it coincides with 
the geodesic center of $\ch$ with respect to $\ch$~\cite[Corollary 2.3.5]{aronov1993furthest}. Ahn at al.~\cite{1-center} presented an algorithm to
compute the geodesic center of a simple $n$-gon 
in $O(n)$ time. Since we have $\ch$, we can compute $g$ in $O(n+m)$ time.

We subdivide $P$ into three \emph{funnels} with respect to $g$ as
follows.  There are at most three sites equidistant from $g$ by the
general position assumption.  Let $s_1, s_2$ and $s_3$ be such sites
sorted in clockwise order along the boundary of $\ch$.  ($s_3$ might
not exist.)  While computing the center $g$, we can obtain such sites.
For each $i=1,2,3$, we extend the edge of $\pi(g,s_i)$ incident to the
center $g$ towards $g$ until it escapes from $P$.  Let $s_i'$ be the
point on $\bd P$ hit by this extension.  See
Figure~\ref{fig:funnel}(a).  Then the three line segments
$gs_1', gs_2'$ and $gs_3'$ subdivide $P$ into three regions whose
boundary consists of a part of $\bd P$ and two line segments sharing a
common endpoint $g$.  We call each region a \emph{funnel}. We call the
common boundary of $P$ and a funnel the \emph{bottom side} of the
funnel.

\begin{figure}
  \begin{center}
    \includegraphics[width=0.6\textwidth]{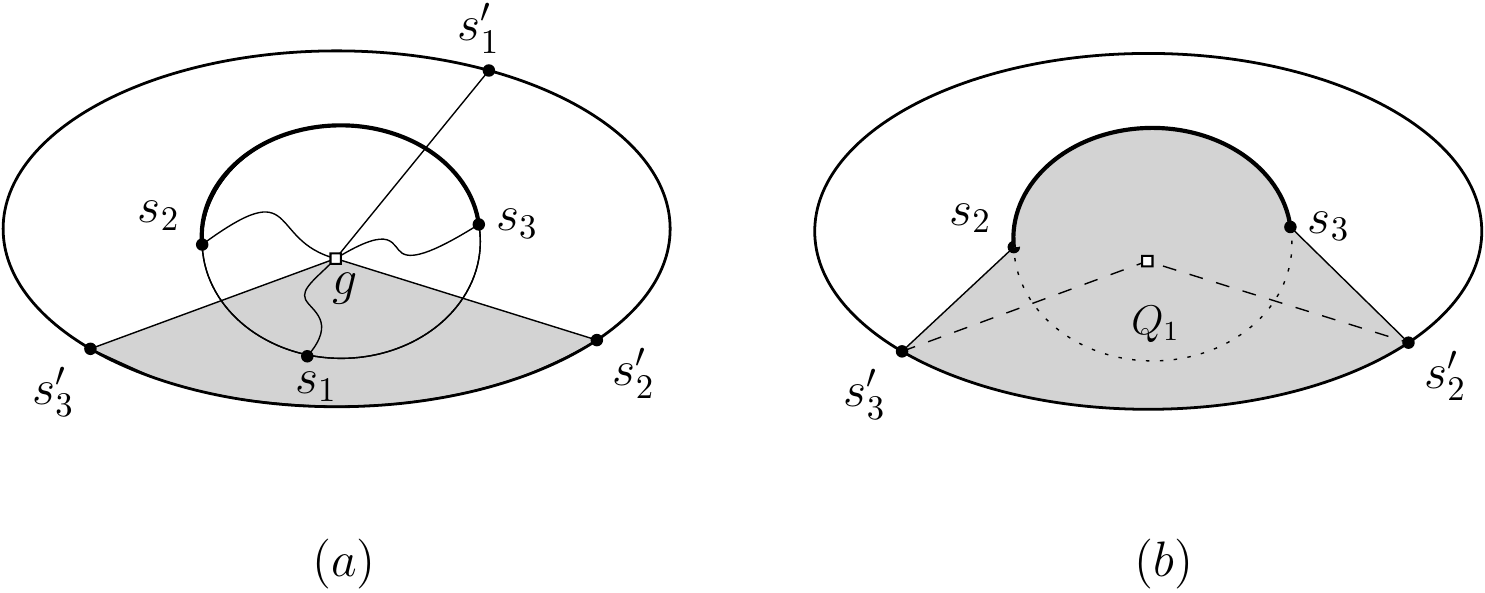}
    \caption{\small (a) The funnel $F_1$ (gray region). (b)
        There are some points $p, q\in Q_1$ such that the geodesic path between $p$ and $q$
         restricted to lie in $Q_1$ is not the same as $\pi(p,q)$, the
          geodesic path between $p$ and $q$ restricted to lie in $P$.}
      \label{fig:funnel}
    \end{center}
\end{figure}

We denote the funnel bounded by $gs_2'$ and $gs_3'$ by $F_1$. We
denote the set of the sites of $S$ lying on the part of the boundary of $\ch$
from $s_2$ to $s_3$ in clockwise order by $S_1$.  Similarly, we define
$F_2, F_3$ and $S_2, S_3$. Note that $F_i$'s are pairwise interior
disjoint.

We can compute $s_1',s_2'$ and $s_3'$ in $O(\log n)$
time~\cite{chazelle1994ray}.  Therefore, we can obtain $F_i$ and $S_i$
for $i=1,2,3$ in $O(n+m)$ time in total.  Now we consider a few
properties of the funnels.  By definition, the $S$-farthest neighbors of $g$ are
$s_1, s_2$ and $s_3$. Thus, $g$ lies on the common boundary of
$\mathsf{Cell}(S,s_i)$ for $i=1,2,3$. By
Corollary~\ref{cor:ray-in-non-refined-cell}, we have the following lemma.

\begin{lemma}\label{lem:wall}
  The line segment $gs_i'$ for $i=1,2,3$ is contained in
  $\mathsf{Cell}(S,s_i)$.
\end{lemma}

Due to the following property, we can obtain $\fvd[P,S]$ by simply
merging $\fvd[P,S_i]$ restricted to $F_i$ for $i=1,2,3$.  We show how
to compute $\fvd[P,S_i]$ restricted to $F_i$ in the following subsection.

\begin{lemma}
  For each $i=1,2,3$, every point in $F_i$ has its $S$-farthest
  neighbor in $S_i$.
\end{lemma}
\begin{proof}
  Aronov et al.~\cite{aronov1993furthest} showed that $\fvd[P,S]$
  forms a tree, that is, every nonempty Voronoi cell is incident to
  the boundary of $P$. Moreover, they showed that every Voronoi cell
  of $\fvd[P,S]$ is connected. By Lemma~\ref{lem:wall}, $gs_2'$ is
  contained in $\mathsf{Cell}(S,s_2)$ and $gs_3'$ is contained in
  $\mathsf{Cell}(S,s_3)$.  Therefore, for any site $s\in S$ with
  $\mathsf{Cell}(S,s)\cap F_i\neq\emptyset$, its Voronoi cell
  $\mathsf{Cell}(S,s)$ intersects the bottom side of $F_i$.
	
  The ordering lemma states that
  the order of sites along the boundary of $\ch$ is the same as the
  order of Voronoi cells along $\bd P$.  Therefore, a site $s$ whose
  Voronoi cell intersects the bottom side of $F_i$ is in $S_i$. Thus,
  the lemma holds.
\end{proof}

The following property is used to compute $\fvd[P,S_i]$ restricted to
$F_i$ for $i=1,2,3$ in the following section.
\begin{lemma}\label{lem:separate}
  For each funnel $F_i$, there are two points $p_1, p_2\in\bd P$ such
  that $\pi(p_1,p_2)$ separates $F_i$ and $S_i$.
\end{lemma}
\begin{proof}
  Since $g$ is the geodesic center of $\ch$, there are two points,
  $q_1$ and $q_2$, on the boundary of $\ch$ such that $\pi(q_1,q_2)$ contains $g$
  and $\pi(q_1,q_2)$ separates
  $\{s_2,s_3\}$ and $\{s_1\}$. Otherwise, we can move the
  position of $g$ slightly to reduce $d(g,s_i)$ for all $i=1,2,3$, which
  contradicts that $g$ is the geodesic center of $\ch$ and $s_i$'s are
  the $S$-farthest neighbors of $g$.  Note that one part of $\ch$
  bounded by $\pi(q_1,q_2)$ contains all of $\pi(g,s_2), \pi(g,s_3)$ and
  $S_1$.
	
  We extend the edge of $\pi(q_1,q_2)$ incident to $q_j$ towards $q_j$ until
  it escapes from $P$, and let $p_j$ be the point on $\bd P$ hit by the
  extension  for each $j=1,2$.
  Then $\pi(p_1,p_2)$ contains $\pi(q_1,q_2)$, and therefore
  a part of $P$ bounded by $\pi(p_1,p_2)$ contains $\pi(g,s_2)$,
  $\pi(g,s_3)$ and $S_1$.  Thus, $gs_3'$ and
  $gs_2'$ are contained in the other part of $P$ bounded by
  $\pi(p_1,p_2)$, and so does $F_1$.  This means that $\pi(p_1,p_2)$
  separates $F_1$ and $S_1$. The argument also works
  for the other pairs of $F_i$ and $S_i$ for $i=2,3$ analogously.
\end{proof}

\subsection{Computing \texorpdfstring{$\fvd$}{FVD} Restricted to Each
  Funnel}\label{sec:FVDfunnel}
Consider $F_1$ and $S_1$. We can handle $F_i$ and $S_i$ for $i=2,3$
analogously.  We want to compute $\fvd[P,S_1]$ restricted to $F_1$.
The algorithm in Theorem~\ref{thm:boundary} requires all sites to lie
on the boundary of a simple polygon. However, in our case, some sites of $S_1$
may not lie on the boundary of $P$.
 Our general strategy is to construct three pairs
$(P_j,S_j)$ of subpolygons $P_j$ of $P$ and subsets $S_j$ of $S$ with
$j=a,b,c$ such that the sites of $S_j$ lie on the boundary of
$P_j$. Then we apply the algorithm in Theorem~\ref{thm:boundary} to
each pair $(P_j,S_j)$ and compute the geodesic farthest-point Voronoi diagram
restricted to $P_j$.
We show that we can obtain $\fvd[P,S_1]$ restricted to
$F_1$ by combining the diagrams.

Consider the subpolygon $Q_1$ of $P$ whose boundary consists of the
bottom side of $F_1$, $\pi(s_2',s_3)$, $\pi(s_3',s_2)$ and the part of
$\ch$ from $s_2$ to $s_3$ in clockwise order.  See
Figure~\ref{fig:funnel}(b).  Note that all sites of $S_1$ are on the
boundary of $Q_1$. However, applying the algorithm in
Theorem~\ref{thm:boundary} with input polygon $Q_1$ may not
give a correct diagram in this case.  This is because there are
points $p, q\in Q_1$ such that the geodesic path between $p$ and $q$
is not contained in $Q_1$.

To avoid this, we consider the geodesic convex hull $H_1$ of $Q_1$
instead.  See Figure~\ref{fig:three-pairs}(a). Since $Q_1$ is a simple
polygon contained in $P$, we can compute $H_1$ in $O(n+m)$
time~\cite{toussaint}.  Let $t_2$ and $t_3$ be the sites of $S$ such
that $\pi(s_3',t_2)$ and $\pi(s_2',t_3)$ lie on the boundary of $H_1$
and intersect $\ch$ only at $t_2$ and $t_3$, respectively.  There
exist such two points by Lemma~\ref{lem:separate}.  Note that $H_1$
contains the geodesic path of any two points lying in $H_1$.

Recall that our goal is to compute $\fvd[P,S_1]$ restricted to $F_1$.
It coincides with $\fvd[H_1,S_1]$ restricted to $F_1$ since $H_1$
contains the geodesic path of any two points in $H_1$ and contains both
$S_1$ and $F_1$.  However, there might be some sites of
$S_1$ in the interior of $H_1$.

\begin{observation}\label{obs:hourglass}
  $\fvd[P,S_1]$ restricted to $F_1$ coincides with $\fvd[H_1,S_1]$
  restricted to $F_1$.
\end{observation}

\begin{figure}
  \begin{center}
    \includegraphics[width=0.8\textwidth]{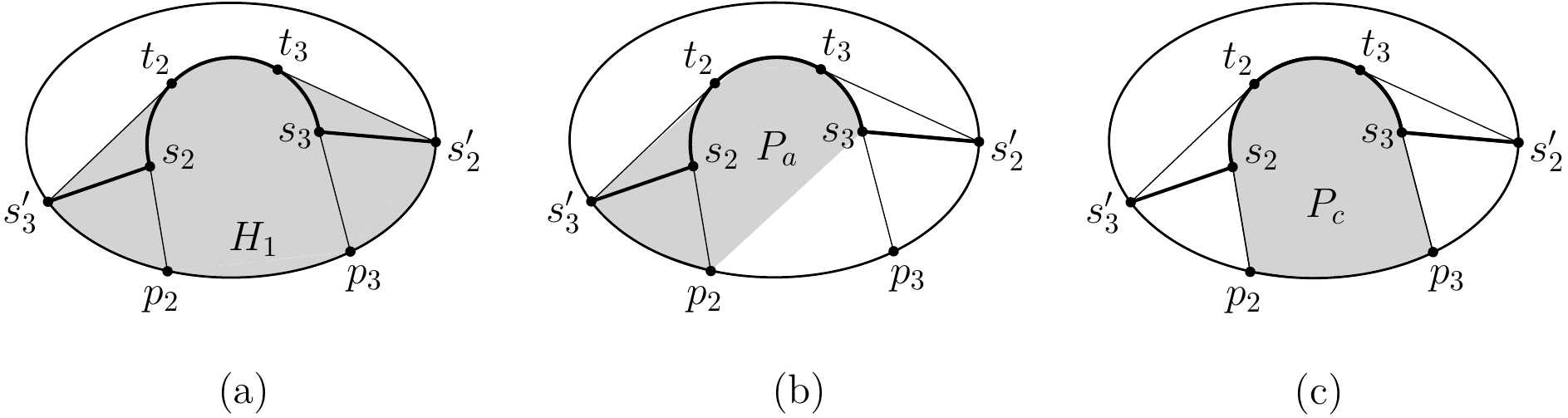}
    \caption {\small (a) The gray region $H_1$ contains the geodesic paths of any two points in $H_1$
      and contains both $S_1$ and $F_1$.  (b)	Every point in $P_a\setminus P_c$ has its $S_1$-farthest neighbor on
      $\ch[t_2,s_3]\cup\{s_2\}$. 
      (c) The boundary of $P_c$ contains all sites of $S_1$.}
    \label{fig:three-pairs}
  \end{center}
\end{figure}

We consider three subpolygons $P_j$ of $H_1$ associated with site sets $S_j$ with $j=a,b,c$ whose union is $H_1$.
We will see that every site in $S_j$ lies on the boundary of $P_j$.
For any two sites $s$ and $s'$ in $S_1$, we use $\ch[s,s']$ to denote 
the part of the boundary of $\ch$ lying from $s$ to $s'$ in clockwise order.
Similarly, for any two points $p$ and $p'$ on the boundary of $P$, we use $P[p,p']$ to denote the 
part of $\bd P$ lying from $p$ to $p'$ in clockwise order.

Consider $\ch[s_2,s_3]$.
We extend its two edges of adjacent to $s_2$ and to $s_3$ towards $s_2$ and $s_3$
until they escape from $P$ at points $p_2$ and $p_3$ of $\bd P$, respectively.
If $p_2$ or $p_3$ does not lie on the bottom side of $F_1$, we simply set
$p_2=s_3'$ or $p_3=s_2'$, respectively. 
See Figure~\ref{fig:three-pairs}(a).

\begin{lemma}\label{lem:order-footpoint}
  The four points $s_2', p_3, p_2$ and $s_3'$ lie on the boundary of
  $\bd P$ in clockwise order.
\end{lemma}
\begin{proof}
  Lemma~\ref{lem:separate} implies that $\pi(s_3',s_2')$ intersects
  $s_2p_2$. Let $p_2'$ be an intersection point.  Consider the clockwise
  angle from the line segment $p_2's_2$ to the edge of $\pi(p_2',s_2')$
  incident to $p_2'$.  This angle must be less than $\pi/2$.
  Otherwise, we have 
  $d(s_2',s')>d(s_2',s_2)$ by~\cite[Corollary~2]{pollackComputingCenter},
  where $s'$ is the clockwise neighbor of $s_2$ along the boundary of $\ch$.
  This contradicts to the definition of $s_2'$.
  The same holds for $s_3'$ with respect to the intersection point $p_3'$ of 
  $\pi(s_2',s_3')$ with $s_3p_3$ and the counterclockwise neighbor of $s_3$.

  Now we claim that $s_2p_2$ and $s_3p_3$ does not intersect each other. 
  Assume to the contrary that they intersect each other at $x$.
  Since $\pi(s_3',s_2')$ separates $s_2$ and $s_3$ from $p_2$ and $p_3$,
    there are two cases: either $x$ lies in the subpolygon of $P$ induced by $\pi(s_3',s_2')$
    in which $s_2$ and $s_3$ lie or not.
  See Figure~\ref{fig:neighbor}(a-b).
  In any case, consider the pseudo-triangle with three corners $x, p_2'$ and $p_3'$.
  The path $\pi(p_3',p_2')$ is a concave chain with respect to this pseudo-triangle. 
  This contradicts that the angles at $p_3'$ and $p_2'$ are less than $\pi/2$.
  Therefore, the claim holds, and the four points $s_2', p_3, p_2$ and $s_3'$ lie on the boundary of
  $\bd P$ in clockwise order.
\end{proof}

Let $P_a$ be the subpolygon of $H_1$ whose boundary consists of
$\ch[t_2,s_3]$, $\pi(s_3,p_2)$, $P[p_2,s_3']$ and $\pi(s_3',t_2)$.
See Figure~\ref{fig:three-pairs}(b).  Similarly, let $P_b$ be the
subpolygon of $H_1$ whose boundary consists of $\ch[s_2,t_3]$,
$\pi(t_3,s_2')$, $P[s_2',p_3]$ and $\pi(p_3,s_2)$.  Let $P_c$ be the
subpolygon of $H_1$ whose boundary consists of $\ch[s_2,s_3]$,
$s_3p_3$, $P[p_3,p_2]$ and $p_2s_2$.  See
Figure~\ref{fig:three-pairs}(c).

\begin{lemma}\label{lem:neighbor}
	Every point in $P_a\setminus P_c$ has its $S_1$-farthest neighbor on
	$\ch[t_2,s_3]\cup\{s_2\}$. Similarly, every point in $P_b\setminus P_c$ has its
	$S_1$-farthest neighbor on $\ch[s_2,t_3]\cup\{s_3\}$.
\end{lemma}
\begin{proof}
  We prove only the first part of the lemma. The second part can be
  proved analogously.  We claim that the $S_1$-farthest neighbor of
  $p_2$ is in $\ch [t_2,s_3]$.  If the claim is true, every point in
  $P [p_2,s_3']$ has its $S_1$-farthest neighbor on $\ch[t_2,s_3]$ due
  to the ordering lemma~\cite[Corollary 2.7.4]{aronov1993furthest}.
  Moreover, every point on $\pi(s_3',t_2)$, $\ch[t_2,t_3]$, and
  $\pi(t_3,s_2')$ has its $S_1$-farthest neighbor on
  $\ch[t_2,s_3]\cup\{s_2\}$ due to the ordering lemma and the
  definition of $s_3'$ and $s_2'$.  Due to Corollary~\ref{cor:ray-in-non-refined-cell}, every point in $P_a\setminus P_c$ has its
  $S_1$-farthest neighbor on $\ch[t_2,s_3]\cup\{s_2\}$.
  
  To prove the claim, assume to the contrary that the
  $S_1$-farthest neighbor of $p_2$ is in
  $\ch[s_2,t_2]\setminus\{t_2\}$. 
  To make the description easier, we assume that $p_2s_2$ is
  vertical. See Figure~\ref{fig:neighbor}(c).  We first observe that
  $\ch[s_2,t_2]$ is a convex chain with respect to $\ch$.  If it is
  not true, there is some vertex of $P$ that appears on
  $\ch[s_2,t_2]\setminus\{t_2\}$, and $\pi(s_3',t_2)$ overlaps with
  $\ch[s_2,t_2]$ at the vertex, which contradicts that $\pi(s_3',t_2)$
  intersects $\ch$ only at $t_2$.
	
  \begin{figure}
    \begin{center}
      \includegraphics[width=0.9\textwidth]{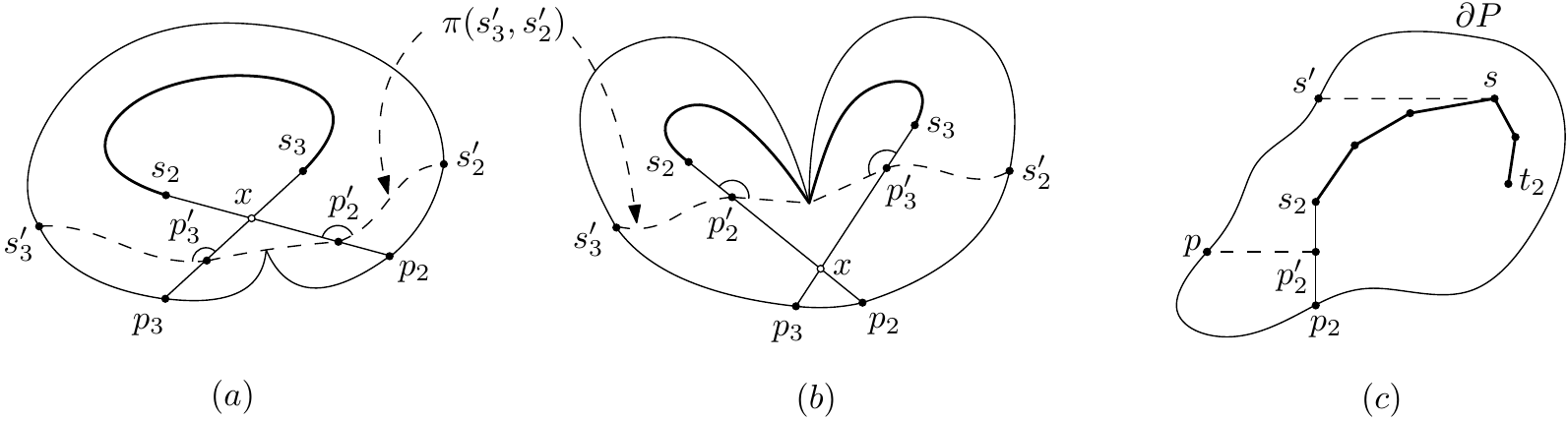}
      \caption {\small 
      	(a,b) If $s_2p_2$ and $s_3p_3$ intersect at $x$, the angles at $p_3'$
        and $p_2'$ are at least $\pi/2$.
      	(c) The proof of Lemma~\ref{lem:order-footpoint} implies
        that $s_3'$ lies on $P[p_2,p'']$.  However, since
        $\pi(s_3',t_2)$ intersects $\ch$ only at $t_2$, the point
        $s_3'$ lies on $P[s',p_2]$, which is a
        contradiction. \label{fig:neighbor}}
    \end{center}
  \end{figure}
	
  Since $\ch[s_2,t_2]$ is a convex chain and $t_2$ is not
  the $S_1$-farthest neighbor of $p_2$, $t_2$ is not the highest point of the chain.
  Let $s'$ be the first point on $\bd P$ hit by the ray from the highest point $s$ of the chain
  going to the left horizontally.
  Since $\pi(s_3',t_2)$
  intersects $\ch$ only at $t_2$, the point $s_3'$ lies on
  $P[s',p_2]$.
	
  Now, as we did in the proof of Lemma~\ref{lem:order-footpoint}, we
  consider an intersection point $p_2'$ between $\pi(s_3',s_2')$ and
  $s_2p_2$. We already showed that the clockwise angle from the line
  segment $p_2's_2$ to the edge of $\pi(p_2',s_2')$ incident to $p_2'$ is less than
  $\pi/2$.  Thus, the counterclockwise angle from $p_2's_2$ to the
  edge of $\pi(p_2',s_3')$ incident to $p_2'$ is larger than $\pi/2$.
  Let $p$ be the first point on $\bd P$ hit by the ray from $p_2'$
  going to the left horizontally.
  Since the counterclockwise angle is larger than $\pi/2$, the point
  $s_3'$ lies on $P[p_2,p]$.
	
  Since $s$ is the highest point of $\ch[s_2,t_2]$, the three points
  $p_2, p$ and $s'$ lie on the boundary of $P$ in clockwise order.
  Therefore, $P[s',p_2]$ and $P[p_2,p]$ are interior disjoint, which
  is a contradiction.
  Therefore, the claim holds.
\end{proof}

Moreover, the proof of this lemma implies the following corollary.
\begin{corollary}\label{cor:s}
  There is no point in $F_1\cap (P_a\setminus P_c)$ whose
  $S_1$-farthest neighbor is $s_2$.  Similarly, there is no point in
  $F_1\cap (P_b\setminus P_c)$ whose $S_1$-farthest neighbor is $s_3$.
\end{corollary}

After obtaining $P_a, P_b$ and $P_c$ in $O(n+m)$ time, we compute
$\fvd[H_1,S_1]$ restricted to $F_1$ as follows.  For $\fvd[P_c,S_1]$, we use the fact that
all sites of $S_1$ lie on the boundary of $P_c$. 
Since $P_c$ has $O(n)$ reflex vertices, we can compute
$\fvd[P_c,S_1]$ in $O((n+m)\log\log n)$ time using 
Theorem~\ref{thm:few}.
Then we cut $\fvd[P_c,S_1]$ along the boundary
of $F_1$.  Recall that $F_1$ is a subpolygon of $S_1$ bounded by two
line segments. Moreover, each line segment is contained in
$\mathsf{Cell}(S,s_i)$ for $i=2,3$. Thus, we can cut $\fvd[P_c,S_1]$
along these line segments and obtain $\fvd[P_c,S_1]$ restricted to
$F_1$ in $O(n+m)$ time once we have $\fvd[P_c,S_1]$. 

For $\fvd[P_a\setminus P_c, S_1]$ restricted to $F_1$, we use the fact that it coincides
with $\fvd[P_a,S_a]$ restricted to $F_1\cap (P_a\setminus P_c)$, where $S_a$ is
the set of sites of $S$ lying on $\ch[t_2,s_3]$, which follows from
Lemma~\ref{lem:neighbor} and Corollary~\ref{cor:s}. 
Note that all sites of $S_a$ lie on the
boundary of $P_a$.  Therefore, we can compute $\fvd[P_a, S_a]$ in
$O((n+m)\log\log n)$ time by Theorem~\ref{thm:few}.  Since we
already have $\fvd[P_c,S_1]$ restricted to $F_1$, we can obtain the part of
$\fvd[P_a,S_a]$ restricted to $F_1\cap (P_a\setminus P_c)$ in $O(n+m)$ by cutting
$\fvd[P_a,S_a]$ along the boundaries of $P_c$ and $F_1$.  Similarly, we can
compute $\fvd[P_b\setminus P_c,S_1]$ restricted to $F_1$ in the same time.

Since $P_c$, $P_a\setminus P_c$ and $P_b\setminus P_c$ are pairwise
interior disjoint, we can combine the three Voronoi diagrams easily
and obtain $\fvd[H_1,S_1]$ restricted to $F_1$ in $O(n+m)$ time.  Recall that our goal in
this subsection is to compute $\fvd[P,S_1]$ restricted to $F_1$.  By
Observation~\ref{obs:hourglass}, it is equivalent to 
$\fvd[H_1,S_1]$ restricted to $F_1$.
 Since
$\fvd[H_1,S_1]$ restricted to $F_1$ can be computed in $O((n+m)\log\log n)$ time, we can
compute $\fvd[P,S]$ in $O((n+m)\log\log n+m\log m)$ time in total including
the time for computing the geodesic convex hull $\ch$ of the sites.
For $n=O(m)$, we have $m\log\log n=O(m\log m)$.
For $n=\Omega(m)$, we have $m\log\log n=O(n\log\log n)$. Therefore,
we have the following theorem. 

\begin{theorem}
  The farthest-point geodesic Voronoi diagram of $m$ points in a
  simple $n$-gon can be computed in $O(n \log\log n+m\log m)$
  time.
\end{theorem}

\bibliographystyle{abbrvnat} 

\begin{thebibliography}{18}
	\providecommand{\natexlab}[1]{#1}
	\providecommand{\url}[1]{\texttt{#1}}
	\expandafter\ifx\csname urlstyle\endcsname\relax
	\providecommand{\doi}[1]{doi: #1}\else
	\providecommand{\doi}{doi: \begingroup \urlstyle{rm}\Url}\fi
	
	\bibitem[Aggarwal et~al.(1989)Aggarwal, Guibas, Saxe, and
	Shor]{aggarwal1989linear}
	A.~Aggarwal, L.~J. Guibas, J.~Saxe, and P.~W. Shor.
	\newblock A linear-time algorithm for computing the {V}oronoi diagram of a
	convex polygon.
	\newblock \emph{Discrete \& Computational Geometry}, 4\penalty0 (6):\penalty0
	591--604, 1989.
	
	\bibitem[Ahn et~al.(2016)Ahn, Barba, Bose, Carufel, Korman, and Oh]{1-center}
	H.-K. Ahn, L.~Barba, P.~Bose, J.-L. Carufel, M.~Korman, and E.~Oh.
	\newblock A linear-time algorithm for the geodesic center of a simple polygon.
	\newblock \emph{Discrete Comput. Geom.}, 56\penalty0 (4):\penalty0 836--859,
	2016.
	
	\bibitem[Aronov et~al.(1993)Aronov, Fortune, and Wilfong]{aronov1993furthest}
	B.~Aronov, S.~Fortune, and G.~Wilfong.
	\newblock The furthest-site geodesic {V}oronoi diagram.
	\newblock \emph{Discrete \& Computational Geometry}, 9\penalty0 (3):\penalty0
	217--255, 1993.
	
	\bibitem[Asano and Toussaint(1985)]{at-cgcsp-85}
	T.~Asano and G.~Toussaint.
	\newblock Computing the geodesic center of a simple polygon.
	\newblock Technical Report SOCS-85.32, McGill University, 1985.
	
	\bibitem[Chazelle(1982)]{c-tpca-82}
	B.~Chazelle.
	\newblock A theorem on polygon cutting with applications.
	\newblock In \emph{Proceedings of the 23rd Annual Symposium on Foundations of
		Computer Science (FOCS 1982)}, pages 339--349, 1982.
	
	\bibitem[Chazelle et~al.(1994)Chazelle, Edelsbrunner, Grigni, Guibas,
	Hershberger, Sharir, and Snoeyink]{chazelle1994ray}
	B.~Chazelle, H.~Edelsbrunner, M.~Grigni, L.~Guibas, J.~Hershberger, M.~Sharir,
	and J.~Snoeyink.
	\newblock Ray shooting in polygons using geodesic triangulations.
	\newblock \emph{Algorithmica}, 12\penalty0 (1):\penalty0 54--68, 1994.
	
	\bibitem[Guibas et~al.(1987)Guibas, Hershberger, Leven, Sharir, and
	Tarjan]{shortest-path-tree}
	L.~Guibas, J.~Hershberger, D.~Leven, M.~Sharir, and R.~Tarjan.
	\newblock Linear-time algorithms for visibility and shortest path problems
	inside triangulated simple polygons.
	\newblock \emph{Algorithmica}, 2\penalty0 (1):\penalty0 209--233, 1987.
	
	\bibitem[Guibas and Hershberger(1989)]{guibasShortestPathQueries}
	L.~J. Guibas and J.~Hershberger.
	\newblock Optimal shortest path queries in a simple polygon.
	\newblock \emph{Journal of Computer and System Sciences}, 39\penalty0
	(2):\penalty0 126--152, 1989.
	
	\bibitem[Hershberger and Suri(1997)]{hershberger1993matrix}
	J.~Hershberger and S.~Suri.
	\newblock Matrix searching with the shortest-path metric.
	\newblock \emph{SIAM Journal on Computing}, 26\penalty0 (6):\penalty0
	1612--1634, 1997.
	
	\bibitem[Klein(1989)]{Klein}
	R.~Klein.
	\newblock \emph{Concrete and abstract {V}oronoi diagrams}.
	\newblock Springer-Verlag Berlin Heidelberg, 1989.
	
	\bibitem[Klein and Lingas(1994)]{klein1994}
	R.~Klein and A.~Lingas.
	\newblock Hamiltonian abstract {V}oronoi diagrams in linear time.
	\newblock In \emph{Prooceedings of the 5th International Symposium on
		Algorithms and Computation (ISAAC 1994)}, pages 11--19. Springer Berlin
	Heidelberg, 1994.
	
	\bibitem[Mitchell(2000)]{m-gspno-00}
	J.~S.~B. Mitchell.
	\newblock Geometric shortest paths and network optimization.
	\newblock In \emph{Handbook of Computational Geometry}, pages 633--701.
	Elsevier, 2000.
	
	\bibitem[Oh and Ahn(2017)]{Oh-2017}
	E.~Oh and H.-K. Ahn.
	\newblock {Voronoi diagrams for a moderate-sized point-set in a simple
		polygon}.
	\newblock In \emph{Proceedings of the 33rd International Symposium on
		Computational Geometry (SoCG 2017)}, volume~77, pages 52:1--52:15. Schloss
	Dagstuhl--Leibniz-Zentrum f\"ur Informatik, 2017.
	
	\bibitem[Oh et~al.(2016)Oh, Barba, and Ahn]{oba-fpgvdpbsp-16}
	E.~Oh, L.~Barba, and H.-K. Ahn.
	\newblock The farthest-point geodesic {V}oronoi diagram of points on the
	boundary of a simple polygon.
	\newblock In \emph{Proceedings of the 32nd International Symposium on
		Computational Geometry (SoCG 2016)}, volume~51, pages 56:1--56:15. Schloss
	Dagstuhl--Leibniz-Zentrum f\"ur Informatik, 2016.
	
	\bibitem[Papadopoulou(1999)]{kpairpath}
	E.~Papadopoulou.
	\newblock $k$-pairs non-crossing shortest paths in a simple polygon.
	\newblock \emph{International Journal of Computational Geometry and
		Applications}, 9\penalty0 (6):\penalty0 533--552, 1999.
	
	\bibitem[Pollack et~al.(1989)Pollack, Sharir, and Rote]{pollackComputingCenter}
	R.~Pollack, M.~Sharir, and G.~Rote.
	\newblock Computing the geodesic center of a simple polygon.
	\newblock \emph{Discrete \& Computational Geometry}, 4\penalty0 (6):\penalty0
	611--626, 1989.
	
	\bibitem[Suri(1989)]{suri1989computing}
	S.~Suri.
	\newblock Computing geodesic furthest neighbors in simple polygons.
	\newblock \emph{Journal of Computer and System Sciences}, 39\penalty0
	(2):\penalty0 220--235, 1989.
	
	\bibitem[T.(1986)]{toussaint}
	T.~G. T.
	\newblock An optimal algorithm for computing the relative convex hull of a set
	of points in a polygon.
	\newblock In \emph{Proceeding of EURASIP-86, Part 2}, pages 853--856, 1986.
	
\end{thebibliography}

\end{document}